\documentclass[11pt]{article}

\usepackage{geometry}
\usepackage{amssymb,amsmath,amsfonts,eurosym,graphicx,caption,color,setspace,sectsty,comment,footmisc,caption,pdflscape,subfigure,array,bm,listings,multirow}
\usepackage[table]{xcolor}
\usepackage[colorlinks,linkcolor=purple, urlcolor  = purple, anchorcolor=brown, citecolor=violet]{hyperref}
\usepackage{algorithm}
\usepackage{algorithmicx}
\usepackage{algpseudocode}
\usepackage{amsthm}
\usepackage{amsmath}
\usepackage[normalem]{ulem}
\usepackage{apacite}
\usepackage{mathrsfs}
\usepackage{booktabs}
\usepackage{adjustbox}
\usepackage{cleveref}
\usepackage{etoolbox} 
\usepackage{multirow} 
\usepackage{centernot}
\usepackage{cancel}

\usepackage{booktabs} 
\usepackage{array} 
\usepackage[normalem]{ulem}

\usepackage{rotating}
\usepackage{threeparttable}

\usepackage[title]{appendix}

\newcommand\iid{\stackrel{\rm i.i.d.}{\sim}}
\newcommand\p{\stackrel{p}{\rightarrow}}

\def\qed{\rule{2mm}{2mm}}
\def\independent{\perp \!\!\! \perp}

\definecolor{undercover}{gray}{0.92}
\newcommand{\uc}[1]{\cellcolor{undercover}#1}

\newtheoremstyle{myremark}
  {\topsep}   
  {\topsep}   
  {\normalfont}  
  {}          
  {\bfseries} 
  {.}         
  { }         
  {}          

\usepackage[pagewise,mathlines]{lineno}
\mathchardef\dash="2D

\newtheorem{theorem}{Theorem}[section]
\newtheorem{lemma}{Lemma}[section]
\newtheorem{definition}{Definition}

\newtheorem{proposition}{Proposition}[section]
\newtheorem{assumption}{Assumption}[section]
\usepackage{bbm} 

\Crefname{assumption}{Assumption}{Assumptions}

\theoremstyle{myremark}

\newtheorem{remark}{Remark}[section]

\AtEndEnvironment{remark}{~\qed}
\AtEndEnvironment{example}{~\qed}

\DeclareMathOperator*{\var}{Var}

\numberwithin{equation}{section}

\definecolor{skyblue}{rgb}{0.53, 0.81, 0.92}
\definecolor{steelblue}{rgb}{0.27, 0.51, 0.71}

\begin{document}
\lstset{
basicstyle=\ttfamily\small,
numbers=left,
keywordstyle= \color{ blue!70},commentstyle=\color{red!50!green!50!blue!50}
}

\title{ Nonparametric Empirical Bayes Confidence Intervals\footnote{ I am deeply grateful to Ivan Canay, Joel Horowitz, Federico Bugni, and Eric Auerbach for their guidance and encouragement throughout this project. I also thank Kirill Borusyak, Jiafeng Chen, Denis Chetverikov, Federico Crippa, Jiaying Gu, Michal Koles\'{a}r, Soonwoo Kwon, Carlos Lamarche, Amilcar Velez, and Kaspar W\"{u}thrich for helpful conversations and comments. ChatGPT provided assistance with writing and proofreading the manuscript. All remaining errors are my own.  }}
\author{Zhen Xie\\[6pt]
Department of Economics\\[6pt]
Northwestern University\\[6pt]
\url{zhenxie@u.northwestern.edu}\\[8pt]
}
\date{\today }
\maketitle


\begin{abstract}
Empirical Bayes methods can improve inference on unobservable individual effects by borrowing strength across units. This paper proposes nonparametric empirical Bayes confidence intervals (NP-EBCIs) for unobservable individual effects in a normal means model. The oracle intervals are constructed from posterior quantiles under a point-identified, fully nonparametric prior; feasible intervals replace these quantiles with nonparametric estimates. The NP-EBCIs are asymptotically exact in the sense that both their conditional and marginal coverage probabilities converge to the nominal level. The flexibility of this nonparametric construction has an unavoidable statistical cost. We demonstrate that posterior quantiles, unlike posterior means, inherit the severe ill-posedness of nonparametric deconvolution: the minimax optimal estimation rate is logarithmic. This logarithmic rate is minimax optimal for errors in the conditional coverage probability, and the resulting errors in the marginal coverage probability also vanish at the same logarithmic rate. Despite these slow asymptotic rates, simulations show that the NP-EBCIs remain close to nominal coverage when the prior is non-Gaussian, and deliver substantial length reductions relative to intervals that treat each unit in isolation.
\end{abstract}

\noindent\textbf{Keywords}: Empirical Bayes, Individual heterogeneity, Uncertainty quantification, Nonparametric methods, Ill-posed inverse problems \\[4pt]
\noindent\textbf{JEL classification codes}: C11, C14, C13

\thispagestyle{empty}

\newpage

\section{Introduction}

Characterizing individual heterogeneity has become increasingly prominent in applied economics. Much modern empirical work treats individual effects, rather than only their average, as the primary object of interest.\footnote{See \citeA{walters2024empirical,bonhomme2024estimating} for recent overviews. A growing body of empirical work studies unit-specific parameters, including neighborhood effects (\citeNP{chetty2018impactsI}), workplace heterogeneity (\citeNP{card2013workplace}), teacher value-added (\citeNP{chetty2014measuringI}), manager effects (\citeNP{fenizia2022managers}), geographic variation in health care (\citeNP{finkelstein2016sources}), bank-specific credit-supply effects (\citeNP{amiti2018much}), employer-specific discrimination (\citeNP{kline2022systemic,kline2024discrimination}), police-officer heterogeneity (\citeNP{goncalves2021few}), among many others. } This heterogeneity is often substantively important in its own right, as it describes how these economically relevant effects vary across units. Yet individual effects are unobserved and must be inferred from noisy data, often with heteroskedastic variances. Credible documentation of such heterogeneity therefore requires valid and informative uncertainty quantification for the individual effects themselves.

We consider the problem of constructing confidence intervals for unit-specific effects in a heteroskedastic normal means model. For each unit $i$, let $Y_i$ be a noisy estimate of the unobservable individual effect $\theta_i$ and satisfy $Y_i|\theta_i\sim N(\theta_i,\sigma_i^2)$ with known sampling variance $\sigma_i^2$. The naive $z$-interval $Y_i\pm z_{1-\alpha/2}\sigma_i$ has exact frequentist coverage for $\theta_i$, but it can be quite wide when $\sigma_i^2$ is large.\footnote{Here $z_{\alpha}$ is the $\alpha$-quantile of the standard normal distribution. } Motivated by this concern, \citeA{cox1975prediction} and \citeA{morris1983parametric} introduced parametric empirical Bayes confidence intervals (EBCIs) constructed from the posterior distribution of $\theta_i|Y_i$ under a common Gaussian prior for all $\theta_i$. Under this normal-normal model, the resulting parametric EBCIs borrow strength across
units and can be much shorter than the naive $z$-interval (\citeNP{morris1983parametric}). However, when the true prior is non-Gaussian, they can substantially undercover.

\citeA[hereafter AKP]{armstrong2022robust} are the first to propose EBCIs that are robust to the failure of the Gaussian prior assumption. Like Cox-Morris EBCIs, their intervals are centered at linear shrinkage estimators, but use worst-case critical values calibrated to guarantee average coverage over a class of priors characterized by certain moments of $\theta_i$. Because the full prior distribution is nonparametrically identified in the normal means model, this moment-based calibration leaves some available information unused. As a result, their robust EBCIs can be conservative when the true prior is far from the least favorable one in that moment class.

This paper proposes nonparametric empirical Bayes confidence intervals (NP-EBCIs) that fully utilize the entire prior distribution. Following \citeA{cox1975prediction,morris1983parametric}, we construct these intervals from the lower and upper quantiles of the posterior distribution of $\theta_i|Y_i$. Rather than imposing a Gaussian prior, we assume $\theta_i\sim G$ and allow the distribution $G$ to be fully nonparametric. We target the oracle posterior credible intervals under the prior $G$ and construct their feasible counterparts by estimating the posterior quantiles nonparametrically. This contrasts with \citeA{armstrong2022robust}, whose robust EBCIs are calibrated using only moment information about $\theta_i$.

The NP-EBCI construction makes posterior quantiles central to our analysis, and our first contribution is to characterize the minimax optimal rate for estimating them. We establish this rate by deriving a minimax lower bound for the estimation risk and constructing a one-step kernel-based estimator that achieves it. This optimal rate is inherently logarithmic, reflecting the fundamental limits of empirical Bayes posterior quantile estimation. This result highlights a sharp contrast with posterior mean estimation (\citeNP{jiang2009general,brown2009nonparametric}). In the normal means model, recovering the unknown prior $G$ is a severely ill-posed deconvolution problem with a logarithmic minimax rate (\citeNP{carroll1988optimal}).\footnote{ Throughout the paper, we use ``prior distribution'' and ``mixing distribution'' interchangeably for $G$. In empirical Bayes language, $G$ is the prior distribution for $\theta_i$, while the term ``mixing distribution'' reflects that the marginal law of $Y_i$ is obtained by mixing the normal distribution over $G$, see Equation (\ref{equ:convolution}). }  Although posterior mean and posterior quantiles are both nonlinear functionals of $G$, the posterior mean can circumvent the ill-posedness and be estimated at a nearly parametric rate (\citeNP{zhang1997empirical,zhang2009generalized}).\footnote{In the nonparametric empirical Bayes literature, faster rate results for the posterior mean are typically stated in terms of average mean squared errors for estimating $\theta_i$; see e.g. \citeA{jiang2009general,chen2026empirical}.   } We show that posterior quantiles cannot: because they are defined through a \textit{non-smooth} criterion function, their estimation inherits the severe ill-posedness of nonparametric deconvolution.

We next investigate the \textit{conditional} coverage properties of the feasible NP-EBCI, which relies on estimated posterior quantiles. By conditional coverage, we mean the probability that the reported interval contains \(\theta_i\) after \(Y_i\) has been observed. This conditional perspective has a long history in statistics, dating back at least to \citeA{cox1958some},\footnote{Related classic discussions of the conditional perspective include \citeA{buehler1959some,wallace1959conditional,birnbaum1962foundations,robinson1979conditional}; see also \citeA{rubin1984bayesianly,robins2000conditioning} for later discussions.} who argued that conditioning is needed so that the uncertainty statement reflects what can be learned from the realized data. Our second contribution is to characterize the minimax optimal rate at which the error in conditional coverage probability vanishes, and to show that the feasible NP-EBCI attains this rate. This optimal rate is again logarithmic, reflecting the fundamental limits of nonparametric empirical Bayes inference when targeting exact conditional coverage. Ultimately, the feasible NP-EBCI achieves asymptotically exact conditional coverage, but its coverage probability converges to the nominal level at this minimax optimal logarithmic rate.

As our third contribution, we show that the feasible NP-EBCI has asymptotically exact \textit{marginal} coverage, with coverage error that again vanishes at a logarithmic rate. By marginal coverage, we mean the ex-ante probability that the reported interval contains \(\theta_i\) under repeated sampling of both $\theta_i$ and the observed data. This notion is commonly referred to as ``empirical Bayes coverage'' (\citeNP{morris1983parametric}). Our result differs from the parametric EBCI literature, where the error in marginal coverage can vanish at a faster parametric rate under a correctly specified Gaussian prior.\footnote{See e.g. \citeA{datta2002asymptotic,chatterjee2008parametric,yoshimori2014second}.} However, their theoretical guarantees are vulnerable to misspecification when the true prior is non-Gaussian. Our result also differs from \citeA[~Theorem 4.1]{armstrong2022robust}, who show that the empirical Bayes coverage of their robust EBCI is asymptotically at least the nominal level, thus their actual coverage need not be exact even asymptotically. The feasible NP-EBCI, by contrast, achieves asymptotically exact empirical Bayes coverage, but with coverage error of logarithmic order.

This theoretical contrast motivates a comparison of the finite-sample performance of these EBCIs. In our simulations, all three EBCI procedures deliver average length reductions relative to the naive $z$-interval. AKP robust EBCIs with the second moment often exceed nominal coverage and yield the least length reductions. Cox-Morris parametric EBCIs are typically the shortest, but this efficiency can come at the cost of substantial undercoverage when the true prior is non-Gaussian, especially in low signal-to-noise ratio regimes. The feasible NP-EBCI, despite its severe ill-posedness, remains much closer to the nominal coverage level, even in small samples, while still delivering substantial length reductions.

\paragraph{Related literature.} This paper contributes to the nonparametric empirical Bayes literature on the normal means problem. The central idea in this literature is to treat many related decision problems jointly, learn prior information from related units, and then use it to construct Bayes posterior summaries or decision rules for each unit (\citeNP{robbins1956empirical,robbins1964empirical,efron2014two,efron2019bayes}, see also \citeNP{koenker2024empirical,walters2024empirical} and references therein). Much of this literature has focused on estimation of individual effects,\footnote{See \citeA{chen2026empirical} and references therein. For other settings beyond normal mean models, see e.g. \citeA{gu2017unobserved,liu2020forecasting,gilraine2020new,gaillac2024predicting,kwon2026optimal,cheng2025optimal}.  } while more recent work, including \citeA{jiang2019comment} and \citeA{koenker2020empirical}, has considered confidence intervals for individual effects. Both papers provide constructions and numerical evidence, but do not develop theoretical guarantees for their procedures. Our contribution is to establish a general theory for posterior quantile estimation and the coverage properties of posterior-quantile-based EBCIs. More broadly, because posterior quantiles are Bayes rules under asymmetric loss, these results are also relevant for other empirical Bayes decision problems that involve posterior quantiles (e.g. \citeNP{gu2023invidious}, Section 2.5 in \citeNP{walters2024empirical}).

This paper is broadly related to, but different from \citeA{ignatiadis2022confidence} and \citeA{ignatiadis2025empirical}. Specifically, \citeA{ignatiadis2022confidence} study confidence intervals for empirical Bayes estimands, such as the posterior mean, rather than confidence intervals for the individual effect $\theta_i$ itself. \citeA{ignatiadis2025empirical} study empirical partially Bayes multiple testing. As they note in Footnote 2, their $p$-values can be inverted to construct individual intervals for $\theta_i$. A key distinction is that they place a prior on the variances rather than individual effects, pooling information across units through the variance distribution. Consequently, in the homoskedastic model, their interval reduces to the naive $z$-interval (see their Example 14), while our procedure continues to exploit the common prior on the individual effects.

Our objective differs from that of the deconvolution literature. Deconvolution methods aim to recover the latent distribution of unobserved heterogeneity from noisy measurements (see e.g. \citeNP{horowitz2014ill,schennach2020mismeasured} for reviews). Our goal, by contrast, is individualized inference for $\theta_i$ itself. The role of deconvolution in our analysis is therefore technical rather than substantive: the posterior quantiles underlying our NP-EBCIs are functionals of the latent distribution, so the problem is built on the same Gaussian deconvolution structure.

Our results on minimax optimal rates for posterior quantiles connect to the literature on functional estimation in ill-posed inverse problems. For linear functionals of the mixing distribution $G$, the minimax rate can depend sharply on the smoothness of the weighting function that defines the functional (\citeNP{butucea2009adaptive,pensky2017minimax}). Our setting is different, because both posterior mean and posterior quantiles are \textit{nonlinear} functionals of the mixing distribution. But a similar insight emerges after local linearization: the two functionals differ sharply in their local dependence on the mixing distribution. The posterior mean involves a comparatively smooth weighting function and therefore can circumvent the severe ill-posedness. By contrast, the posterior quantile is locally governed by a discontinuous weighting function and therefore inherits the logarithmic difficulty of the underlying inverse problem.

\paragraph{Outline.} \Cref{section:ebci} presents the set-up, defines the oracle NP-EBCI, and contrasts it with the Cox-Morris parametric EBCI and AKP robust EBCIs. Assuming a homoskedastic normal mean model, \Cref{section:theoretical_result} derives the minimax optimal rate of estimating posterior quantiles, and establishes the conditional and marginal coverage properties of the feasible NP-EBCI. \Cref{section:practical_implementation} returns to the heteroskedastic setting and presents a feasible implementation of the NP-EBCI, including bandwidth selection. \Cref{section:simulation} uses Monte Carlo simulations to compare the finite-sample performance of NP-EBCI against existing methods. \Cref{section:conclusion} concludes. Proofs and extended discussions are collected in the Appendix.

\section{Empirical Bayes confidence intervals for individual effects}

\label{section:ebci}

\subsection{Set-up}

We consider the following normal means model:
\begin{equation}
Y_i \mid \theta_i \sim N(\theta_i,\sigma_i^2), \ \ \ \theta_i\iid G(\cdot), \ \ \ \text{for} \ i=1,...,n
\label{normal_mean}
\end{equation}
where $Y_i$ is a noisy estimate of the unobserved outcome of interest $\theta_i$ for each unit $i$. For example, $\theta_i$ may represent the value-added of teacher $i$, and $Y_i$ is an estimate constructed from the associated students' test scores. The unit-specific parameter $\theta_i$ is thought to be random, and drawn from the population distribution $G(\cdot)$. The variances $\sigma_i^2$ are known, but $G$ and ${\bm \theta}=\{\theta_1,...,\theta_n\}$ are unknown.

The model (\ref{normal_mean}) arises naturally from a hierarchical data structure. For instance, let $\tilde{Y}_{i,j}$ denote the test outcome of student $j=1,\ldots,J_{i}$ taught by teacher $i$, and suppose that these test scores are normally distributed with a teacher-specific mean and variance: $\tilde{Y}_{i,j}|\theta_i\sim N(\theta_i,\sigma_i^2)$ for $j=1,...,J_{i}$ where $\theta_i$ represents the value-added of teacher $i$. A natural estimator for $\theta_i$ is the average test score of students taught by teacher $i$: $Y_i=J_{i}^{-1}\sum_{j=1}^{J_i}\tilde{Y}_{ij}$. Then we have that $Y_i|\theta_i\sim N(\theta_i,\sigma_i^2/J_i)$ which falls into the model (\ref{normal_mean}).

The model (\ref{normal_mean}) fits into the classical measurement errors model, with $Y_i=\theta_i+\varepsilon_i$ and $\varepsilon_i\sim N(0,\sigma_i^2)$. The normal error assumption is restrictive, but standard in the empirical Bayes literature with cross-sectional data (\citeNP{efron2012large,koenker2024empirical,walters2024empirical}). It can be replaced by any other known error distribution $F_{\varepsilon_i}(\cdot)$, or even fully relaxed to an unknown distribution when panel data or other repeated measurements are available.\footnote{See, for example, \citeA{horowitz1996semiparametric} and \citeA{li1998nonparametric}, as well as the surveys by \citeA{schennach2020mismeasured} and \citeA[Section 7]{bonhomme2024estimating}, for further details. }

The mixing distribution $G(\cdot)$ is nonparametrically identified from the marginal law of $Y_i$. Assume that the mixing distribution $G(\cdot)$ has a continuous density $g(\cdot)$ with respect to Lebesgue measure. Given the known variance $\sigma_i^2$, the marginal density of $Y_i$ is given by
\begin{equation}
f_{Y_i}(y)=\int \frac{1}{\sigma_i}\phi\left(\frac{y-\theta}{\sigma_i}\right)\mathrm{d}G(\theta)
\label{equ:convolution}
\end{equation}
where $\phi(\cdot)$ is the standard normal density. For an integrable function $f:\mathbb{R}\rightarrow\mathbb{C}$, write $f^{\star}(t)=\int_{\mathbb{R}}\exp(itx)f(x)\mathrm{d}x$ for its Fourier transform. Let $f_{Y_i}^{\star}(t), g^{\star}(t)$ be the Fourier transforms of the densities $f_{Y_i}(y), g(\theta)$ respectively. The convolution theorem (Section 3.3.2, \citeNP{schennach2020mismeasured}) gives that $f_{Y_i}^{\star}(t)=g^{\star}(t)f_{\varepsilon_i}^{\star}(t)$ where $f_{\varepsilon_i}^{\star}(t)=\exp\left(-\frac{1}{2}\sigma_i^2t^2\right)$ is the Fourier transform of the normal error distribution. Since $f_{\varepsilon_i}^{\star}(t)\neq 0$ for every $t\in\mathbb{R}$, we recover $g^{\star}(t)=f_{Y_i}^{\star}(t)/f_{\varepsilon_i}^{\star}(t)$. Because characteristic functions uniquely determine probability laws, this establishes the identification of $G$.\footnote{This simplified exposition is intended to align with the theoretical analysis in \Cref{section:theoretical_result}. However, we note that the identification argument applies to any Borel probability measure $G$ (\citeNP[Proposition 8.50]{folland1999real}). }

\subsection{Empirical Bayes confidence intervals}
\label{subsection:nonparametric_ebci}

We aim to construct an individual interval $\mathrm{CI}_i$ for each unobservable individual effect $\theta_i$ with a pre-specified coverage probability. Following \citeA{cox1975prediction,morris1983parametric}, we take an empirical Bayesian approach: assume that $\theta_i\sim G$ for some prior distribution $G(\cdot)$ before any observation is taken. After observing $Y_i=y$, uncertainty about $\theta_i$ is summarized by the posterior distribution $\theta_i\mid Y_i$.

If the true prior $G$ were known, we could construct an \textit{oracle} Bayesian credible interval by taking the lower and upper $\alpha/2$ posterior quantiles.\footnote{Our interval is the ``equal-tailed credible set'' for $\theta_i$. One could alternatively construct a highest posterior density (HPD) credible set, defined as the set where the posterior density exceeds some threshold (\citeNP[Definition 5]{Berger1985}). However, a practical limitation of the HPD credible set is that it may be disjoint, making it difficult to report.    } Let $\mathcal{D}_i=(Y_i,\sigma_i)$, we denote this infeasible oracle interval by
\begin{equation}
\mathrm{CI}_i^{\mathrm{NP}*}=\Big{[} q_{G}(\alpha/2;\mathcal{D}_i)\ ,\ q_{G}(1-\alpha/2;\mathcal{D}_i) \Big{]}
\label{npci_infeasible}
\end{equation}
where the posterior quantile is defined as
\begin{equation}
q_{G}(\tau;\mathcal{D}_i)=\inf\Big{\{}  u: \operatorname{P}_G\left( \theta_i\leq  u \ | \ Y_i=y \right) \geq \tau \Big{\}} \ \ \ \text{for}  \ \tau\in(0,1).
\label{equ:posterior_quantile}
\end{equation}
The subscript $G$ is used to emphasize that these posterior quantiles are functionals of the underlying prior distribution.

Because $G$ is unknown in practice, we must distinguish between the oracle target (\ref{npci_infeasible}) and its feasible analogue. Our proposed nonparametric empirical Bayes confidence interval (NP-EBCI) replaces the oracle posterior quantiles with nonparametric estimates:
\begin{equation}
\text{CI}_i^{\text{NP}}=\Big{[}\widehat{q}_{G}(\alpha/2,\mathcal{D}_i)\ , \ \widehat{q}_{G}(1-\alpha/2;\mathcal{D}_i)\Big{]}.
\label{equ:ebci_feasible}
\end{equation}
The theoretical justification of the feasible procedure (\ref{equ:ebci_feasible}) relies on large-sample asymptotics as $n\rightarrow\infty$, in the spirit of \citeA{robbins1964empirical}. Before turning to that analysis in \Cref{section:theoretical_result}, we discuss several oracle properties of the interval (\ref{npci_infeasible}). First, the oracle interval (\ref{npci_infeasible}) admits a simple decision-theoretic characterization.

\begin{proposition}
The quantile-based posterior credible interval (\ref{npci_infeasible}) minimizes
\begin{equation*}
\mathcal{R}_i(L,U)=\operatorname{E}\left[ U(\mathcal{D}_i)-L(\mathcal{D}_i)+\frac{2}{\alpha}\left( L(\mathcal{D}_i)-\theta_i \right)_{+}+\frac{2}{\alpha}\left( \theta_i-U(\mathcal{D}_i) \right)_{+} \right]
\end{equation*}
over all measurable interval rules $(L,U)$ satisfying $L(\mathcal{D}_i)\leq U(\mathcal{D}_i)$ almost surely, where the expectation is taken over $(Y_i,\theta_i)$.
\label{prop:decision_theoretic}
\end{proposition}

\Cref{prop:decision_theoretic} gives a direct decision-theoretic justification for using posterior quantiles as interval endpoints. The interval (\ref{npci_infeasible}) is not only an equal-tailed posterior credible interval, but also the Bayes rule for an interval-valued decision problem that trades off interval length against linear penalties for undercoverage and overcoverage. In this sense, the lower and upper posterior quantiles arise as optimal endpoints under a particular loss function, in line with the decision-theoretic treatment of Bayesian interval estimation studied by \citeA{winkler1972decision,Berger1985}. The criterion in \Cref{prop:decision_theoretic} also coincides with the interval score used to evaluate prediction intervals (\citeNP{gneiting2007strictly}).

The oracle interval (\ref{npci_infeasible}) has the usual conditional coverage property of an empirical Bayes confidence interval (EBCI) for $\theta_i$ (\citeNP[Definition 3.2]{carlin2000bayes}):
\begin{equation}
\operatorname{P}\left(\theta_i\in\mathrm{CI}_i^{\mathrm{NP}*} \ | \ Y_i\right)=1-\alpha \ \ \ \text{for each} \ i=1,...,n
\label{equ:conditional_coverage}
\end{equation}
that summarizes the uncertainty about $\theta_i$ after the data are observed. It is interpreted as a data-specific measure of uncertainty: it evaluates inference procedure based on the realized sample rather than averaging over hypothetical, unobserved outcomes. This conditional perspective goes back at least to \citeA{cox1958some} who wrote that inference should reflect what can be learned from the data that we have, and it is central to Bayesian literature which treats the observed data as known and concerns the remaining uncertainty through the conditional distribution of unknowns given knowns (\citeNP{rubin1984bayesianly}, see also \citeNP[Section 1.6]{Berger1985}).

The oracle interval (\ref{npci_infeasible}) also satisfies the marginal coverage property, often called empirical Bayes coverage (\citeNP{morris1983parametric}, \citeNP[Definition 3.1]{carlin2000bayes}). To see this, take averages over the observed data distribution:
\begin{align*}
\operatorname{P}\left( \theta_i\in\mathrm{CI}_i^{\mathrm{NP}*} \right)&=\operatorname{E}_{\mathbf{Y}}\left[ \operatorname{P}\left(\theta_i\in\mathrm{CI}_i^{\mathrm{NP}*} \ | \ Y_i\right) \right]=1-\alpha \\[7pt]
&=\operatorname{E}_{\bm \theta}\left[\operatorname{P}\left(\theta_i\in\mathrm{CI}_i^{\mathrm{NP}*} \ | \ \bm \theta\right)  \right]=1-\alpha
\end{align*}
the probability in $\operatorname{P}\left( \theta_i\in\mathrm{CI}_i \right)$ is taken over both ${\bm \theta}=(\theta_1,...,\theta_n)$ and the data $\mathbf{Y}=(Y_1,...,Y_n)$.\footnote{This notion of marginal coverage also appears in the conformal prediction literature, see e.g. \citeA{angelopoulos2023gentle}. } Thus marginal coverage is an \emph{ex ante} statement: it characterizes the uncertainty prior to observing any data. It can also be interpreted as the population average of frequentist coverage $\operatorname{P}(\theta_i\in\mathrm{CI}_i^{\mathrm{NP}*}| \bm \theta)$ that considers the uncertainty from the data. This marginal coverage statement is weaker than the conditional coverage property (\ref{equ:conditional_coverage}) because it need only hold on average and therefore may be poorly calibrated for the sample data at hand. Even so, the notion of marginal coverage is standard in the empirical Bayes literature.

We now compare the oracle NP-EBCI with the existing alternatives in the literature. The natural starting point is the naive $z$-interval $Y_i\pm z_{1-\alpha/2}\sigma_i$, which treats each unit in isolation and has exact frequentist coverage for $\theta_i$ without imposing any prior structure. The limitation is that when $\sigma_i^2$ is large, the interval can be unnecessarily wide. This concern leads naturally to the empirical Bayes confidence interval (EBCI) literature, which aims to improve inference for each unit by combining the noisy observation $Y_i$ with prior information learned from the distribution of individual effects.

\citeA{cox1975prediction} and \citeA{morris1983parametric} initiated the parametric EBCI literature by constructing individual intervals for $\theta_i$ from the posterior distribution under a Gaussian prior. In the homoskedastic case $\sigma_i^2\equiv \sigma^2$, they assume $\theta_i\iid N(0,A)$. Under this normal-normal specification, the posterior distribution takes a simple form: $\theta_i|Y_i=y\sim N(\frac{A}{A+\sigma^2}y,\frac{A\sigma^2}{A+\sigma^2})$ which yields the oracle interval
\begin{equation}
\mathrm{CI}_i^{\mathrm{Cox}_{*}}=w_{\mathrm{EB}}Y_i\pm w_{\mathrm{EB}}^{1/2}\sigma\times z_{1-\alpha/2} \ \ \ \text{where} \ w_{\mathrm{EB}}=\frac{A}{A+\sigma^2}
\label{equ:cox_ebci}
\end{equation}
A feasible version of this interval replaces the unknown variance component $A$ with an estimator. This construction can produce substantially shorter intervals than the naive $z$-interval. This gain, however, comes from the parametric prior restriction.

Like \citeA{cox1975prediction} and \citeA{morris1983parametric}, we construct intervals from the posterior distribution of $\theta_i|Y_i$. The difference is that we make no parametric assumptions about the prior distribution. Instead we treat the prior $G$ as fully nonparametric and take the oracle posterior interval in (\ref{npci_infeasible}) as the target. The interval proposed in practice is the feasible analog in (\ref{equ:ebci_feasible}), obtained by replacing the oracle posterior quantiles with nonparametric estimates.

A different alternative is the robust EBCI of \citeA{armstrong2022robust}, which targets worst-case marginal coverage over a moment class rather than the posterior interval under a point-identified prior. In the homoskedastic case $\sigma_i^2\equiv \sigma^2$, their oracle interval takes the form
\begin{equation*}
\mathrm{CI}_i^{\mathrm{AKP}_{*}}=w_{\mathrm{EB}}Y_i\pm w_{\mathrm{EB}}\sigma \times \chi(\alpha) \ \ \ \text{where} \ w_{\mathrm{EB}}=\frac{A}{A+\sigma^2}
\end{equation*}
where $A=E[Y_i^2]-\sigma^2$. The critical value $\chi(\alpha)$ is chosen such that the optimal value of the following optimization problem is at least $1-\alpha$:
\begin{equation}
\inf_{G} \ \operatorname{P}_G\left[ \theta_i \in \mathrm{CI}_i^{\mathrm{AKP}}(\chi) \right] \ \ \ \text{s.t.} \ \operatorname{E}_G[\theta_i^2]=A.
\label{equ:optimization}
\end{equation}
By construction, the marginal coverage probability of their oracle robust EBCI is at least $1-\alpha$ for all prior distributions that satisfy the second moment constraint. However, the actual coverage probability may exceed the nominal level when the true prior is not least-favorable distribution that solves (\ref{equ:optimization}). The oracle NP-EBCI, by contrast, relies on nonparametric identification of the prior distribution and fully exploits the entire distribution of $\theta_i$. Consequently, the interval (\ref{npci_infeasible}) attains exact marginal coverage at the nominal level $1-\alpha$. Moreover, the AKP robust EBCI does not target the exact posterior conditional coverage property in (\ref{equ:conditional_coverage}) which conditions on the observed data. By contrast, the oracle NP-EBCI (\ref{npci_infeasible}) satisfies this property by construction.

We make two additional remarks to clarify how certain terminology in empirical Bayes inference are understood in our setting.

\begin{remark}
\citeA{morris1983parametric} assumed that the prior belongs to a known class $\mathscr{G}$, and evaluated the performance of procedures uniformly over all $G\in\mathscr{G}$. He also suggested placing a ``second-stage'' prior on $\mathscr{G}$ to account for uncertainty about the prior itself. In our setting, $G$ is point identified. The remaining uncertainty is sampling uncertainty, because $G$ must still be estimated from the data. One could, in principle, place a prior on $\mathscr{G}$ to account for this estimation uncertainty, see also \Cref{remark_adjustement}.
\end{remark}
\begin{remark}
\citeA{morris1983parametric} discussed the robust Bayes literature; see also \citeA[Chapter 4.7]{Berger1985} for a broader treatment that considers sensitivity analysis as the ``subjective'' prior varies over an $\varepsilon$-contamination class. In our setting, the prior is point identified, so we do not represent prior uncertainty through such a class. Robust Bayes methods would nevertheless be relevant if one wished to study sensitivity to prior misspecification or contamination, which we do not pursue in this paper.
\end{remark}

\section{Theoretical results under homoskedasticity}
\label{section:theoretical_result}

In this section, we analyze the theoretical properties of the feasible NP-EBCI. First we consider nonparametric estimation of the posterior quantile: we establish its optimal rates of convergence in \Cref{subsection:rate_quantile} and propose a rate-optimal estimator in \Cref{subsection:estimator_quantile}. Second, in \Cref{subsection:ebci}, we establish the optimal rates of the errors in conditional coverage probability and show that the feasible NP-EBCI based on plug-in estimate of posterior quantile achieves this rate. Finally, we examine its marginal coverage probability and connects it to the empirical Bayes literature.

\paragraph{Notation.} For \( 1 \leq p < \infty \), the space \( L^p(\mathbb{R}) \) is defined as \( L^p(\mathbb{R}) = \{ f : \mathbb{R} \to \mathbb{C} \mid \int_{\mathbb{R}} |f(x)|^p \, \mathrm{d}x < \infty \} \), with associated norm \( \|f\|_p = \left( \int_{\mathbb{R}} |f(x)|^p \, \mathrm{d}x \right)^{1/p} \). For \( p = \infty \), we define the norm as \( \|f\|_\infty = \operatorname{ess\,sup}_{x \in \mathbb{R}} |f(x)| \). For any function \( f : \mathbb{R} \to \mathbb{C} \) such that \( f \in L^1(\mathbb{R}) \), we define its Fourier transform by $f^{\star}(t) = \int_{\mathbb{R}} \exp(itx) f(x) \, \mathrm{d}x$ for \( t \in \mathbb{R} \).

Throughout \Cref{section:theoretical_result}, we assume that $\sigma_i^2\equiv \sigma^2$ for technical convenience, which implies that $(Y_i,\theta_i)\iid \mathcal{F}$. In \Cref{subsection:rate_quantile,subsection:estimator_quantile}, we consider the posterior quantile associated with a hypothetical random variable $(Y_{n+1},\theta_{n+1})\sim \mathcal{F}$, evaluated at $Y_{n+1}=y$, and estimate it using the data $\{Y_i\}_{i=1}^n$. In \Cref{subsection:ebci,subsection:marginal_coverage}, since our goal is to construct individual intervals $\mathrm{CI}_i$ for each $\theta_i$, the relevant posterior quantile corresponds to $(Y_i,\theta_i)$ and is estimated using the leave-one-out data $\{Y_j\}_{j\neq i}$. We focus on a fixed quantile level $\tau\in(0,1)$, tailored to the objective of the application.

\subsection{Fundamental limits of estimating posterior quantiles}
\label{subsection:rate_quantile}

Both posterior mean\footnote{The posterior mean, defined as $m(y)=\operatorname{E}_G[\theta|Y=y]$, is one of the most widely studied objects in empirical Bayes literature. It is the optimal predictor of the individual parameter $\theta_i$ under the mean-square error, conditional on the observed data $Y=y$.    } and posterior quantiles are nonlinear integral functionals of the unknown mixing distribution $G$. In the normal means model, the observed distribution of $Y$ is the convolution of $G$ with the normal error distribution, so recovering $G$ from the data is a deconvolution problem. With normal errors, this inverse problem is severely ill-posed: the optimal rate for estimating the mixing density $g$ is only logarithmic (\citeNP{carroll1988optimal}, see also \citeA{horowitz2014ill} for a review).\footnote{ The normal density is a canonical example of a ``super-smooth'' density, which leads to the worst-case rates of convergence in nonparametric deconvolution. When the error density is ``ordinarily smooth'', the optimal rate is polynomial (\citeNP{fan1991optimal}); see \citeA[page 511]{schennach2020mismeasured} for further discussion.    } In contrast to this severe ill-posedness, the posterior mean can be estimated at a nearly \textit{parametric} rate (\citeNP{zhang1997empirical,zhang2009generalized}). Since the posterior quantile is also an integral functional of the unknown mixing density, it is natural to ask whether such ``integration'' can likewise lead to faster rates of convergence.

We will show that the answer is negative. Unlike the posterior mean, the posterior quantile is defined through a non-smooth criterion function. This non-smoothness prevents the associated Fourier weighting term from offsetting the exponential decay induced by the Gaussian error density. Thus, posterior quantile estimation inherits the severe ill-posedness of Gaussian deconvolution. To show that this difficulty is intrinsic, rather than a deficiency of a particular estimator, we derive a minimax lower bound on the risk of posterior quantile estimation over a functional class. This minimax result characterizes the fundamental limits of the estimation problem and provides a benchmark for evaluating the quality of an estimation procedure.

We begin by rewriting the posterior quantile in a form that exposes its non-smoothness. Although defined in (\ref{equ:posterior_quantile}) as the solution to a moment condition weighted by the normal error density, it is more useful to reframe it as the minimizer of a convex optimization problem.

\begin{proposition}
The posterior quantile (\ref{equ:posterior_quantile}) can be characterized as the minimizer of an $\ell_1$ loss function:
\begin{equation}
q_{G}(\tau;y) = \underset{q}{\arg\min} \ \int \rho_{\tau}\left( \theta-q \right)\phi\left( \frac{y-\theta}{\sigma} \right) \mathrm{d}G(\theta)
\label{equ:obj_Wq}
\end{equation}
where $\rho_{\tau}(u)=u(\tau-\mathbf{1}\{ u\leq 0 \})$ is the check function. We denote the objective function by $W(q)$.
\end{proposition}

We transform the objective function (\ref{equ:obj_Wq}) into the frequency domain using Fourier basis functions. Define the weighting function $M(q,\theta;y)=\rho_{\tau}\left( \theta-q \right)\phi(\frac{y-\theta}{\sigma})$ that appears in the integral of the objective function (\ref{equ:obj_Wq}). By the Parseval-Plancherel theorem (Theorem 9.13 in \citeA{rudin1987}), the objective function $W(q)$ in \Cref{equ:obj_Wq} can be expressed as
\begin{equation}
W(q)=\int M(q,\theta;y)g(\theta)\mathrm{d}\theta=(2\pi)^{-1} \int \overline{M^{\star}(q,t;y)}g^{\star}(t) \mathrm{d}t =(2\pi)^{-1} \int  \cfrac{\overline{M^{\star}(q,t;y)}}{f_{\varepsilon}^{\star}(t)} f_{Y}^{\star}(t)\mathrm{d}t
\label{equ:obj_Wq2}
\end{equation}
where $\overline{M^{\star}(q,t;y)}=\int \exp(-it\theta)M(q,\theta;y)\mathrm{d}\theta$ is the conjugate Fourier transform of $M(q,\theta;y)$.

The Fourier representation (\ref{equ:obj_Wq2}) shows that the problem has two ingredients. First, Gaussian deconvolution is severely ill-posed because $f_{\varepsilon}^{\star}(t)$ decays exponentially as $|t|\rightarrow \infty$, so division by $f_{\varepsilon}^{\star}(t)$ amplifies high-frequency noise. Second, whether this difficulty can be mitigated depends on the decay of the Fourier weighting term. For smooth functionals, such as the posterior mean, that weighting term decays fast enough to downweigh the exponential decay at the tail, leading to a nearly parametric rate.\footnote{ A similar insight appears in the semiparametric models with measurement errors that achieve the parametric $\sqrt{n}$ rate, see \citeA[Page 512]{schennach2020mismeasured}
  for related discussion.  }  For posterior quantiles, however, the check function in (\ref{equ:obj_Wq}) is non-smooth at zero, so $\overline{M^{\star}(q,t;y)}$ decays too slowly to offset it. The posterior quantile therefore inherits the severe ill-posedness of Gaussian deconvolution.

To state the formal results, we introduce the following regularity assumptions.

\begin{assumption}
The mixing density $g\in L^1(\mathbb{R})\cap L^2(\mathbb{R})$.
\label{assumption_g_regularity}
\end{assumption}

\begin{assumption}
We assume that the mixing density $g$ belongs to the Sobolev class
\begin{equation*}
\mathcal{G}(s,L)=\left\{ g: \int \ |g^{\star}(t)|^2(t^2+1)^{s}\mathrm{d}t < L  \right\}
\end{equation*}
where $s\geq \frac{1}{2}$ is an integer smoothness parameter.
\label{assumption1}
\end{assumption}

We denote the posterior density evaluated at $\theta = u$ by
\begin{equation}
\pi(u|Y=y)=\frac{\phi\left(\frac{y-u}{\sigma}\right)g(u)}{\int \phi\left(\frac{y-\theta}{\sigma}\right)g(\theta) \mathrm{d}\theta}=\frac{\phi\left(\frac{y-u}{\sigma}\right)g(u)}{f_{Y}(y)}.
\label{equ:posterior_density}
\end{equation}

\begin{assumption}
The normalized posterior density evaluated at the target posterior quantile $q_0:=q_G(\tau;y)$ is strictly positive, $|\phi(\frac{y-q_0}{\sigma})g(q_0)|>0$.
\label{assumption_density}
\end{assumption}

\begin{assumption}
The Fourier transformation of the mixing density satisfies that  $g^{\star}(t)\in L^1(\mathbb{R})$. It implies the density $g$ is continuous and uniformly bounded as $\|g\|_{\infty}\leq (2\pi)^{-1}\|g^{\star}\|_1<\infty$.
\label{assumption2}
\end{assumption}

\begin{assumption}
The target parameter $q_G(\tau;y)$ lies in a compact parameter space.
\label{assumption4}
\end{assumption}

\Cref{assumption_g_regularity} is a regularity condition that facilitates the application of Fourier-based methods. \Cref{assumption1} is equivalent to requiring that the sum of the squared $L^2(\mathbb{R})$-norms of the first $s$ derivatives of the mixing density $g(\cdot)$ be bounded, i.e. $\sum_{l=0}^{s}\|g^{(l)}\|_2^2<L$, see e.g. \citeA[Appendix A.2]{meister2009deconvolution}. This Sobolev condition provides global control of the smoothness and is well suited for analyzing posterior quantile, a nonlinear integral functional of mixing density, while H\"{o}lder condition imposes smoothness only locally around some specific point.

\Cref{assumption_density} is the identification condition for the posterior quantile by requiring that the criterion function (\ref{equ:obj_Wq}) is non-flat at its solution, ensuring a unique minimizer. \Cref{assumption2} together with \Cref{assumption_g_regularity}, ensures that the inverse Fourier transform of $g^{\star}(t)$ is well-defined and yields a continuous mixing density $g(\cdot)$. Consequently, the posterior distribution has no atoms at the target quantile. Finally, \Cref{assumption4} ensures that the posterior quantile is well-defined and lies in a compact subset of $\mathbb{R}$.

\begin{theorem}[Minimax lower bounds]
Let \Cref{assumption_g_regularity,assumption1,assumption_density,assumption2,assumption4} hold. Then for any estimator $\hat{q}$ based on the i.i.d. data $\{Y_i\}_{i=1}^n$, we have that
\begin{align*}
\inf_{\widehat{q}} \ \sup_{g\in\mathcal{G}(s,L)} \ \operatorname{E}\Big{[} \left( \widehat{q}_G(\tau;y)-q_G(\tau;y) \right)^2\Big{]}&\geq \mathrm{const}\cdot (\log n)^{-(2s+1)/2} 
\end{align*}
for $n$ sufficiently large.
\label{rate_posterior_quantile}
\end{theorem}

\Cref{rate_posterior_quantile} establishes a minimax lower bound for estimating posterior quantiles uniformly over the Sobolev space $\mathcal{G}(s,L)$. In particular, it shows that no estimator can attain a faster rate of convergence than logarithmic uniformly over this function class. This lower bound reflects that posterior quantile estimation, due to the \textit{non-smoothness} of its criterion function, inherits the intrinsic difficulty of nonparametric deconvolution. The next subsection shows that this lower bound is attained by a feasible estimator. Hence the logarithmic rate is the minimax optimal rate for estimating posterior quantiles in the homoskedastic normal-means model.

\Cref{rate_posterior_quantile}, together with the nearly parametric rate result for posterior mean estimation (\citeNP{zhang1997empirical,zhang2009generalized}), highlights a broader point in functional estimation for ill-posed inverse problems: the minimax optimal rate is determined \textit{jointly} by the degree of ill-posedness of the inverse operator and the regularity of the target functional. Posterior quantiles are too irregular to circumvent the difficulty of the underlying severely ill-posed deconvolution problem, and their minimax rate therefore remains logarithmic.

\begin{remark}
It is worthwhile to mention that as with the monotonicity property of the posterior mean (\citeNP{houwelingen1983monotone,koenker2014convex}), the posterior quantile $q_{G}(\tau;y)$ is non-decreasing in $y$ for a broader class of error distributions $F_{\varepsilon}$ that includes the normal (\Cref{prop:quantile_mlr}). However, this monotonicity property does not improve the rate of convergence of the posterior quantile in the normal means model. The reason is that the monotonicity holds automatically for any valid densities, it does not impose additional restrictions on the risk when deriving the lower bound. Moreover, the posterior quantile $q_{G}(\tau;y)$ is non-decreasing in $\tau$,\footnote{This is known as the ``quantile non-crossing'' property, see \citeA[Chapter 2.5]{koenker2005quantile}} but this likewise does not improve the rate of convergence for the same reason above.
\end{remark}

\begin{remark}
\citeA{escanciano2023irregular} shows that in many structural models, quantile functionals of nonparametric unobserved heterogeneity have \textit{infinite} efficiency bounds. This occurs because the non-smoothness of the influence function violates the necessary condition in \citeA{van1991differentiable}. Our result in \Cref{rate_posterior_quantile} also illustrates the implications of this non-smoothness for statistical performance but focuses on characterizing minimax optimal rates.
\end{remark}

\subsection{A rate-optimal estimator of posterior quantiles}

\label{subsection:estimator_quantile}

In this subsection, we propose a one-step estimator for the posterior quantile that achieves the optimal rate established in \Cref{rate_posterior_quantile}.

Observe that the representation (\ref{equ:obj_Wq2}) expresses the posterior quantile directly as a function of the data, and naturally suggests an estimator for the posterior quantile, defined as the minimizer of the empirical counterpart of the objective function. In this representation, the only unknown component is $f_{Y}^{\star}(t)$. Therefore, to construct the empirical analog $W_n(q)$, we replace $f_Y^{\star}(t)$ with $\hat{f}^{\star}_Y(t)K^{\star}(h_nt)$, where $\hat{f}^{\star}_Y(t)=n^{-1}\sum_{i=1}^n \exp(it Y_i)$ is the empirical characteristic function, and $K^{\star}(t)$ denotes the Fourier transformation of a kernel function $K(\cdot)$ with bandwidth $h_n$. The kernel serves as a regularization device to ensure the existence of the integral when $f_Y^{\star}(t)$ is replaced by its empirical analog. The resulting kernel estimator of the posterior quantile is then
\begin{equation}
\widehat{q}_{G}(\tau;y)=\underset{q}{\arg\min} \ (2\pi)^{-1} \int  \cfrac{\overline{M^{\star}(q,t;y)}}{f_{\varepsilon}^{\star}(t)} \left(\widehat{f}^{\star}_Y(t)K^{\star}(h_n t) \right) \mathrm{d}t
\label{equ:estimator_quantile}
\end{equation}
where the sample objective function is denoted by $W_n(q)$.

Our proposed estimator (\ref{equ:estimator_quantile}) is a one-step procedure that bypasses the need to estimate the mixing distribution $G(\cdot)$ as an intermediate step. This estimation strategy is in the spirit of $f$-modeling for posterior mean estimation, which involves estimating only the marginal density $f_Y(y)$ and its derivative (\citeNP{robbins1956empirical,efron2011tweedie,efron2014two}). Another advantage of this direct estimation is that the regularization parameter can be tuned directly for the target functional, rather than for an intermediate estimate of the mixing distribution.

Our estimator for the posterior quantile is, by construction, similar to the classical deconvolution kernel density estimator (\citeNP{carroll1988optimal,stefanski1990deconvolving}). However, as discussed in \Cref{subsection:rate_quantile}, the posterior quantile  is a nonlinear integral functional of the deconvolution density. This leads to two key distinctions and challenges of the theoretical analysis. First, our estimator is defined as the solution to a Fourier integral, rather than an integral itself. Second, the integrand involves a weighting function $\overline{M^{\star}(q,t;y)}$ whose (non-)smoothness crucially affects the convergence rate of the estimator.

We then formally establish the theoretical properties of the proposed estimator $\hat{q}_{G}(\tau;y)$.

\begin{assumption}
\begin{itemize}
\item[(i)] The kernel satisfies that $K\in L^1(\mathbb{R})\cap L^2(\mathbb{R})$ with $K_{h_n}=h_n^{-1}K(\cdot/h_n)$. Its Fourier transformation $K^{\star}(t)$ is supported on $[-1,1]$.
\item[(ii)]Let $K:\mathbb{R}\rightarrow\mathbb{R}$ be a bounded function that satisfies $\int v^d K(v)\mathrm{d}v=0$ for $d=1,...,s$, $\int |v|^{s+1}|K(v)|\mathrm{d}v<\infty$.
\end{itemize}
\label{assumption3}
\end{assumption}

\Cref{assumption3} (i) requires that the Fourier transform of the kernel function be bounded, integrable and have compact support. \Cref{assumption3} (ii) requires $K$ to be higher-order kernel, that is a device designed to increase the rate of convergence of the asymptotic bias. Here we use a kernel of the order $s+1$, which is one degree higher than the smoothness of the mixing density.

\begin{theorem}[Consistency and rates of convergence]
Let \Cref{assumption_g_regularity,assumption1,assumption_density,assumption2,assumption3,assumption4} hold. Setting the bandwidth $h_n\asymp(\log n)^{-1/2}$, we have that
\begin{equation*}
\Big{|}\widehat{q}_G(\tau;y)-q_G(\tau;y)\Big{|}=O_p\left((\log n)^{-(2s+1)/4}\right).
\end{equation*}
\label{quantile_estimator}
\end{theorem}
\vspace{-2.5em}
\noindent\Cref{quantile_estimator} establishes that our proposed kernel estimator for the posterior quantile is consistent and achieves the optimal rate specified in \Cref{rate_posterior_quantile}. The slow rate of convergence reflects the difficulty of the estimation problem, and is not an indicator that our estimator is deficient.

The key step in the proof of \Cref{quantile_estimator} is to derive the following asymptotic expansion
\begin{equation}
r_n\left(\widehat{q}-q_0\right)=-r_n\left[ \phi\left(\frac{y-q_0}{\sigma}\right)g(q_0) \right]^{-1}\frac{\partial W_n(q_0)}{\partial q}+o_p(1)
\label{equ:a_expansion}
\end{equation}
where $r_n$ denotes the rate of convergence, $\hat{q}=\hat{q}_G(\tau;y)$ is the estimated posterior quantile, and $q_0=q_{G}(\tau;y)$ its population counterpart. This expansion is analogous to the classical Bahadur-Kiefer representation of the sample quantiles (\citeNP{bahadur1966note,kiefer1967bahadur}), and more generally to the linearization of M-estimators (see e.g. Section 3.2.4 in \citeNP{van1996weak}).

However, establishing the asymptotic expansion in \eqref{equ:a_expansion} is delicate when the functional of interest is possibly irregular, as is often the case in ill-posed inverse problems. One can invoke stochastic equicontinuity and empirical process arguments to control the nonlinear remainder term,\footnote{See Assumption 3.5 in \citeA{chen2015sieve} for a formal statement of the required high-level conditions.  } but this approach has a direct implication on the rates of the functional estimation.\footnote{The issue arises for regular functionals in severely ill-posed inverse problems. Consider a two-step procedure for estimating the posterior mean, where a deconvolution kernel density estimator is used as a plug-in. In this context, stochastic equicontinuity arguments typically bound the nonlinear remainder term by requiring uniform control over the first-stage estimate, i.e. $\|\hat{g}-g\|_{\infty}$, which results in a suboptimal convergence rate for the posterior mean itself.   } We instead exploit the convexity of the sample criterion in (\ref{equ:estimator_quantile}) and apply the convexity lemma of \citeA{pollard1991asymptotics,hjort2011asymptotics}, thereby avoiding a separate stochastic equicontinuity argument. The key step is to establish a uniform quadratic approximation of the localized criterion on compact sets; the convexity lemma then implies that the sample argmin is close to the argmin of this approximation. Thus, the proof relies on the local uniform control already established for the deconvolution kernel density estimator, rather than additional high-level empirical process conditions.

 Having established the asymptotic expansion (\ref{equ:a_expansion}), the rate of convergence is governed by the leading term $\partial W_n(q_0)/\partial q$. The same Fourier-space heuristic described after (\ref{equ:obj_Wq2}) applies here: once the criterion function is linearized, this term behaves like a functional with a non-smooth weighting function. Its Fourier coefficient therefore decays too slowly to offset the exponential decay of $f_{\varepsilon}^{\star}(t)$ as $|t|\rightarrow\infty$, resulting in a logarithmic convergence rate. For comparison, Appendix \ref{appendix:posterior_mean} studies the analogous kernel-based estimator for the posterior mean. There the leading term is associated with a smooth weighting function whose Fourier coefficient does offset this decay, so with a suitable bandwidth choice the estimator achieves a nearly parametric rate.

\begin{remark}[Nonsmooth criterion function]
Our setting presents an ill-posed inverse problem characterized by a non-smooth criterion function. Instead of tackling the non-smoothness componentwise, we linearize $W_n(q)$ around the true parameter. This builds on the insight that the non-smoothness may ``average out'' in the limit (\citeNP{huber1967behavior,pollard1985new}; \citeNP[Chapter 7]{newey1994large}; \citeNP{chen2003estimation}). However, unlike regular models, the (non)-smoothness of the criterion function crucially affects the rate of convergence of our estimator.\footnote{ The reason is that the second-order expansion of our estimator includes the difference between deconvolution density and its sample counterpart. To ensure this term converges, one must choose the bandwidth at a logarithm rate, which renders the leading term and thus the estimator converges at a logarithmic rate. In a regular model such as quantile regression, the leading term in its asymptotic expansion converges at a parameter rate. Although the second-order term involves the difference between regression error density and its sample counterpart, it does not affect the parameter rate of the leading term. As a result, quantile regression estimator achieves the parametric rate. } Moreover, since our estimator already attains its optimal rate, smoothing the criterion -- as in \citeA{horowitz1998bootstrap} -- cannot improve convergence rates without additional assumptions.
\end{remark}

\begin{remark}[Alternative two-step estimators]
Another widely used estimation strategy in the empirical Bayes literature is called ``$g$-modeling'' (\citeNP{koenker2019gmodeling}). It proceeds in two steps. First, one obtains an estimator of the mixing distribution $\hat{G}$ by maximizing generalized or sieve maximum likelihood (\citeNP{jiang2009general,koenker2014convex, efron2016empirical}). Second, the posterior quantile $q_{\hat{G}}(\tau;y)$ is computed based on plug-in $\hat{G}$. We do not expect these two-step procedures to improve the rate of convergence due to \Cref{rate_posterior_quantile}.
\label{remark:two_step}
\end{remark}

\subsection{Conditional coverage of the feasible NP-EBCI}
\label{subsection:ebci}

We now examine the conditional coverage of the feasible NP-EBCI. Because the target is the conditional coverage probability $\operatorname{P}_{G}(\theta_i\in\mathrm{CI}_i \mid Y_i=y)$, we evaluate $\mathrm{CI}_i$ at the realized observation $Y_i$, while estimating its endpoints from the remaining sample $\mathbf{Y}_{-i}:=\{Y_j\}_{j\neq i}$. This leave-one-out construction separates the observation being conditioned on from the data used to learn the prior.

Let $\hat{q}_G(\tau;Y_i)$ denote the leave-one-out analogue of the posterior quantile estimator in \Cref{equ:estimator_quantile}. We define the feasible NP-EBCI as
\begin{equation}
\text{CI}_i^{\text{NP}}=\Big{[}\hat{q}_{G}(\alpha/2;Y_i)\ , \ \hat{q}_{G}(1-\alpha/2;Y_i)\Big{]}.
\label{equ:ebci}
\end{equation}

To evaluate accuracy, we consider the error in conditional coverage probability, that is the difference between the actual and nominal coverage probabilities. The next theorem characterizes the best possible rate at which this error can vanish uniformly over the function class $\mathcal{G}(s,L)$.

\begin{theorem}
Suppose \Cref{assumption_g_regularity,assumption1,assumption_density,assumption2,assumption4} hold. Let $\mathscr{C}$ denote the class of all two-sided intervals constructed based on the data $\mathbf{Y}_{-i}=\{Y_j\}_{j\neq i}$, that is the collection of $[\ell(\mathbf{Y}_{-i}),U(\mathbf{Y}_{-i})]$ with $\ell(\mathbf{Y}_{-i})\leq U(\mathbf{Y}_{-i})$ almost surely, then we have that
\begin{equation*}
\inf_{{\mathrm{CI}_i}\in\mathscr{C}} \ \sup_{g\in\mathcal{G}(s,L)} \  \operatorname{E}_{\mathbf{Y}_{-i}}\Big{[}\left(\operatorname{P}_G\left(\theta_i \in \mathrm{CI}_i(\mathbf{Y}_{-i})\ | \ Y_i=y\right)-(1-\alpha) \right)^2 \Big{]}\geq \mathrm{const}\cdot (\log n)^{-(2s+1)/2}
\end{equation*}
for $n$ sufficiently large.
\label{ebci_coverage_lower_bound}
\end{theorem}
\Cref{ebci_coverage_lower_bound} provides a minimax lower bound on the errors in conditional coverage probabilities. The infimum is taken over all two-sided intervals constructed from the data, and the supremum is over all mixing densities satisfying \Cref{assumption1}. Thus, if conditional coverage is the target criterion, no leave-one-out two-sided interval can have a coverage error that vanishes faster than a logarithmic rate. In this sense, \Cref{ebci_coverage_lower_bound} reveals the fundamental limits on the accuracy of nonparametric empirical Bayes inference if the conditional coverage is the target criterion.

The proof of \Cref{ebci_coverage_lower_bound} relies on Fano's inequality (\citeNP{scarlett2019introductory}). Unlike the lower bound construction for mean square error, where alternative data distributions are separated by the induced target parameters, the separation here is based on the conditional coverage probabilities. We construct many alternative mixing densities around a baseline such that the induced observed laws of $\mathbf{Y}_{-i}$ are statistically indistinguishable, while the conditional coverage errors are bounded away from zero at the desired rate. Hence, if an interval had conditional coverage error uniformly smaller than this rate, it could decode which alternative generated the data. But Fano's inequality precludes this possibility because the alternatives are statistically close.

\begin{remark}
The minimax rates of errors in coverage probabilities of \textit{frequentist} confidence intervals have been considered in \citeA{hall1995uniform,calonico2022coverage}. Their notion of coverage is different from ours. In their setting, the parameter is fixed and coverage refers to the sampling probability that the interval contains this fixed parameter. In our setting, the object of inference is the random, unobservable individual effect $\theta_i$, and \Cref{ebci_coverage_lower_bound} concerns about the conditional coverage probability.
\end{remark}

We next show that the feasible NP-EBCI (\ref{equ:ebci}) achieves the rate in the lower bound established in \Cref{ebci_coverage_lower_bound}. Because we construct the feasible interval by replacing infeasible posterior quantiles with their estimators, our task is to translate the theoretical results of posterior quantile estimation (\Cref{quantile_estimator}) into a statement about conditional coverage. To do so, we impose one additional regularity condition on the posterior density at the relevant quantiles.

\begin{assumption}
We assume that the posterior density (\ref{equ:posterior_density}) evaluated at the posterior quantile is finite, i.e. $\pi(q(\tau;y)|Y=y)<\infty$.
\label{assumption_posterior_density}
\end{assumption}
\begin{theorem}
Let \Cref{assumption_g_regularity,assumption1,assumption_density,assumption2,assumption3,assumption4,assumption_posterior_density} hold. Setting the bandwidth $h_n\asymp(\log n)^{-1/2}$, the proposed individual interval (\ref{equ:ebci}) is a $100(1-\alpha)\%$ asymptotic conditional empirical Bayes confidence interval:
\begin{equation*}
\operatorname{P}\Big{(}\theta_i\in\mathrm{CI}_i^{\mathrm{NP}} \ \big{|} \ Y_i\Big{)}=1-\alpha+O_p\left((\log n)^{-(2s+1)/4}\right)
\end{equation*}
that holds uniformly over $g\in\mathcal{G}(s,L)$.
\label{ebci_condition_coverage}
\end{theorem}

\Cref{ebci_condition_coverage} shows that the conditional coverage probability of the feasible NP-EBCI (\ref{equ:ebci}) asymptotically fully adapts to the nominal $1-\alpha$, but differs at a logarithmic rate. The logarithmic conditional coverage error comes from that the endpoints themselves can only be estimated at a logarithmic rate. Nevertheless, \Cref{ebci_coverage_lower_bound,ebci_condition_coverage} together establish the optimal rate at which the worst-case conditional coverage probability converges to the nominal level, and the feasible NP-EBCI (\ref{equ:ebci}) attains this rate.

\begin{remark}
The feasible NP-EBCI (\ref{equ:ebci}) replaces the infeasible endpoints with their empirical analogs. A limitation is that it does not take into account the sampling uncertainty from estimating posterior quantile. \citeA{laird1987empirical,carlin1990approaches} propose adjustments to incorporate this source of uncertainty in parametric EBCIs within the normal-normal model. Their methods can be adapted to our setting where the prior is fully nonparametric and may improve the finite-sample performance of coverage probabilities. However, such adjustments cannot improve the rate of errors in conditional coverage probabilities due to \Cref{ebci_coverage_lower_bound}.
\label{remark_adjustement}
\end{remark}

\subsection{Marginal coverage of the feasible NP-EBCI}
\label{subsection:marginal_coverage}

We turn to the analysis of the marginal coverage probability of the feasible NP-EBCI (\ref{equ:ebci}). As this coverage statement is defined with respect to the joint distribution of both $\theta$ and $\{Y_i\}_{i=1}^n$, our analysis requires additional regularity conditions that hold uniformly for $y\in\mathcal{Y}\subset\mathbb{R}$.

\begin{assumption}
The normalized posterior density evaluated at $q(\tau;y)$, i.e. $\phi(y-q(\tau;y)/\sigma)g(q(\tau;y))$ is bounded below uniformly over $y\in\mathcal{Y}$. That is, there exists $\underline{\pi}>0$ such that $\inf_{y\in\mathcal{Y}}|\phi(y-q(\tau;y)/\sigma)g(q(\tau;y))|\geq \underline{\pi}$.
\label{assumption_uniform_bound_posterior_density}
\end{assumption}

\begin{assumption}
The expected posterior density evaluated at posterior quantile $q(\tau;y)$ is finite, that is $\int  \pi\left(q(\tau;y)|Y=y\right)\mathrm{d}F_{Y}(y)<\infty.$
\label{assumption_expected_density}
\end{assumption}

\Cref{assumption_uniform_bound_posterior_density} is a standard condition used to establish $L_2$ or uniform convergence in the quantile regression literature (see, e.g., Assumption 3 in \citeA{horowitz2005nonparametric}, Condition D1 in \citeA{belloni2011}). \Cref{assumption_expected_density} is weaker than the uniform version of \Cref{assumption_posterior_density}, requiring only that the expectation be finite.

\begin{theorem}
Let \Cref{assumption_g_regularity,assumption1,assumption2,assumption3,assumption4,assumption_uniform_bound_posterior_density,assumption_expected_density} hold. Setting the bandwidth $h_n\asymp(\log n)^{-1/2}$, the marginal coverage probability of individual confidence interval (\ref{equ:ebci}) satisfies that
\begin{equation*}
\operatorname{P}\Big{(} \theta_i\in\mathrm{CI}_i^{\mathrm{NP}}  \Big{)}=1-\alpha+O\left((\log n)^{-(2s+1)/4}\right)
\end{equation*}
that holds uniformly over $g\in\mathcal{G}(s,L)$, and the probability measure $\operatorname{P}$ is with respect to the joint distribution of ${\bm \theta}=(\theta_1,...,\theta_n)$ and the data $\mathbf{Y}=(Y_1,...,Y_n)$.
\label{ebci_coverage}
\end{theorem}

For the oracle NP-EBCI, marginal coverage follows immediately from conditional coverage by iterated expectations. \Cref{ebci_coverage}, however, concerns the feasible NP-EBCI whose endpoints are estimated. To obtain the marginal coverage result for this feasible interval, we need to control the errors in conditional coverage uniformly over $y\in\mathcal{Y}$ so they can be averaged over the distribution of $Y_i$. The averaging does not improve the convergence rate, since the dominant error still comes from estimating posterior quantiles, which suffers from the severe ill-posedness. Hence the errors in the coverage probability of the feasible NP-EBCI (\ref{equ:ebci}) vanish at the same logarithmic rate under both marginal and conditional coverage.

\Cref{ebci_coverage} gives the accuracy assessment about the feasible NP-EBCI under a fully nonparametric prior. The resulting slow logarithmic rate again arises from the intrinsic difficulty of nonparametric estimation in ill-posed inverse problems. It is useful to compare this result with the two alternatives discussed above: Cox-Morris parametric EBCIs and robust EBCIs of \citeA{armstrong2022robust}.

We first compare \Cref{ebci_coverage} with the marginal coverage result for Cox-Morris parametric EBCIs. Under a correctly specified Gaussian prior, Cox-Morris EBCIs can achieve a parametric rate for errors in marginal coverage probabilities; see e.g. Theorem 1 in \citeA{yoshimori2014second}. This advantage, however, relies on correct specification of the normal-normal model. When the true prior is non-Gaussian, especially in low signal-to-noise ratio regimes, Cox-Morris EBCIs can substantially undercover. The feasible NP-EBCI avoids this source of parametric misspecification by leaving the prior fully nonparametric. The price is that its marginal coverage error vanishes only at a logarithmic rate.

We then compare our result with AKP robust EBCIs. The oracle NP-EBCI targets exact marginal coverage under the true, point-identified prior. AKP instead calibrate their robust EBCI to guarantee marginal coverage over a moment class of priors, and their Theorem 4.1 implies that, under exchangeability, the feasible interval satisfies  $\lim\inf_{n\rightarrow\infty} \operatorname{P}(\theta_i\in \mathrm{CI}_i^{\mathrm{AKP}})\geq 1-\alpha $ (\citeNP{armstrong2022robust}). Because this calibration relies only on finite-dimensional moments, rather than on the full prior distribution, the statistical problem is simpler than posterior quantile estimation. Furthermore, their Theorem 4.1 requires only the convergence of the relevant moments, while our \Cref{ebci_coverage} imposes a stronger continuity condition on the prior distribution. The gain of AKP is robustness over the moment class; the cost is conservatism, since their EBCIs need not be asymptotically exact. The NP-EBCI instead targets exact marginal coverage under the true point-identified prior, but pays the logarithmic price of the severely ill-posed inverse problem.

Taken together, \Cref{ebci_coverage} clarifies the trade-offs behind the feasible NP-EBCI. Relative to the Cox-Morris parametric EBCI, it sacrifices the faster coverage error rates achievable under a correctly specified Gaussian prior in exchange for robustness to prior misspecification. Relative to the AKP robust EBCI, it gives up worst-case moment-class robustness in exchange for asymptotically exact marginal coverage under a point-identified prior. \Cref{section:simulation} examines how these trade-offs play out in finite samples.

\section{Practical Implementation under heteroskedasticity }
\label{section:practical_implementation}

This section returns to the baseline heteroskedastic normal mean model in (\ref{normal_mean}) and describes how to implement the proposed NP-EBCI when units have known but unequal variances $\sigma_i^2$. \Cref{section:theoretical_result} imposed the homoskedastic restriction $\sigma_i^2\equiv \sigma^2$ to focus on the optimal rate results. In empirical applications, however, heteroskedasticity is the natural default. Accordingly, in this section, we outline the practical implementation of our approach and discuss bandwidth selection.

\subsection{Baseline implementation  }
\label{subsection:baseline_implementation}

We now describe the baseline implementation of the proposed NP-EBCI under heteroskedastic Gaussian noise. Let $\mathcal{D}_i=(Y_i,\sigma_i)$. The oracle posterior quantile can be written as
\begin{equation*}
q_G(\tau;\mathcal{D}_i)=\underset{q}{\arg\min} \ \int \overline{M^{\star}(q,t;\mathcal{D}_i)}g^{\star}(t)\mathrm{d}t
\end{equation*}
where $\overline{M^{\star}(q,t;\mathcal{D}_i)}=\int \exp(-it\theta)\rho_{\tau}(\theta-q)\phi\left( (Y_i-\theta)/\sigma_i\right)\mathrm{d}\theta.$ The Fourier transform of the mixing density $g^{\star}(t)$ can be identified from the following population identity
\begin{equation}
\frac{1}{n-1}\sum_{j\neq i} \operatorname{E}\left[ \exp(it Y_j) \right]=g^{\star}(t)\left[\frac{1}{n-1}\sum_{j\neq i} \exp\left(-\frac{1}{2}\sigma_j^2t^2\right)\right].
\label{equ:average_characteristic}
\end{equation}
Here, the leave-one-out form is used for estimation rather than identification: it allows us to estimate the posterior quantile from the other units, while $\mathcal{D}_i$ enters only through the evaluation step.

The representation (\ref{equ:average_characteristic}) leads to a natural plug-in estimator. Define
\begin{equation*}
\hat{f}_Y^{\star}(t)=\frac{1}{n-1}\sum_{j\neq i} \exp(itY_j), \quad \bar{f}^{\star}_{\varepsilon}(t)=\frac{1}{n-1}\sum_{j\neq i} \exp\left(-\frac{1}{2}\sigma_j^2t^2\right).
\end{equation*}
For a given bandwidth $h_n$, we estimate the posterior quantile by
\begin{equation*}
\widehat{q}_G(\tau;\mathcal{D}_i)=\underset{q}{\arg\min} \ \int \ \cfrac{\overline{M^{\star}(q,t;\mathcal{D}_i)}}{\bar{f}^{\star}_{\varepsilon}(t)}\left(\hat{f}_Y^{\star}(t)K^{\star}(h_nt)\right)\ \mathrm{d}t
\end{equation*}
As in the homoskedastic case, the kernel $K^{\star}(h_nt)$ regularizes the Fourier integral. In practice, we use the compactly supported flat-top spectral kernel (\citeNP{politis1999multivariate}) but incorporate a cosine taper (\citeNP{harris1978use}) to smooth the transition from the flat-top region to zero and avoid a hard spectral cutoff. We  discuss the choice of the bandwidth parameter $h_n$ in \Cref{subsection:bandwidth_choice}.

Using the estimated posterior quantiles, we construct the NP-EBCI
\begin{equation}
\text{CI}_i=\Big{[}\widehat{q}_{G}(\alpha/2;\mathcal{D}_i)\ , \ \widehat{q}_{G}(1-\alpha/2;\mathcal{D}_i)\Big{]}, \quad i=1,...,n.
\label{equ:individual_CI_i}
\end{equation}

So far, we have presented the baseline implementation for heteroskedastic variances treated as fixed and known. The same construction also covers the case when $\sigma_i^2$ is random across units and independent with $\theta_i$. To see this, under $\theta_i\independent \sigma_i$, we have $\operatorname{E}[\exp(it Y_i)|\sigma_i]=g^{\star}(t)\exp\left(-\frac{1}{2}\sigma_i^2t^2\right)$, so averaging over the realized units gives the same population identity in (\ref{equ:average_characteristic}). In this sense, the baseline implementation can also be viewed as conditioning on observed heteroskedastic variances. Appendix \ref{subsection:precision_dependene} discusses the extension to the case in which $\theta_i$ and $\sigma_i^2$ may be correlated.

\subsection{Empirical choice of the bandwidth parameter}
\label{subsection:bandwidth_choice}

We choose a common bandwidth $h_n $ to minimize average interval length subject to an average coverage requirement.\footnote{One could, in principle, choose a separate bandwidth for each $\mathrm{CI}_i$ by minimizing the error of conditional coverage probability of that interval. This selection method is natural but computationally impractical because it requires solving $n$ distinct bandwidth selection problems.} Since this criterion depends on the unknown mixing distribution $G$, we implement it by a plug-in rule that replaces $G$ with its nonparametric maximum likelihood estimator (NPMLE) $\hat{G}(\cdot)$ (\citeNP{soloff2025multivariate}). In this respect, our bandwidth selection method is close in spirit to plug-in or rule-of-thumb bandwidth choice in nonparametric kernel density estimation (see e.g. \citeNP[Section 2.7]{li2007nonparametric}). However, unlike these methods relying on preliminary pilot values, our procedure uses the plug-in NPMLE $\hat{G}$, which itself is a consistent estimator of the mixing distribution $G$ (\citeNP{kiefer1956consistency,pfanzagl1988consistency}) and requires \textit{no} additional smoothing tuning parameter (\citeNP{polyanskiy2020self}).

It is important to distinguish the role of the NPMLE in our implementation from the two-step plug-in approach discussed in \Cref{remark:two_step}. Our feasible NP-EBCI is based on the one-step kernel estimator for posterior quantiles. The NPMLE is only used for bandwidth selection. We do not suggest estimating posterior quantiles with the NPMLE plug-in. In simulations not reported here, the NP-EBCI based on these plug-in estimates can substantially undercover, which aligns with the simulation evidence in \citeA{koenker2020empirical}. One plausible explanation is that the NPMLE has discrete support (\citeNP{lindsay1983geometry}), so plugging it directly into the posterior quantile can yield an overly concentrated posterior distribution and hence intervals that are too short.

The bandwidth selection procedure is detailed as follows:

\noindent\textbf{Step 1.} Partition the units $\{1,\ldots,n\}$ into $V$ disjoint folds $\mathcal{I}_1,\dots,\mathcal{I}_V$, and let $\mathcal{I}_v^{c}=\{1,\ldots,n\}\setminus \mathcal I_v$. For each fold $v=1,...,V$, we compute the leave-$v$-out nonparametric maximum likelihood estimator for the mixing distribution $\hat{G}^{(-v)}(\cdot)$.

\noindent\textbf{Step 2.} Let $\mathcal{H}_n$ be a grid of candidate bandwidths. For each $h\in\mathcal{H}_n$, each fold $v=1,\dots,V$ and each held-out unit $i\in\mathcal{I}_v$, we estimate the lower and upper endpoints $\hat{q}_h^{(-v)}(\alpha/2;\mathcal{D}_i)$ and $\hat{q}_h^{(-v)}(1-\alpha/2;\mathcal{D}_i)$ on the fold $\mathcal{I}_v^c$, but evaluate at the held-out unit $\mathcal{D}_i$.

\noindent\textbf{Step 3.} For each $h\in\mathcal{H}_n$, we compute the average conditional coverage and average length of intervals under the out-of-fold prior estimate $\hat{G}^{(-v)}(\cdot)$:
\begin{align*}
\widehat{C}(h)&=\frac{1}{n}\sum_{v=1}^V\sum_{i\in\mathcal{I}_v} \operatorname{P}_{\hat{G}^{-v}}\left( \theta_i \in \left[ \hat{q}_h^{(-v)}(\alpha/2;\mathcal{D}_i)\ , \ \hat{q}_h^{(-v)}(1-\alpha/2;\mathcal{D}_i) \right] \  \Big{|} \  \mathcal{D}_i \right) \\
\widehat{L}(h)&=\frac{1}{n}\sum_{v=1}^V\sum_{i\in\mathcal{I}_v}  \left[ \hat{q}_h^{(-v)}(1-\alpha/2;\mathcal{D}_i)-\hat{q}_h^{(-v)}(\alpha/2;\mathcal{D}_i) \right]
\end{align*}

\noindent\textbf{Step 4.} We set the bandwidth that yields the shortest average-length interval among those that satisfy the nominal coverage constraint:
\[
\hat h
=
\underset{h\in\mathcal H_n}{\arg\min}\ \widehat L(h)
\qquad\text{s.t.}\qquad
\widehat C(h)\ge 1-\alpha.
\]
Appendix \ref{section:computation_npebci} provides further strategies that can substantially reduce the computational cost of this bandwidth selection procedure.

\section{Monte Carlo Simulations}
\label{section:simulation}

Both NP-EBCI and the AKP robust EBCI move beyond the parametric Gaussian prior, but use the information about the prior distribution in different ways. NP-EBCI targets the oracle posterior interval under a point-identified nonparametric prior and thus exploits the entire distribution, but suffers from the severe ill-posedness. The AKP robust EBCI, by contrast, calibrate intervals to guarantee worst-case coverage uniform over a moment class of priors, so the resulting intervals may be conservative when the true prior is not least favorable. This section uses Monte Carlo simulations to examine how these differences appear in finite samples. We compare the feasible NP-EBCI with AKP (using the second moment) and the Cox-Morris parametric EBCI, and summarize performance by average marginal coverage and average length reduction relative to the naive $z$-interval.

\subsection{Monte Carlo Design}

We generate simulated data from the heteroskedastic normal means model (\ref{normal_mean}) with known variances:
\begin{equation*}
Y_i|{\bm \theta} \sim N(\theta_i,\sigma_i^2), \ \ \ \theta_i\iid G(\cdot), \ \ \ \text{for} \ i=1,...,n.
\end{equation*}
We consider sample sizes $n\in\{100,200,500,1000\}$ and target signal-to-noise ratios
\begin{equation*}
\mathrm{SNR}=\cfrac{\var(\theta_i)}{\frac{1}{n}\sum_{i=1}^n \sigma_i^2} \ \in\{0.1,0.2,0.4,1\}.
\end{equation*}
For each design cell $(G,n,\mathrm{SNR})$, we first draw $\sigma_i^2 \iid \mathrm{Lognormal}(0,1)$, $i=1,...,n$, and rescale it to attain the target SNR, and then keep the resulting variance vector fixed across all Monte Carlo replications. We then generate 500 independent replications by redrawing only $\{Y_i,\theta_i\}_{i=1}^n$.

We consider eight designs for the mixing distribution $G$:

\noindent\textbf{Design 1.} Gaussian: $\theta_i\sim N(0,1)$

\noindent\textbf{Design 2.} Laplace: $\theta_i\sim \text{Laplace}(0,1/\sqrt{2})$

\noindent\textbf{Design 3.} Student-$t$: $\theta_i\sim \sqrt{0.2}t_{2.5}$

\noindent\textbf{Design 4.} Bimodal Gaussian mixture: $\theta_i\sim 0.5N(-5/\sqrt{29},4/29)+0.5N(5/\sqrt{29},4/29)$

\noindent\textbf{Design 5.} Spike-and-slab: $\theta_i\sim 0.9 N(0,0.05^2)+0.1N(0,3.159^2)$

\noindent\textbf{Design 6.} 3-point distribution that places masses $(0.25,0.5,0.25)$ on $(-\sqrt{2},0,\sqrt{2})$

\noindent\textbf{Design 7.} The least favorable distribution for AKP robust EBCI

\noindent\textbf{Design 8.} The least favorable distribution for Cox-Morris parametric EBCI

All designs for the mixing distribution $G$ are normalized to have unit variance, so differences across designs are driven by the shape of $G$. Design 1 is the Gaussian benchmark under which the Cox-Morris normal-normal model is correctly specified. Designs 2 and 3 introduce heavier tails and hence represent progressively larger departures from the Gaussian benchmark. Designs 1 and 3 have supersmooth mixing densities in the Fourier sense, so they are relatively favorable cases for nonparametric deconvolution (\citeNP{lacour2006rates}), and thus for feasible NP-EBCI. Design 4 is included to examine how prior multimodality affects the posterior quantiles that form NP-EBCI. Design 5 is a nearly sparse prior with most mass concentrated near zero and a small fraction of large effects.

Designs 6-8 are discrete priors, thus violate the absolute continuity condition on $G$ (\Cref{assumption_g_regularity,assumption1,assumption2}), which underlies our theoretical analysis of NP-EBCI. Design 7 is the least favorable prior for the AKP robust EBCI (Design 5, \citeNP{armstrong2022robust}). Among priors with the same second moment, it attains the smallest marginal coverage and therefore comes closest to the worst-case coverage bound. Design 8 is the least favorable prior for Cox-Morris parametric EBCI (Design 6, \citeNP{armstrong2022robust}), under which the interval attains its lowest marginal coverage.

For each EBCI method, we focus on the average marginal coverage and the average length reduction relative to the naive $z$-interval, both averaged across units $i=1,\dots,n$ and across all Monte Carlo replications. The nominal coverage level is $1-\alpha=95\%$ throughout.

\subsection{Results and Discussion}

\begin{figure}[h]
\begin{center}
\includegraphics[scale=0.65]{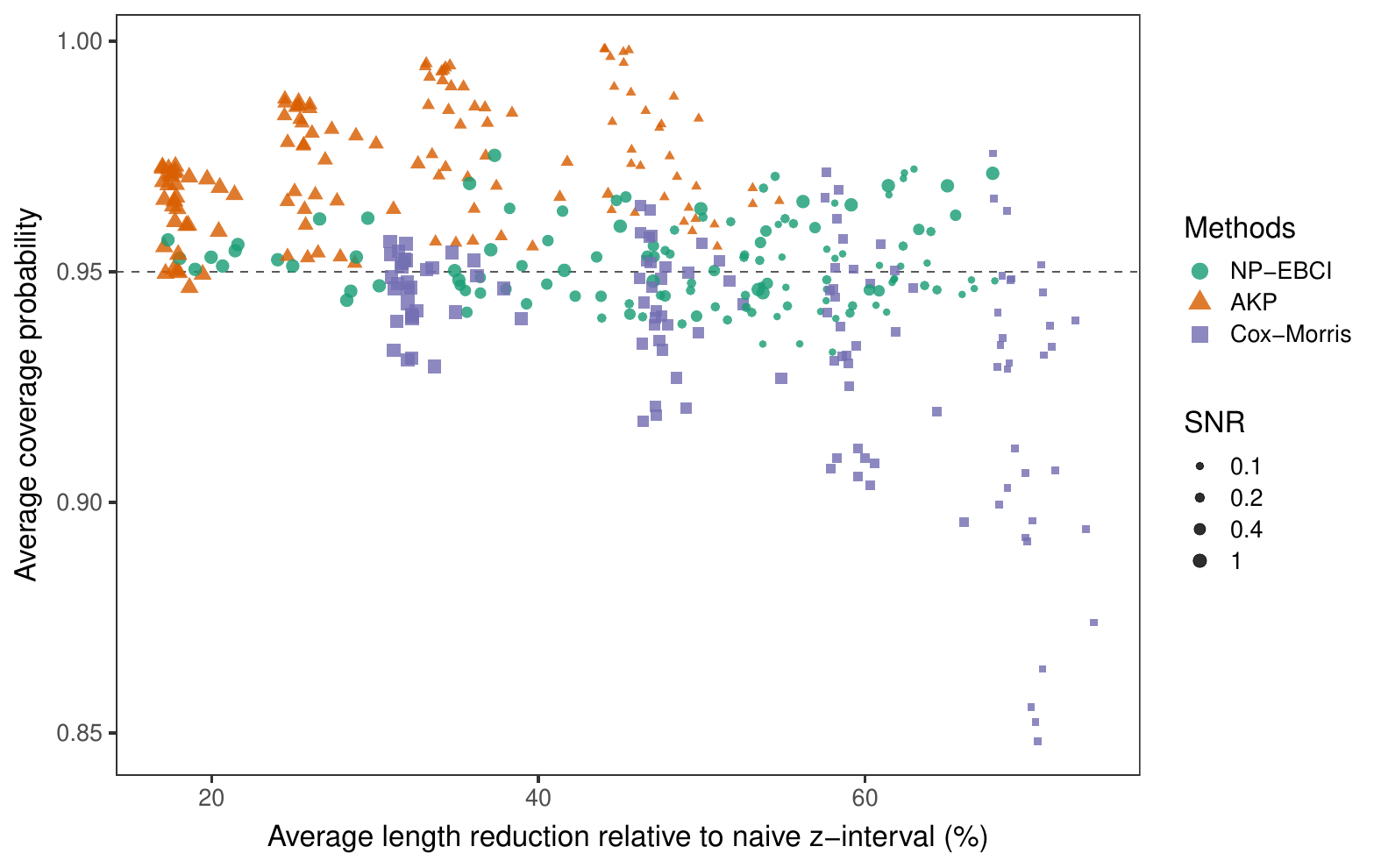}
\end{center}
\caption{\label{fig:coverage_length}\emph{Average coverage probability and length reduction across simulation designs.}}
\vspace{1em}
\begin{minipage}{\textwidth}
\footnotesize
\noindent\textit{Note.} Each point corresponds to one DGP $\times$ sample size $\times$ SNR cell. The dashed horizontal line is the nominal coverage level $1-\alpha=95\%$.
\end{minipage}
\end{figure}

\Cref{fig:coverage_length} summarizes the full simulation grid; each point corresponds to one $\text{DGP}\times \text{sample size}\times \text{SNR}$ cell.\footnote{Monte Carlo standard errors are small: across the three EBCI procedures, they are at most 0.17 percentage points for average coverage and 0.45 percentage points for average length reduction.} First, all three EBCI procedures deliver substantial average length reductions relative to the naive $z$-interval, reflecting the classical empirical Bayes gain from borrowing strength across related units (\citeNP{morris1983parametric}). Second, \Cref{fig:coverage_length} reveals a clear ordering among the three EBCI procedures in terms of average coverage and length. AKP (second moment) is conservative: many of its points lie above the nominal $95\%$ coverage line, but its length reductions are comparatively modest. Cox-Morris parametric EBCI often delivers the largest length reductions, but many of these gains come with substantial undercoverage. NP-EBCI achieves large length reductions while keeping empirical marginal coverage much closer to the nominal level.

\Cref{fig:coverage_length} further shows that the average coverage and length of EBCIs depend strongly on the signal-to-noise ratio (SNR). Lower-SNR cells tend to exhibit larger average length reductions for all three EBCIs, which is consistent with the fact that empirical Bayes pooling is most valuable when noise is large relative to signal (\citeNP{morris1983parametric}). These same cells also generate the sharpest differences in coverage across methods. At SNR=$0.1$, AKP tends to lie above the nominal $95\%$ line, Cox-Morris falls below it in many cells, and NP-EBCI remains much closer to the nominal $95\%$ line. This pattern motivates a closer examination of the low-SNR regime in \Cref{tab:main_hard_regime_cov_lenred}.

\begin{table}[h]
\centering
\caption{Average coverage probabilities and average length reductions for SNR $=0.1$}
\label{tab:main_hard_regime_cov_lenred}
\begin{threeparttable}
\footnotesize
\setlength{\tabcolsep}{6pt}
\renewcommand{\arraystretch}{1.15}
\begin{tabular}{lcccccc}
\toprule
& \multicolumn{3}{c}{Average coverage} & \multicolumn{3}{c}{Average length reduction (\%)} \\
\cmidrule(lr){2-4}\cmidrule(lr){5-7}
DGP & NP-EBCI & AKP & Cox--Morris & NP-EBCI & AKP & Cox--Morris \\
\midrule
\multicolumn{7}{l}{\textit{Panel A. } $n=100$, SNR $=0.1$} \\[2pt]
Gaussian      & 0.947 & 0.964 & \uc{0.848} & 53.7 & 49.2 & 70.6 \\[2pt]
Laplace       & 0.940 & 0.961 & \uc{0.864} & 54.6 & 49.6 & 70.9 \\[2pt]
$t_{2.5}$     & 0.953 & 0.965 & \uc{0.874} & 58.1 & 54.7 & 74.0 \\[2pt]
Bimodal       & 0.960 & 0.971 & \uc{0.856} & 54.7 & 48.5 & 70.2 \\[2pt]
Spike--slab   & 0.954 & 0.965 & 0.939 & 58.6 & 53.1 & 72.9 \\[2pt]
3-point       & 0.953 & 0.961 & \uc{0.852} & 54.9 & 48.9 & 70.4 \\[2pt]
Least favorable for AKP    & 0.947 & 0.960 & 0.934 & 55.1 & 50.8 & 71.4 \\[2pt]
Least favorable for Cox-Morris    & 0.934 & 0.955 & \uc{0.907} & 53.7 & 51.0 & 71.6 \\[2pt]
\midrule
\multicolumn{7}{l}{\textit{Panel B. } $n=1000$, SNR $=0.1$} \\[2pt]
Gaussian      & 0.949 & 0.997 & 0.941 & 61.8 & 44.4 & 68.1 \\[2pt]
Laplace       & 0.941 & 0.990 & 0.934 & 61.3 & 44.6 & 68.3 \\[2pt]
$t_{2.5}$     & 0.946 & 0.988 & 0.952 & 66.7 & 48.3 & 70.8 \\[2pt]
Bimodal       & 0.972 & 0.998 & 0.976 & 63.0 & 44.0 & 67.8 \\[2pt]
Spike--slab   & 0.948 & 0.973 & 0.949 & 67.9 & 45.7 & 69.0 \\[2pt]
3-point       & 0.970 & 0.998 & 0.966 & 62.3 & 44.1 & 67.9 \\[2pt]
Least favorable for AKP    & 0.952 & 0.963 & \uc{0.929} & 63.8 & 44.5 & 68.1 \\[2pt]
Least favorable for Cox-Morris    & 0.941 & 0.983 & \uc{0.900} & 57.3 & 44.5 & 68.2 \\
\bottomrule
\end{tabular}
\vspace{0.4em}
\begin{tablenotes}[flushleft]
\footnotesize
\item \textit{Note.} Nominal average confidence level is $1-\alpha=95\%$. Average CI length reduction is measured relative to the naive $z$-interval. For each DGP, average coverage and average length reduction refer to averages across units $i=1,...,n$ and across 500 Monte Carlo repetitions. The cells shaded with gray colors highlight the undercoverage.
\end{tablenotes}
\end{threeparttable}
\end{table}

\Cref{tab:main_hard_regime_cov_lenred} focuses on the regime $\text{SNR}=0.1$. Panel A reports the results for small samples $n=100$. NP-EBCI maintains average coverage near the nominal $95\%$ level across all designs (0.934 to 0.960), while reducing the average length by $53\%$ to $58.6\%$ relative to the naive $z$-interval. Although Cox-Morris EBCIs achieve larger average length reductions ($70.2\%$ to $74.0\%$), these gains come with substantial undercoverage in most cells, even in the Gaussian design, coverage falls to 0.848. One possible explanation is that when the prior variance $A$ is small relative to $\sigma_i^2$, the linear shrinkage factor $w_i=A/(A+\sigma_i^2)$ is also correspondingly small. Consequently, if the estimate $\hat{A}$ is biased downward in small samples, both the estimated posterior means and variances shrink aggressively toward zero, producing artificially narrow intervals (see e.g. \citeA{li2010adjusted}).

Panel B of \Cref{tab:main_hard_regime_cov_lenred} reports the results for large samples $n=1000$. NP-EBCI continues to maintain average coverage close to the nominal $95\%$ level (0.941 to 0.972), while reducing average length by $57.3\%$ to $67.9\%$. By contrast, AKP (second moment) substantially overcovers in nearly every design (0.973 to 0.998) except its own least-favorable design, and yields modest length reductions ($44.0\%$-$48.3\%$). While Cox-Morris EBCI produces the shortest intervals (reductions of $67.8\%$ to $70.8\%$) and improves coverage compared to Panel A, it still undercovers markedly in some specific least-favorable designs. Furthermore, as \citeA{armstrong2022robust} emphasize, the Cox-Morris EBCI's empirical Bayes coverage can drop to 0.74 for the nominal $95\%$ level in the low-SNR limit.

\begin{figure}[h]
\begin{center}
\includegraphics[scale=0.43]{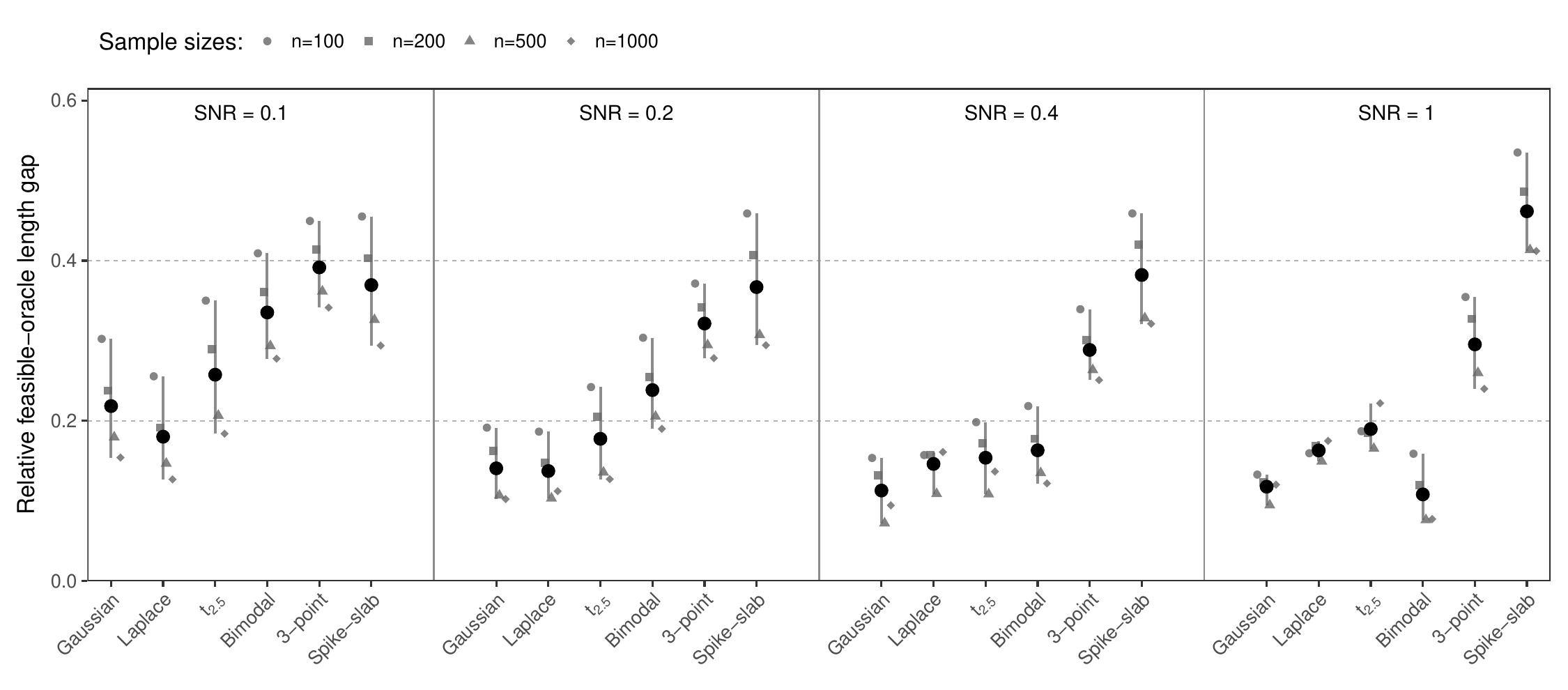}
\end{center}
\caption{\label{fig:oracle_length_gap}\emph{Relative feasible-oracle length gap of NP-EBCI across prior designs and SNRs.}}
\vspace{1em}
\begin{minipage}{\textwidth}
\footnotesize
\noindent\textit{Note.} The $y$-axis ``oracle length gap'' is defined in (\ref{equ:def_oracle_length_gap}). Gray markers correspond to $n=100,200,500,1000$. The black circle is the average across these sample sizes, and the vertical segment gives the range across them.
\end{minipage}
\end{figure}

The strong finite-sample performance of feasible NP-EBCI in average marginal coverage may seem at odds with the logarithmic ill-posedness established in \Cref{ebci_coverage}. It is useful to notice that the bandwidth choice in \Cref{subsection:bandwidth_choice} is designed to minimize average interval length subject to an average coverage requirement, rather than to minimize the mean squared errors of estimating posterior quantiles. Such coverage-oriented regularization may keep average marginal coverage close to the nominal level, with the finite-sample cost of ill-posedness showing up primarily as excess length relative to oracle NP-EBCI rather than as severe undercoverage. For this reason, we compare the average length of feasible NP-EBCI with that of its oracle counterpart in \Cref{fig:oracle_length_gap}.

In \Cref{fig:oracle_length_gap}, we define ``relative feasible-oracle length gap'' as
\begin{equation}
\left[\overline{\text{Length}}\left(\mathrm{CI}_i^{\mathrm{NP}}\right)-\overline{\text{Length}}\left(\mathrm{CI}_i^{\mathrm{NP}^{*}} \right)\right]\Big{/}\ \overline{\text{Length}}\left(\mathrm{CI}_i^{\mathrm{NP}}\right)
\label{equ:def_oracle_length_gap}
\end{equation}
where $\overline{\text{Length}}(\mathrm{CI}_i^{\mathrm{NP}})$ and $\overline{\text{Length}}(\mathrm{CI}_i^{\mathrm{NP}^{*}} )$ denote, respectively, the average lengths of the feasible NP-EBCI in \Cref{equ:individual_CI_i} and its oracle counterpart in \Cref{npci_infeasible}.


\Cref{fig:oracle_length_gap} first shows that the relative feasible-oracle length gap of NP-EBCI changes only marginally with sample size $n$. This aligns with the logarithmic ill-posedness of posterior quantile estimation being driven primarily by the structure of the inverse problem rather than by the sample size. Notably, this feasible-oracle gap varies significantly more across prior designs than across SNRs. In every SNR panel, the gap is small for the smooth priors but largest for the discrete and spike-and-slab priors. Although the gap varies slightly with SNR within a given design, the ordering across designs is quite stable. This pattern suggests that the relevant ill-posedness is governed mainly by the shape of $G$: smooth priors are comparatively favorable, while priors with sharp local structure, such as point masses or narrow spikes, make posterior quantiles harder to estimate.

\section{Conclusion}
\label{section:conclusion}

This paper proposes empirical Bayes confidence intervals based on posterior quantiles for unobservable individual effects under a fully nonparametric prior, and establishes their theoretical properties. On the one hand, the feasible NP-EBCI achieves asymptotically exact conditional and marginal coverage. On the other hand, this flexibility comes at an unavoidable statistical cost: because posterior quantiles inherit the severe ill-posedness of nonparametric deconvolution, the relevant errors in quantile estimation and in coverage probabilities vanish at logarithmic rates. Thus the accuracy of nonparametric empirical Bayes inference is fundamentally limited by the difficulty of underlying inverse problem. Despite these slow asymptotic rates, the simulations show that the feasible NP-EBCI remains much closer to nominal coverage, even in small samples, is less conservative than the AKP robust EBCI, and delivers substantial length reductions relative to the naive $z$-interval.

While this paper focuses on empirical Bayes confidence intervals, an important direction for future work is to study other empirical Bayes decision problems that involve posterior quantiles. As \citeA{walters2024empirical} emphasizes, posterior quantiles can be more appropriate when the objective is to identify units in the tails of the distribution while controlling the risk of costly selection errors. This concern arises naturally in empirical Bayes compound decision problems such as ranking and selection; see e.g. \citeA{gu2023invidious}. An interesting question for future research is whether our minimax results can be extended to characterize the statistical difficulty of such tail-oriented empirical Bayes decision problems.

A useful feature of the NP-EBCI construction is that it relies on point identification of the distribution of individual effect $G$, rather than on the normal means structure alone. This suggests another direction for future work: extending posterior-quantile-based empirical Bayes inference to a broader class of latent-variable models in which $G$ remains point identified (\citeNP{schennach2020mismeasured}). Such extensions would help clarify which parts of our theory are specific to Gaussian deconvolution and which reflect a broader difficulty of posterior-quantile-based empirical Bayes inference.

\bibliography{empirical_bayes_ref}

\newpage
\begin{appendices}
\section{Proofs of main results}

The additional notation will be used throughout the Appendix: the symbol $*$ denotes the convolution of two functions, defined by \((f * g)(x) = \int_{\mathbb{R}} f(x - y)g(y)\,\mathrm{d}y\), provided that \( f, g \in L^1(\mathbb{R}) \). For sequences \( a_n \) and \( b_n \), we write \( a_n \lesssim b_n \) if there exists a constant \( C > 0 \) such that \( a_n \leq C b_n \) for all sufficiently large \( n \), and \( a_n \gtrsim b_n \) if \( a_n \geq C^{-1} b_n \) for all sufficiently large \( n \). We write \( a_n \asymp b_n \) when both \( a_n \lesssim b_n \) and \( a_n \gtrsim b_n \) hold, meaning that \( a_n \) and \( b_n \) are of the same order up to constants.

\subsection{Analytical properties of the weighting function for posterior quantile }

The posterior quantile $q_0:=q(\tau;y)$ is the solution to the moment condition
\begin{equation*}
\int\left( \tau-\mathbf{1}\{\theta\leq q_0\}\right)\phi\left(\frac{y-\theta}{\sigma}\right)\mathrm{d}G(\theta)=0.
\end{equation*}
In this subsection, we characterize the decay behavior of the Fourier coefficients of the weighting function $\psi(q_0,\theta;y)=\left( \tau-\mathbf{1}\{\theta\leq q_0\}\right)\phi\left(\frac{y-\theta}{\sigma}\right)$:
\begin{equation}
\overline{\psi^{\star}(q_0,t;y)}:=\int \exp(-it\theta)\left( \tau-\mathbf{1}\{\theta\leq q_0\}\right)\phi\left(\frac{y-\theta}{\sigma}\right)\mathrm{d}\theta
\label{equ:weighting_function}
\end{equation}
Its decay behavior plays a crucial role in our analysis of convergence rates of the posterior quantile.


\begin{proposition}
The Fourier coefficients of the weighting function $\overline{\psi^{\star}(q_0(y),t;y)}$ has the following asymptotic expansion as $|t|\rightarrow\infty$.
\begin{equation*}
\overline{\psi^{\star}(q_0(y),t;y)}=\sigma \tau \exp(-ity)\exp\left(-\frac{1}{2}t^2\sigma^2\right)+\frac{\exp(-itq_0(y))}{it}\phi\left(\frac{y-q_0(y)}{\sigma}\right)+o(t^{-1})
\end{equation*}
where the leading term is of order $t^{-1}$, and the remainder term $o(t^{-1})$ holds uniformly over $y \in \mathbb{R}$.
\label{propo:fourier_decay}
\end{proposition}
\vspace{-2em}
\begin{proof}
We carry out infinite integration by parts for $\overline{\psi^{\star}(q_0,t;y)}$ since $\phi(\cdot)$ is infinitely smooth:
\begin{align*}
\overline{\psi^{\star}(q_0,t;y)}&=\sigma\tau \exp(-ity)\exp\left(-\frac{1}{2}t^2\sigma^2\right) +\exp(-itq_0(y))\sum_{k=1}^{\infty}\frac{(-1)^{k-1}}{(it)^{k}\sigma^{k-1}}H_{k-1}\left(\frac{y-q_0(y)}{\sigma}\right)\phi\left(\frac{y-q_0(y)}{\sigma}\right).
\end{align*}
where $H_k(\cdot)$ is the $k$-th Hermite polynomial, e.g. $H_0(z)=1, H_1(z)=z, H_2(z)=z^2-1$.

An important property is that $\sup_{z\in\mathbb{R}} |H_k(z)|\phi(z)<C_k$ for any $k\in\mathbb{N}_{+}$. Thus for each $k\geq 2$:
\begin{align*}
\left|\exp(-itq_0(y))\frac{(-1)^{k-1}}{(it)^{k}\sigma^{k-1}}H_{k-1}\left(\frac{y-q_0(y)}{\sigma}\right)\phi\left(\frac{y-q_0(y)}{\sigma}\right)\right|\leq \frac{C_{k-1}}{(it)^k\sigma^{k-1}}=O(t^{-k})
\end{align*}
that holds for all $y\in\mathbb{R}$. Then we obtain the desired results.
\end{proof}
\begin{remark}
The leading $t^{-1}$ decay rate arises from the discontinuity at $\theta = q_0$, which governs the asymptotic behavior of the Fourier coefficients $\overline{\psi^{\star}(q_0,t;y)}$.
\end{remark}
\begin{remark}
The asymptotic expansion in \Cref{propo:fourier_decay} reflects the key features of Fourier coefficients for functions with discontinuities. The second term decays at the slower rate of $t^{-1}$, reflecting the impact of the discontinuity introduced by the indicator function $\mathbf{1}\{\theta\leq q_0\}$. This slower decay rate reflects the Gibbs phenomenon which arises in the Fourier analysis of functions with jump discontinuities.
\end{remark}

\subsection{Proof of \Cref{rate_posterior_quantile}}
\label{subsection:proof_rate_posterior_quantile}

We follow the classical two-point method of \citeA{lecam1973convergence}, see also \citeA{fan1991optimal,butucea2009adaptive,pensky2017minimax} for related lower bound arguments in deconvolution problems. The proof has three parts. First, we construct two mixing densities $g_0$ and $g_1=g_0+H_T$ in $\mathcal{G}(s,L)$, where $H_T$ is a local perturbation concentrated at frequencies of order $T$. Second, we show that the corresponding posterior quantiles are separated by order $T^{-s-1/2}$. Third, we show that the induced data distributions $((F_Y)_0, (F_Y)_1)$ are statistically indistinguishable in the sense that their $\chi^2$ distance is of order $1/n$. \Cref{prop:A2_two_point} then yields the desired minimax lower bound.

The non-standard step in our proof is to translate the perturbation in the mixing densities into a separation of their posterior quantiles. This step handles the nonlinearity of the functional by relying on the local behavior of the objective function around its minimizer and is where the non-smoothness enters.

\begin{proposition}
Assume two densities $g_0, g_1\in\mathcal{G}(s,L)$. Then, with respect to an arbitrary estimator $\widehat{q}$ of $q_{G}(\tau;y)$ based on i.i.d. data $\{Y_i\}_{i=1}^n$ having the density $f_{Y}=g*f_{\varepsilon}$ that satisfies that
\begin{equation*}
\sup_{g\in\mathcal{G}(s,L)} \ \operatorname{E}\Big{[} \widehat{q}_G(\tau;y)-q_G(\tau;y) \Big{]}^2\geq \mathrm{const} \cdot \Big{|} q_{G_1}(\tau;y)-q_{G_0}(\tau;y) \Big{|}^2
\end{equation*}
for sufficiently large $n$, if we have that
\begin{equation}
\chi^2\left( (f_{Y})_1,(f_{Y})_0 \right)=\int \cfrac{\left( (f_{Y})_1(y)-(f_{Y})_0(y) \right)^2}{(f_{Y})_0(y)}\ \mathrm{d}y =O(1/n)
\label{equ:divergence}
\end{equation}
where $(f_{Y})_1=g_1*f_{\varepsilon}, (f_{Y})_0=g_0*f_{\varepsilon}$.
\label{prop:A2_two_point}
\end{proposition}
\begin{proof}[Proof of \Cref{prop:A2_two_point}]
The proof follows similar arguments to those used in Proposition 2.12 of \citeA{meister2009deconvolution} to derive the lower bound for the deconvolution density $g(\cdot)$.
\end{proof}

We now consider two candidate deconvolution densities
\begin{equation}
g_0(\theta)=\frac{\lambda}{\pi(\lambda^2+\theta^2)} \ \ \ \text{and} \ \ \
g_1(\theta)=g_0(\theta)+H_T(\theta)
\label{equ:density_pair}
\end{equation}
where $g_0(\theta)$ is a fixed baseline density and $H_T(\theta)$ is a local perturbation that will be specified later.

The baseline density $g_0$ is Cauchy with Fourier transformation $g_0^{\star}(t)=\exp(-\lambda|t|)$. Moreover
\begin{equation*}
\int |g_0^{\star}(t)|^2(1+t^2)^s\mathrm{d}t=\int \exp(-2\lambda|t|)(1+t^2)^s\mathrm{d}t\rightarrow 0
\end{equation*}
as $\lambda\rightarrow\infty$. Therefore, we may choose $\lambda$ large enough so that $\int  |g_0^{\star}(t)|^2(1+t^2)^s\mathrm{d}t\leq L/4$ so that $g_0\in\mathcal{G}(s,L)$. Let $q_0=q_{G_0}(\tau;y)$. Since $g_0$ is continuous and strictly positive, the posterior quantile $q_0$ is well defined and unique.

To construct a localized least-favorable perturbation, we use the standard notation $C_c^{\infty}(E)$ for the space of infinitely differentiable functions on $\mathbb{R}$ whose support is compact and contained in $E$; see \citeA[Chapter 9.1, p.~282]{folland1999real}. Choose a nonnegative function $k^{\star}\in C_c^{\infty}((1,2))$ that is not identically zero, and define
\begin{equation*}
\xi(u):=-i\,\text{sgn}(u)\,k^{\star}(|u|),
\qquad
h(x):=\frac1{2\pi}\int e^{-iux}\xi(u)\,du.
\end{equation*}
Since \(\operatorname{supp}(k^\star)\subset(1,2)\), there exists some $\varepsilon>0$ such that $k^{\star}(u)=0$ for all $u\notin[1+\varepsilon,2-\varepsilon]$. Since $k^{\star}(|u|)=0$ for $|u|<1+\varepsilon$, we have that $\xi(u)=0$ for $|u|<1+\varepsilon$. So near $u=0$, $\xi$ is just the zero function. Away from 0, both $\mathrm{sgn}(u)$ and $k^{\star}(|u|)$ are smooth. Therefore \(\xi\in C_c^\infty(\mathbb R)\). In particular, $\operatorname{supp}(\xi)\subset [-2,-1]\cup[1,2]$.

Moreover, since \(k^\star\) is real-valued, it follows that $\xi(-u)=\overline{\xi(u)}$ for all $u\in\mathbb R,$ and therefore the inverse Fourier transform \(h\) is real-valued. Recall that the Schwartz class (\citeNP[~p.237]{folland1999real}) consists of smooth functions\footnote{Here $C^{\infty}(E)$ means the smooth functions on $E$: all partial derivatives of every finite order exist and are continuous, see \citeA[~p.235]{folland1999real}. } whose derivative decay faster than any polynomial grows at infinity:
\begin{equation}
\mathcal{S}(\mathbb{R})=\left\{ f\in C^{\infty}(\mathbb{R}): \ \sup_{x\in\mathbb{R}^n} \ (1+|x|)^N|\partial^{\alpha}f(x)|<\infty \quad \text{for all} \ N,\alpha \right\}
\label{equ:schwartz}
\end{equation}
clearly $C_c^{\infty}\subset \mathcal{S}$. Since $\xi\in C_c^{\infty}(\mathbb{R})$, it follows that $\xi\in\mathcal{S}(\mathbb{R})$ as well. By \citeA[~Corollary 8.23, Corollary 8.28]{folland1999real}, the Fourier transform maps $\mathcal{S}(\mathbb{R})$ into itself, and the inverse Fourier transform does as well. In particular, taking $N=2$ and $\alpha=0$ in \eqref{equ:schwartz}, there exists a constant $C_h>0$ such that
\begin{equation*}
|h(x)|\leq C_h(1+x^2)^{-1} \quad \quad \text{for all} \ x\in\mathbb{R}.
\end{equation*}
Finally, Fourier inversion gives \(h^\star=\xi\), so $\int h(x)\,dx = h^\star(0)=\xi(0)=0$. For $T\ge 1$ and a small constant $\nu>0$, we define $H_T(\theta):=\nu T^{-s+1/2}h\bigl(T(\theta-q_0)\bigr)$. Then its Fourier transform is
\begin{equation}
H_T^{\star}(t)
=
\nu T^{-s-1/2}(-i\,\text{sgn}(t))\exp(itq_0)k^{\star}(|t|/T).
\label{equ:H_T_fourier}
\end{equation}

\subsubsection*{Step 1. Verifying that $g_1(\cdot)\in\mathcal{G}(s,L)$ for sufficiently large $T$. }

First, $g_1$ integrates to one. Indeed, using the change of variable $u=T(\theta-q_0)$,
\begin{equation*}
\int_{\mathbb{R}} H_T(\theta)\mathrm{d}\theta=\nu T^{-s+1/2}\int_{\mathbb{R}} h\bigl(T(\theta-q_0)\bigr) \mathrm{d}\theta=\nu T^{-s-1/2}\int_{\mathbb{R}} h\bigl(u\bigr) \mathrm{d}u=0.
\end{equation*}
Therefore $\int_{\mathbb{R}}g_1(\theta)\mathrm{d}\theta=\int_{\mathbb{R}}g_0(\theta)\mathrm{d}\theta+\int_{\mathbb{R}} H_T(\theta)\mathrm{d}\theta=1$.

Second, $g_1(\theta)$ is nonnegative if $\nu$ is chosen small enough. Since $q_0$ is fixed, there exists a constant $C_0$ such that $1+\theta^2\leq C_0(1+(\theta-q_0)^2)$ for all $\theta\in\mathbb{R}$. Hence
\begin{equation*}
(1+\theta^2)|H_T(\theta)|
\le C_0C_h\nu T^{-s+1/2}\frac{1+(\theta-q_0)^2}{1+T^2(\theta-q_0)^2}
\le C\nu T^{-s+1/2}
\le C\nu
\end{equation*}
because $T\geq 1$ and $s\geq 1/2$. Therefore
\begin{equation*}
|H_T(\theta)|\leq \frac{C\nu}{1+\theta^2} \quad \text{for all} \ \theta\in\mathbb{R}.
\end{equation*}
On the other hand, $g_0(\theta)\geq \left[\pi\lambda(1+\theta^2)\right]^{-1} $ for $\lambda\geq 1$, thus $|H_T(\theta)|\leq C\nu\pi\lambda g_0(\theta)$. If we choose $\nu>0$ so small that $C\nu\pi\lambda\leq 1/2$, then $|H_T(\theta)|\leq \frac{1}{2}g_0(\theta)$ for all $\theta\in\mathbb{R}$. Consequently, $g_1(\theta)=g_0(\theta)+H_T(\theta)\geq g_0(\theta)-|H_T(\theta)|\geq \frac{1}{2}g_0(\theta)\geq 0$ for all $\theta\in\mathbb{R}$. Thus $g_1$ is a density.

Third, we verify that $g_1\in\mathcal{G}(s,L)$. Since $g_1^{\star}=g_0^{\star}+H_T^{\star}$, we have that $|g_1^{\star}(t)|^2\leq 2|g_0^{\star}(t)|^2+2|H_T^{\star}(t)|^2$ and hence
\begin{equation*}
\int_{\mathbb{R}}|g_1^{\star}(t)|^2(1+t^2)^s\,\mathrm{d}t
\le
2\int_{\mathbb{R}}|g_0^{\star}(t)|^2(1+t^2)^s\,\mathrm{d}t
+
2\int_{\mathbb{R}}|H_T^{\star}(t)|^2(1+t^2)^s\,\mathrm{d}t.
\end{equation*}
By construction, the first term is at most $L/2$. For the second term,
\begin{align*}
\int_{\mathbb{R}}|H_T^{\star}(t)|^2(1+t^2)^s\,\mathrm{d}t
&=\nu^2T^{-2s-1}\int_{\mathbb{R}} |\xi(t/T)|^2(1+t^2)^s\,\mathrm{d}t=\nu^2 T^{-2s}\int_{\mathbb{R}} |\xi(u)|^2(1+T^2u^2)^s\mathrm{d}u \\[7pt]
&=\nu^2\int_{\mathbb{R}} |\xi(u)|^2(T^{-2}+u^2)^s\mathrm{d}u \leq \int_{\mathbb{R}} |\xi(u)|^2(1+u^2)^s\mathrm{d}u.
\end{align*}
Because $\xi\in C_c^{\infty}(\mathbb{R})$, the integral is finite. Hence, after possibly decreasing $\nu$, we may ensure that $\int |H_T^{\star}(t)|^2(1+t^2)^s\mathrm{d}t\leq L/4.$ Therefore $g_1\in\mathcal G(s,L)$.

\subsubsection*{Step 2. Lower bound for the distance $|q_{G_1}(\tau;y)-q_{G_0}(\tau;y)|$.}

Here $q_{G_1}(\tau;y), q_{G_0}(\tau;y)$ are posterior quantiles that correspond to densities $g_1, g_0$ defined in \Cref{equ:density_pair}. For notational convenience, write $q_j=q_{G_j}(\tau;y)$ for $j=0,1$. For $j=0,1$, define
\[
\Psi_j(q):=\int \Big( \mathbf{1}\{\theta\leq q\}-\tau \Big)\phi\left(\frac{y-\theta}{\sigma}\right)g_j(\theta)\mathrm{d}\theta.
\]
Then the first-order conditions give that \(\Psi_j(q_j)=0\) for \(j=0,1\).

We first show that $q_1$ lies in a small neighborhood of $q_0$. Since $\Psi_0'(q)=\phi\left(\frac{y-q}{\sigma}\right)g_0(q)$, the positivity assumption at \(q_0\) implies that \(\Psi_0'(q_0)>0\). Hence there exists
\(\delta>0\) such that $m_\delta:=\min\{-\Psi_0(q_0-\delta),\,\Psi_0(q_0+\delta)\}>0$. Moreover, because \(g_1-g_0=H_T\),
\begin{align*}
\sup_{q\in\mathbb R}\big|\Psi_1(q)-\Psi_0(q)\big|
\leq & \sup_{q\in\mathbb{R}} \ \int \phi\left(\frac{y-\theta}{\sigma}\right)|H_T(\theta)|\mathrm{d}\theta     \\[7pt]
\leq &\left\|\phi\left(\frac{y-\cdot}{\sigma}\right)\right\|_{\infty}\|H_T\|_1
\leq C\nu T^{-s-1/2}.
\end{align*}
Therefore, for all sufficiently large \(T\), $\Psi_1(q_0-\delta)<0<\Psi_1(q_0+\delta)$. Since \(q\mapsto \Psi_1(q)\) is increasing, it follows that $q_1\in(q_0-\delta,q_0+\delta)$.

Next, subtracting $\Psi_0(q_0)=0$ from $\Psi_1(q_1)=0$ gives
\begin{align*}
\int \left(\mathbf{1}\{\theta\leq q_0\}-\tau\right)\phi\left(\frac{y-\theta}{\sigma}\right)\left(g_1(\theta)-g_0(\theta)\right)\mathrm{d}\theta
&=
\int \Big( \mathbf{1}\{\theta\leq q_0\}-\mathbf{1}\{\theta\leq q_1\} \Big)\phi\left(\frac{y-\theta}{\sigma}\right)g_1(\theta)\mathrm{d}\theta.
\end{align*}
By the mean-value theorem, there exists \(q^{*}\) between \(q_0\) and \(q_1\) such that
\begin{align*}
\int \Big( \mathbf{1}\{\theta\leq q_0\}-\mathbf{1}\{\theta\leq q_1\} \Big)\phi\left(\frac{y-\theta}{\sigma}\right)g_1(\theta)\mathrm{d}\theta
=
(q_0-q_1)\phi\left(\frac{y-q^*}{\sigma}\right)g_1(q^*).
\end{align*}
Hence
\begin{equation}
q_1-q_0=-\left[ \phi\left(\frac{y-q^*}{\sigma}\right)g_1(q^*)  \right]^{-1}
\left[\int \Big( \mathbf{1}\{\theta\leq q_0\}-\tau \Big)\phi\left(\frac{y-\theta}{\sigma}\right)\left(g_1(\theta)-g_0(\theta)\right)\mathrm{d}\theta\right].
\label{equ:q1q0}
\end{equation}
The denominator in \eqref{equ:q1q0} is bounded away from zero. Indeed, because \(q^*\in[q_0-\delta,q_0+\delta]\), \(g_1\geq g_0/2\), and since $g_1\geq g_0/2$ on this compact interval for all
sufficiently large \(T\), while \(g_0\) is continuous and strictly positive there, there exists \(c_1>0\) such that
\begin{equation}
\phi\left(\frac{y-q^*}{\sigma}\right)g_1(q^*)\geq c_1.
\label{equ:step4_den_lb}
\end{equation}

It remains to lower bound the numerator in \eqref{equ:q1q0}. Define $\psi(q_0,\theta;y)=(\mathbf{1}\{\theta\leq q_0\}-\tau)\phi(\frac{y-\theta}{\sigma})$ and write
\[
I_T:=\int \psi(q_0,\theta;y)H_T(\theta)\mathrm{d}\theta=\frac{1}{2\pi}\int \overline{\psi^{\star}(q_0,t;y)}\,H_T^{\star}(t)\mathrm{d}t.
\]

By \Cref{propo:fourier_decay}, we have that $\overline{\psi^{\star}(q_0,t;y)}=\phi(\frac{y-q_0}{\sigma})\frac{\exp(-itq_0)}{it}+r(t;y)$ where $|r(t;y)|\leq Ct^{-2}$. Using the expression of $H_T^{\star}(t)$ from \eqref{equ:H_T_fourier}, we obtain that $I_T=I_{T1}+I_{T2}$ where
\begin{align*}
I_{T1}&=\frac{\nu}{2\pi}T^{-s-1/2}\phi\left(\frac{y-q_0}{\sigma}\right)
\int \frac{\exp(-itq_0)}{it}(-i\,\mathrm{sgn}(t))\exp(itq_0)k^{\star}(|t|/T)\,\mathrm{d}t  \\[7pt]
I_{T2}&=\frac{\nu}{2\pi}T^{-s-1/2}
\int r(t;y)(-i\,\mathrm{sgn}(t))\exp(itq_0)k^{\star}(|t|/T)\,\mathrm{d}t
\end{align*}
We bound these two terms separately. For the first term,
\begin{align*}
|I_{T1}|&=\frac{\nu}{2\pi}T^{-s-1/2}\phi\left(\frac{y-q_0}{\sigma}\right)
\int \frac{k^{\star}(|t/T|)}{|t|}\mathrm{d}t\overset{(a)}{=}\frac{\nu}{\pi}T^{-s-1/2}\phi\left(\frac{y-q_0}{\sigma}\right)
\int_T^{2T} \frac{k^{\star}(t/T)}{|t|}\mathrm{d}t \\[7pt]
&\overset{(b)}{=}\frac{\nu}{\pi}T^{-s-1/2}\phi\left(\frac{y-q_0}{\sigma}\right)
\int_1^{2} \frac{k^{\star}(u)}{u}\mathrm{d}u.
\end{align*}
where (a) uses $\mathrm{supp}(k^{\star})\subset(1,2)$, and (b) follows from the change of variable $u=t/T$. Since $k^{\star}$ is nonnegative and not identically zero, $\int_1^{2} \frac{k^{\star}(u)}{u}\mathrm{d}u>0$. Therefore, $|I_{T1}|\geq c_2\nu T^{-s-1/2}$ for some constant $c_2>0$.

For the second term,
\begin{align*}
|I_{T2}|&=C\nu T^{-s-1/2}\int |r(t;y)| k^{\star}(|t|/T)\,\mathrm{d}t\leq C\nu T^{-s-1/2}\int_{T}^{2T}t^{-2}\mathrm{d}t\leq C\nu T^{-s-3/2},
\end{align*}
thus $I_{T2}$ is smaller than $I_{T1}$ by an additional factor $T^{-1}$. Hence, by the triangle inequality, $|I_T|\geq |I_{T1}|-|I_{T2}|\geq c_3\nu T^{-s-1/2}$ for all sufficiently large $T$ where $c_3>0$.

Combining this bound with \eqref{equ:q1q0} and \eqref{equ:step4_den_lb}, we conclude that $|q_1-q_0|\geq C \nu T^{-s-1/2}$.

\subsubsection*{Step 3. Upper bound for chi-divergence $\chi^2\left( (f_{Y})_1,(f_{Y})_0 \right)$.  }

Let $u_T(y)=(f_{Y})_1(y)-(f_{Y})_0(y)=H_T*f_{\varepsilon}(y)$, where $f_{\varepsilon}$ is the $N(0,\sigma^2)$ density.

We first show that $(f_{Y})_0$ is bounded below by a Cauchy tail. Choose $\eta>0$ such that $p_{\eta}:=\int_{|z|\leq \eta} f_{\varepsilon}(z)\mathrm{d}z>0$. Then
\begin{equation*}
(f_{Y})_0(y)=\int_{\mathbb{R}} g_0(y-z)f_{\varepsilon}(z)\mathrm{d}z\geq \int_{|z|\leq \eta} g_0(y-z)f_{\varepsilon}(z)\mathrm{d}z\geq p_{\eta} \inf_{|z|\leq \eta} g_0(y-z).
\end{equation*}
Moreover, for $|z|\leq \eta$, $1+(y-z)^2\leq C_{\eta}(1+y^2)$, so $g_0(y-z)=\frac{\lambda}{\pi(\lambda^2+(y-z)^2)}\geq \frac{c_{\eta}}{1+y^2}$. Hence $(f_Y)_0(y)\geq \frac{c}{1+y^2}$ for all $y\in\mathbb{R}$. Therefore
\begin{equation*}
\chi^2\left( (f_{Y})_1,(f_{Y})_0 \right)=\int_{\mathbb{R}} \cfrac{u_T(y)^2}{(f_{Y})_0(y)}\mathrm{d}y\leq C\int_{\mathbb{R}} (1+y^2)u_T(y)^2\mathrm{d}y.
\end{equation*}
Then it suffices to show that $S_T:=\int_{\mathbb{R}} (1+y^2)u_T(y)^2\mathrm{d}y=O(n^{-1})$.

By Parseval's identity, we have that
\begin{equation}
\int u_T(y)^2\mathrm{d}y=\frac{1}{2\pi} \int |H_T^{\star}(t)|^2\exp\left(-\sigma^2 t^2\right)\mathrm{d}t\leq C\nu^2 T^{-2s}\exp(-\sigma^2 T^2).
\label{appendix:equ_integral_1}
\end{equation}
Also, since the Fourier transformation of $yu_T(y)$ is $i\partial_t(u_T^{\star}(t))$,
\begin{equation*}
\int_{\mathbb R}y^2u_T(y)^2\mathrm{d}y
=
\frac{1}{2\pi}\int
\left|
\partial_t\Big(H_T^{\star}(t)\exp(-\sigma^2t^2/2)\Big)
\right|^2
\mathrm{d}t.
\end{equation*}
On the support of $H_T^{\star}$, namely $T\leq |t|\leq 2T$, we have that $|H_T^{\star}(t)|\leq C\nu T^{-s-1/2}$, and $|\partial_tH_T^{\star}(t)|\leq C\nu T^{-s-1/2}.$ Hence
\begin{equation*}
\left|
\partial_t\Big(H_T^{\star}(t)\exp(-\sigma^2t^2/2)\Big)
\right|
\leq
C\nu T^{-s+1/2}\exp(-\sigma^2T^2/2).
\end{equation*}
Since the support has length $O(T)$, it follows that
\begin{equation}
\int_{\mathbb R}y^2u_T(y)^2\mathrm{d}y
\leq
C\nu^2T^{-2s+2}\exp(-\sigma^2T^2).
\label{appendix:equ_integral_2}
\end{equation}
Consequently, $S_{T}\leq C\nu^2T^{-2s+2}\exp(-\sigma^2T^2)$. Now set $T=T_n:=\sigma^{-1}\sqrt{ (1+\eta)\log n }$ for any fixed $\eta>0$, then we have that $S_{T_n}\lesssim Cn^{-(1+\eta)}(\log n)^{-s+1}$, and therefore $S_T=O(n^{-1})$. Hence $\chi^2\!\left((f_Y)_1,(f_Y)_0\right)=O(n^{-1}).$

By Proposition~\ref{prop:A2_two_point},
\[
\sup_{g\in\mathcal G(s,L)} \operatorname{E}\!\left[\bigl(\hat q-q_G(\tau;y)\bigr)^2\right]
\ge
c\,|q_1-q_0|^2
\ge
c\,T_n^{-2s-1}
\asymp
(\log n)^{-(2s+1)/2}.
\]
This proves Theorem~\ref{rate_posterior_quantile}.

\subsection{Proof of \Cref{quantile_estimator}}
\label{appendix:proof_theorem3_2}

For notational convenience, let $\widehat{q}=\widehat{q}_G(\tau;y)$ denote the estimated posterior quantile, and $q_0=q_{G_0}(\tau;y)$ its population counterpart. The proof proceeds in two main steps. First, we establish the following asymptotic expansion in \Cref{subsubsection:expansion}
\begin{equation}
r_n\left(\widehat{q}-q_0\right)=r_n\left[ \phi\left(\frac{y-q_0}{\sigma}\right)g(q_0) \right]^{-1}\Psi_{n}(q_0)+o_p(1)
\label{equ:expansion}
\end{equation}
with a suitable sequence $r_n\rightarrow\infty$ as $n\rightarrow\infty$, and the term $\Psi_n(q_0)$\footnote{Note that its population counterpart $\Psi(q_0)=(2\pi)^{-1}\int \overline{\psi^{\star}(q_0,t;y)}g^{\star}(t)\mathrm{d}t=\int \left(\tau-\mathbf{1}\{\theta\leq q_0\}\right)\phi\left(\frac{y-\theta}{\sigma}\right)g(\theta)\mathrm{d}\theta\equiv 0$, which corresponds to the moment condition. In this sense, $\Psi_n(q_0)$ can be viewed as a sample moment function, and $\widehat{q}_n$ as the associated $Z$-estimator. } is given by
\begin{equation}
\Psi_n(q_0)=(2\pi)^{-1} \int  \cfrac{\overline{\psi^{\star}(q_0,t;y)}}{f_{\varepsilon}^{\star}(t)} \left(\widehat{f}^{\star}_Y(t)K^{\star}(h_n t) \right) \mathrm{d}t
\label{equ:Psi_n_q0}
\end{equation}
where the Fourier coefficient $\overline{\psi^{\star}(q_0,t;y)}$ is defined in \Cref{equ:weighting_function}. Second we establish the properties of the leading term $\Psi_n(q_0)$, namely its consistency and rate of convergence, in \Cref{subsubsection:analysis_Psi_n}.

\subsubsection{Asymptotic expansion}
\label{subsubsection:expansion}

Define that $\widehat{\delta}_n=r_n(\widehat{q}-q_0)$ and the localized, rescaled function $\widetilde{W}_n(\delta):=r_n^2\left\{W_n(q_0 + r_n^{-1} \delta) - W_n(q_0)\right\}$ for $\delta\in\mathbb{R}$. Observe that $\widehat{\delta}_n$ is the unique minimizer of $\widetilde{W}_n$. We now derive a quadratic approximation to $\widetilde{W}_n$ on compact sets, and then use Lemma 2 of \citeA{hjort2011asymptotics} to pass from the approximation of the criterion to the approximation of its minimizer.

By \Cref{proposition_c1}, we can write
\begin{equation*}
\rho_{\tau}\left( \theta-(r_n^{-1}\delta+q_0) \right)-\rho_{\tau}\left( \theta-q_0 \right)=-(r_n^{-1}\delta)\Big{(}\tau-\mathbf{1}\{\theta-q_0\leq 0\}\Big{)}+\Delta
\end{equation*}
where
\begin{equation*}
\Delta=\int_0^{r_n^{-1}\delta}\Big{(} \mathbf{1}\{\theta-q_0\leq s\}-\mathbf{1}\{\theta-q_0\leq 0\} \Big{)}\mathrm{d}s.
\end{equation*}
Then we can decompose the objective function $\widetilde{W}_n(\delta)=-r_n\delta\Psi_n(q_0)+r_n^2W_{n2}(\delta)$ where
\begin{equation}
W_{n2}(\delta)=(2\pi)^{-1}\int \cfrac{\int \exp(-it\theta)\Delta\phi(\frac{y-\theta}{\sigma})\mathrm{d}\theta}{f_{\varepsilon}^{\star}(t)}\left(\widehat{f}^{\star}_Y(t)K^{\star}(h_n t) \right) \mathrm{d}t.
\label{equ:W_n2}
\end{equation}

By Parseval-Plancherel theorem, we have that $W_{n2}(\delta)=\int \Delta\phi\left(\frac{y-\theta}{\sigma}\right)\widehat{g}(\theta)\mathrm{d}\theta$ where
\begin{equation}
\widehat{g}(\theta)=(2\pi)^{-1}\int \cfrac{\exp(-it\theta)}{f_{\varepsilon}^{\star}(t)}\left(\widehat{f}^{\star}_Y(t)K^{\star}(h_n t) \right) \mathrm{d}t
\label{equ:deconvolution_estimator}
\end{equation}
is algebraically the classical deconvolution kernel density estimator.

We calculate $W_{n2}(\delta)$ as follows
\allowdisplaybreaks
\begin{align*}
W_{n2}(\delta)&=\int \left[\int_0^{r_n^{-1}\delta}\Big{(} \mathbf{1}\{\theta-q_0\leq s\}-\mathbf{1}\{\theta-q_0\leq 0\} \Big{)}\mathrm{d}s\right]\phi\left(\frac{y-\theta}{\sigma}\right)\hat{g}(\theta)\mathrm{d}\theta \\[7pt]
&\overset{(a)}{=}r_n^{-1}\int\left[ \int_0^{\delta}\Big{(} \mathbf{1}\{\theta-q_0\leq r_n^{-1}z\}-\mathbf{1}\{\theta-q_0\leq 0\} \Big{)}\mathrm{d}z  \right]\phi\left(\frac{y-\theta}{\sigma}\right)\hat{g}(\theta)\mathrm{d}\theta \\[7pt]
&\overset{(b)}{=}r_n^{-1}\int_0^{\delta} \left[ \int \Big{(} \mathbf{1}\{\theta-q_0\leq r_n^{-1}z\}-\mathbf{1}\{\theta-q_0\leq 0\} \Big{)}\phi\left(\frac{y-\theta}{\sigma}\right)\hat{g}(\theta)\mathrm{d}\theta  \right]\mathrm{d}z \\[7pt]
&=r_n^{-1}\int_0^{\delta}\int_{q_0}^{q_0+r_n^{-1}z}\phi\left(\frac{y-\theta}{\sigma}\right)\hat{g}(\theta)\mathrm{d}\theta \mathrm{d}z
\overset{(c)}{=}r_n^{-1}\int_{q_0}^{q_0+r_n^{-1}\delta}\int_{r_n(\theta-q_0)}^{\delta} \phi\left( \frac{y-\theta}{\sigma} \right)\hat{g}(\theta)\mathrm{d}z \mathrm{d}\theta \\[7pt]
&=r_n^{-1}\int_{q_0}^{q_0+r_n^{-1}\delta}\left[ \delta-r_n(\theta-q_0)\right]\phi\left( \frac{y-\theta}{\sigma} \right)\hat{g}(\theta)\mathrm{d}\theta  \\[7pt]
&\overset{(d)}{=}r_n^{-2}\int_0^{\delta} (\delta-u) \phi\left( \frac{y-q_0-r_n^{-1}u}{\sigma} \right)\hat{g}(q_0+r_n^{-1}u)\mathrm{d}u
\end{align*}
where (a) uses changes of variable $s=r_n^{-1}z$, (b) and (c) switch the order of integration, (d) uses changes of variable again $u=r_n(\theta-q_0)$.

We  now show that $r_n^2W_{n2}(\delta)$ is asymptotically quadratic, uniformly on compact sets. By \Cref{assumption_density,assumption2,assumption4} and continuity of $g$, there exists $\eta>0$ such that for $I=[q_0-\eta,q_0+\eta]$, $c_I=\inf_{q\in I} \ \phi(\frac{y-q}{\sigma})g(q)>0$.

Fix $M>0$. Since $r_n\rightarrow\infty$, for all sufficiently large $n$, $q_0+r_n^{-1}u\in I $ whenever $|u|\leq M$. By the uniform consistency of $\hat{g}$ (\Cref{prop:uniform_consist}) and continuity of $q\mapsto \phi((y-q)/\sigma)g(q)$, we have that
\begin{equation*}
\sup_{|u|\leq M} \ \left|\phi\left( \frac{y-q_0-r_n^{-1}u}{\sigma} \right)\hat{g}(q_0+r_n^{-1}u)-\phi\left(\frac{y-q_0}{\sigma}\right)g(q_0)  \right|=o_p(1).
\end{equation*}
Therefore,
\begin{align*}
&\sup_{|\delta|\leq M} \ \left| r_n^2W_{n2}(\delta)-\frac{\delta^2}{2}\phi\left(\frac{y-q_0}{\sigma}\right)g(q_0) \right| \\[7pt]
=& \sup_{|\delta|\leq M} \ \left| \int_{0}^{\delta} (\delta-u)\left[\phi\left( \frac{y-q_0-r_n^{-1}u}{\sigma} \right)\hat{g}(q_0+r_n^{-1}u)-\phi\left(\frac{y-q_0}{\sigma}\right)g(q_0)  \right]\mathrm{d}u \right| \\[7pt]
\leq & M^2\sup_{|u|\leq M} \ \left|\phi\left( \frac{y-q_0-r_n^{-1}u}{\sigma} \right)\hat{g}(q_0+r_n^{-1}u)-\phi\left(\frac{y-q_0}{\sigma}\right)g(q_0)  \right|=o_p(1).
\end{align*}

Define the quadratic approximation $\widetilde{W}_n^0(\delta)=-r_n\delta\Psi_n(q_0)+\frac{\delta^2}{2}\phi\left(\frac{y-q_0}{\sigma}\right)g(q_0)$. Then for every fixed $M>0, $
\begin{equation*}
\sup_{|\delta|\leq M} \ \left| \widetilde{W}_n(\delta)-\widetilde{W}_n^0(\delta) \right|=o_p(1).
\end{equation*}

Recall that $\widehat{\delta}_n$ is the minimizer of $\widetilde{W}_n(\delta)$, and let
\begin{equation*}
\widehat{\delta}_0=\underset{\delta\in\mathbb{R}}{\arg\min} \ \widetilde{W}_n^0(\delta)=r_n\left[\phi\left(\frac{y-q_0}{\sigma}\right)g(q_0)\right]^{-1}\Psi_n(q_0)
\end{equation*}
where the minimizer is assumed to be unique. The next subsubsection shows that $r_n\Psi_n(q_0)=O_p(1)$, hence $\widehat{\delta}_0=O_p(1)$.

We can now compare the argmins of $\widetilde{W}_n$ and $\widetilde{W}_n^0$. First by \Cref{prop:uniform_consist}, $\sup_{q\in I} |\hat{g}(q)-g(q)|=o_p(1)$. Hence $\inf_{q\in I} \phi(\frac{y-q}{\sigma})\widehat{g}(q)\geq c_I/2>0$ with probability tending to one. On this event, we have that
\begin{equation*}
\Psi_n'(q)=-\phi\left(\frac{y-q}{\sigma}\right)\widehat{g}(q)\leq -\frac{c_I}{2}<0 \quad \quad \text{for} \ q\in I
\end{equation*}
so $\Psi_n$ is strictly decreasing on $I$, and therefore $W_n$ is strictly convex on $I$.

Fix $\varepsilon>0$ and $\eta>0$. Since $\widehat{\delta}_0=O_p(1)$, there exists $M>0$ such that $\operatorname{P}(|\widehat{\delta}_0|>M)\leq \varepsilon$ for all sufficiently large $n$. By Lemma 2 of \citeA{hjort2011asymptotics},
\[
\operatorname{P}\bigl(|\widehat\delta_n-\widehat\delta_0|\ge \eta\bigr)
\le
\operatorname{P}\!\left(
\sup_{|\delta-\widehat\delta_0|\le \eta}
\left|
\widetilde W_n(\delta)-\widetilde W_n^0(\delta)
\right|
\ge
\frac12
\inf_{|\delta-\widehat\delta_0|=\eta}
\Bigl[
\widetilde W_n^0(\delta)-\widetilde W_n^0(\widehat\delta_0)
\Bigr]
\right).
\]
Because $\widetilde W_n^0$ is quadratic,
\[
\inf_{|\delta-\widehat\delta_0|=\eta}
\Bigl[
\widetilde W_n^0(\delta)-\widetilde W_n^0(\widehat\delta_0)
\Bigr]
=
\frac{\eta^2}{2}\phi\!\left(\frac{y-q_0}{\sigma}\right)g(q_0).
\]
Therefore,
\[
\begin{aligned}
\operatorname{P}\bigl(|\widehat\delta_n-\widehat\delta_0|\ge \eta\bigr)
\le
\operatorname{P}(|\widehat\delta_0|>M)+
\operatorname{P}\!\left(
\sup_{|\delta|\le M+\eta}
\left|
\widetilde W_n(\delta)-\widetilde W_n^0(\delta)
\right|
\ge
\frac{\eta^2}{4}\phi\!\left(\frac{y-q_0}{\sigma}\right)g(q_0)
\right).
\end{aligned}
\]
The second probability tends to zero by the uniform approximation established above, while the first is at most $\varepsilon$. Since $\varepsilon>0$ is arbitrary, we conclude that $\widehat\delta_n-\widehat\delta_0=o_p(1)$.

Equivalently,
\begin{equation}
r_n(\widehat q-q_0)
=
\left[\phi\!\left(\frac{y-q_0}{\sigma}\right)g(q_0)\right]^{-1}r_n\Psi_n(q_0)
+
o_p(1),
\end{equation}
which is the desired asymptotic expansion. \qed

\subsubsection{Analysis of $\Psi_n(q_0)$}
\label{subsubsection:analysis_Psi_n}

We decompose $\Psi_n(q_0)=\Psi_{n1}+\Psi_{n2}$ where
\begin{equation}
\begin{aligned}
\Psi_{n1}&=(2\pi)^{-1} \int  \frac{\overline{\psi^{\star}(q_0,t;y)}}{f_{\varepsilon}^{\star}(t)} \left(\widehat{f}_{Y}^{\star}(t)- f_{Y}^{\star}(t) \right)K^{\star}(h_nt) \mathrm{d}t, \\[7pt]
\Psi_{n2}&=(2\pi)^{-1} \int  \frac{\overline{\psi^{\star}(q_0,t;y)}}{f_{\varepsilon}^{\star}(t)}f_{Y}^{\star}(t)\left(K^{\star}(h_nt)-1\right) \mathrm{d}t.
\end{aligned}
\label{equ:equs1}
\end{equation}
The first term $\Psi_{n1}$ captures the effects of random sampling error in estimating $f_{Y}^{\star}(t)$. The second term $\Psi_{n2}$ is a nonstochastic bias arising from the kernel regularization.


\noindent\underline{Analysis of $\Psi_{n2}$.}

By \Cref{propo:fourier_decay}, there exists a constant $C>0$ such that
\[
|\psi^\star(q_0,t;y)|\le \frac{C}{1+|t|}
\qquad\text{for all } t\in\mathbb R.
\]
We split $\Psi_{n2}(q_0)=\Psi_{n2}^{\le}(q_0)+\Psi_{n2}^{>}(q_0)$, where the first term integrates over $|t|\le h_n^{-1}$ and the second over $|t|>h_n^{-1}$.
Because $K^\star$ is supported on $[-1,1]$, we have $K^\star(h_nt)=0$ whenever $|t|>h_n^{-1}$.
Hence
\[
|\Psi_{n2}^{>}(q_0)|
\le C\int_{|t|>h_n^{-1}}\frac{|g^\star(t)|}{1+|t|}\,dt.
\]
By Cauchy-Schwarz inequality and \Cref{assumption1},
\[
|\Psi_{n2}^{>}(q_0)|
\le
C
\Biggl(\int |g^\star(t)|^2(1+t^2)^s\,dt\Biggr)^{1/2}
\Biggl(\int_{|t|>h_n^{-1}}(1+t^2)^{-(s+1)}\,dt\Biggr)^{1/2}
=
O\bigl(h_n^{s+1/2}\bigr).
\]

For the low-frequency part $\Psi_{n2}^{\le}(q_0)$, first we aim to establish that $|K^{\star}(h_nt)-1|\lesssim |h_nt|^{s+1}$. Expanding $K^{\star}(h_nt)$ around $t=0$:
\begin{equation*}
K^{\star}(h_nt)=1+\frac{(h_nt)^{s+1}}{(s+1)!}\int (-iv)^{s+1}K(v)\exp(-i\xi h_ntv)\mathrm{d}v
\end{equation*}
for some $\xi\in[0,1]$. Then,
\begin{align*}
\left| K^{\star}(h_nt)-1 \right|^2&\lesssim  |h_nt|^{2(s+1)}\left[\int \Big{|}(-iv)^{s+1}K(v)\exp(-i\xi h_ntv)\Big{|}\mathrm{d}v \right]^2 \\[7pt]
&\overset{(a)}{\lesssim} |h_nt|^{2(s+1)}\left[ \int |v|^{s+1}|K(v)|\mathrm{d}v \right]^2 \overset{(b)}{\lesssim} |h_nt|^{2(s+1)}
\end{align*}
where (a) is obtained from $|\exp(-it\xi h_ntv)|=1$; (b) follows from \Cref{assumption3} (ii).
 Therefore,
\[
|\Psi_{n2}^{\le}(q_0)|
\le
Ch_n^{s+1}
\int_{|t|\le h_n^{-1}}
\frac{|t|^{s+1}|g^\star(t)|}{1+|t|}
\,dt.
\]
Applying Cauchy-Schwarz inequality again,
\begin{align*}
|\Psi_{n2}^{\le}(q_0)|
&\le
Ch_n^{s+1}
\Biggl(\int |g^\star(t)|^2(1+t^2)^s\,dt\Biggr)^{1/2}
\Biggl(
\int_{|t|\le h_n^{-1}}
\frac{|t|^{2(s+1)}}{(1+t^2)^{s+1}}\,dt
\Biggr)^{1/2}.
\end{align*}
It remains to bound the last integral. On $|t|\leq 1$, the integrand is bounded. On $1\leq |t|\leq h_n^{-1}$, $\frac{|t|^{2(s+1)}}{(1+t^2)^{s+1}}\lesssim 1$. Hence
\[
\int_{|t|\le h_n^{-1}}
\frac{|t|^{2(s+1)}}{(1+t^2)^{s+1}}\,dt
=
O\bigl(h_n^{-1}\bigr).
\]
Therefore we have $|\Psi_{n2}^{\le}(q_0)|=O(h_n^{s+1/2})$, and conclude that $\Psi_{n2}(q_0)=O(h_n^{s+1/2})$.

\noindent\underline{Analysis of $\Psi_{n1}$.}

Observe that $\Psi_{n1}$ is equal to an average of i.i.d. term:
\begin{equation*}
\Psi_{n1}=n^{-1}\sum_{i=1}^n (2\pi)^{-1}\int\frac{\overline{\psi^{\star}(q_0,t;y)}}{f_{\varepsilon}^{\star}(t)} \Big{(} \exp(itY_i)-f_{Y}^{\star}(t) \Big{)}K^{\star}(h_nt)\mathrm{d}t.
\end{equation*}

The variance term of $\Psi_{n1}$ can be computed as
\begin{equation}
\begin{aligned}
\var\left(\Psi_{n1}\right)&\leq \frac{1}{4\pi^2n}\operatorname{E}\left| \int \frac{\overline{\psi^{\star}(q_0,t;y)}}{f_{\varepsilon}^{\star}(t)}\exp(it Y_i)K^{\star}(h_nt)\mathrm{d}t \right|^2\\[7pt]
\end{aligned}
\label{appendix:var_psi}
\end{equation}
Letting $\omega(q_0,t;y)=\overline{\psi^{\star}(q_0,t;y)}\exp(ity)/f_{\varepsilon}^{\star}(t)$, the first term of the right-hand side of \Cref{appendix:var_psi} can be computed as
\begin{align}
\text{First term of Equ. (\ref{appendix:var_psi})}&=\frac{1}{4\pi^2 n}\operatorname{E}\left| \int \omega(q_0,t;y)\exp(-it(y-Y_i))K^{\star}(h_nt)\mathrm{d}t \right|^2 \notag \\[7pt]
&=\frac{1}{4\pi^2 n}\int \left| \int \omega(q_0,t;y)\exp(-itz)K^{\star}(h_nt)\mathrm{d}t \right|^2 f_{Y}(y-z)\mathrm{d}z \notag \\[7pt]
&\leq \frac{1}{4\pi^2n} \Big{\|} f_{Y} \Big{\|}_{\infty} \int \left| \int \omega(q_0,t;y)\exp(-it z) K^{\star} (h_nt)\mathrm{d}t \right|^2 \mathrm{d}z \notag \\[7pt]
&\overset{(a)}{=}\frac{1}{2\pi n}\Big{\|} f_{Y} \Big{\|}_{\infty}  \int \left| \omega(q_0,t;y)K^{\star} (h_nt) \right|^2 \mathrm{d}t \label{appendix:equ_l2}
\end{align}
where (a) is obtained from Parseval's identity, and $\|f_{Y}\|_{\infty}$ is finite since that
\begin{equation*}
\sup_{y\in\mathbb{R}}\ \Big{|} f_{Y}(y) \Big{|} \leq \int \left( \sup_{y\in\mathbb{R}} |g(y-z)| \right)\mathrm{d}F_{\varepsilon}(z)=\Big{\|} g \Big{\|}_{\infty}<\infty.
\end{equation*}


We calculate that $\omega(q_0,t;y)=\sigma\tau +\eta(t)$ where
\begin{equation}
\eta(t)=\exp\left(\frac{1}{2}t^2\sigma^2\right)\left[ -\frac{\exp(-it(q_0-y))}{it}\phi\left(\frac{y-q_0}{\sigma}\right)+o(t^{-1}) \right].
\label{equ:eta_t}
\end{equation}
Then we have that
\begin{equation}
\int \left| \omega(q_0,t;y)K^{\star} (h_nt) \right|^2 \mathrm{d}t \leq  2\sigma^2 \int \left| K^{\star}(h_nt) \right|^2\mathrm{d}t+2\int \left| \eta(t)K^{\star}(h_nt) \right|^2 \mathrm{d}t.
\label{equ:appendix_e13}
\end{equation}

For the first term in \Cref{equ:appendix_e13}, by Parseval's identity and \Cref{assumption3}, we have that $\int \left| K^{\star}(h_nt) \right|^2 \mathrm{d}t$ $=h_n^{-1}\int\left| K^{\star}(t) \right|^2 \mathrm{d}t=  2\pi h_n^{-1}\int \left| K(z) \right|^2\mathrm{d}z$$=O(h_n^{-1})$.

It remains to bound the second term in \Cref{equ:appendix_e13}
\[
I_{n2}(h_n):=\int \left| \eta(t)K^{\star}(h_nt) \right|^2dt .
\]
Choose $T>0$ large enough that, by \eqref{equ:eta_t},
\[
|\eta(t)|
\le
C\,\exp\!\left(\frac{1}{2}t^2\sigma^2\right)|t|^{-1}
\qquad\text{for all }|t|\ge T .
\]
Since $\eta$ is continuous on compact sets and $K^\star$ is bounded, we may split $I_{n2}(h_n)$ as
\[
I_{n2}(h_n)
=
\int_{|t|\le T}\left| \eta(t)K^{\star}(h_nt) \right|^2dt
+
\int_{T<|t|\le h_n^{-1}}\left| \eta(t)K^{\star}(h_nt) \right|^2dt
=: I_{n2}^{(1)}(h_n)+I_{n2}^{(2)}(h_n),
\]
where the upper limit $h_n^{-1}$ comes from the support condition on $K^\star(h_nt)$. The first term satisfies $I_{n2}^{(1)}(h_n)=O(1).$ For the second term, the above bound on $\eta(t)$ yields
\begin{align*}
I_{n2}^{(2)}(h_n)
&\le
C\int_{T<|t|\le h_n^{-1}}
\frac{\exp(t^2\sigma^2)}{t^2}|K^{\star}(h_nt)|^2dt \\[7pt]
&\le
C h_n\int_{Th_n\le |u|\le 1}
\frac{\exp(u^2h_n^{-2}\sigma^2)}{u^2}|K^{\star}(u)|^2du
\end{align*}
with the change of variable $u=h_nt$.

Fix any $\delta\in(0,1)$. Split this integral into the interior region $Th_n\le |u|\le 1-\delta$ and the boundary region $1-\delta\le |u|\le 1$. On the interior region,
\begin{align*}
h_n\int_{Th_n\le |u|\le 1-\delta}
\frac{\exp(u^2h_n^{-2}\sigma^2)}{u^2}|K^{\star}(u)|^2du
&\le
C h_n\exp\!\bigl(\sigma^2(1-\delta)^2h_n^{-2}\bigr)
\int_{Th_n\le |u|\le 1-\delta}u^{-2}du \\
&=
O\!\left(\exp\!\bigl(\sigma^2(1-\delta)^2h_n^{-2}\bigr)\right).
\end{align*}
Since $h_n\to 0$, it follows that $\exp\!\bigl(\sigma^2(1-\delta)^2h_n^{-2}\bigr)
=
o\!\left(h_n^3\exp(\sigma^2h_n^{-2})\right).$

On the boundary region, the contributions from neighborhoods of $u=1$ and $u=-1$ are the same, so it suffices to consider $u\in[1-\delta,1]$. By boundedness of $K^\star$,
\begin{align*}
h_n\int_{1-\delta\le |u|\le 1}
\frac{\exp(u^2h_n^{-2}\sigma^2)}{u^2}|K^{\star}(u)|^2\mathrm{d}u
&\le
C h_n\int_{1-\delta}^{1}\exp(u^2h_n^{-2}\sigma^2)\,\mathrm{d}u \\[7pt]
&=
Ch_n\exp(\sigma^2h_n^{-2})
\int_{0}^{\delta}
\exp\!\left(-\sigma^2(2v-v^2)h_n^{-2}\right)\mathrm{d}v \\[7pt]
&\le
Ch_n\exp(\sigma^2h_n^{-2})
\int_{0}^{\delta}
\exp(-\sigma^2 v h_n^{-2})\,\mathrm{d}v \\[7pt]
&=
O\!\left(h_n^3\exp(\sigma^2h_n^{-2})\right)
\end{align*}
where $u=1-v$, and $2v-v^2\ge v$ for $v\in[0,\delta]\subset[0,1]$.

Combining the interior and boundary bounds, we obtain $
I_{n2}^{(2)}(h_n)
=
O\!\left(h_n^3\exp(\sigma^2h_n^{-2})\right)$ and thus
$
I_{n2}(h_n)
=
O(1)+O\!\left(h_n^3\exp(\sigma^2h_n^{-2})\right)
=
O\!\left(h_n^3\exp(\sigma^2h_n^{-2})\right),
$ since $h_n^3\exp(\sigma^2h_n^{-2})\to\infty$ as $h_n\to 0$.

Substituting this and the bound for the first term into \Cref{equ:appendix_e13}, we get
\[
\int \left| \omega(q_0,t;y)K^{\star}(h_nt) \right|^2dt
=
O(h_n^{-1})
+
O\!\left(h_n^3\exp(\sigma^2h_n^{-2})\right).
\]
Because $h_n^{-1}=o\!\left(h_n^3\exp(\sigma^2h_n^{-2})\right)$, it follows that
\[
\int \left| \omega(q_0,t;y)K^{\star}(h_nt) \right|^2dt
=
O\!\left(h_n^3\exp(\sigma^2h_n^{-2})\right).
\]
Hence, by \Cref{appendix:equ_l2}, $\var(\Psi_{n1})
=
O\!\left(n^{-1}h_n^3\exp(\sigma^2h_n^{-2})\right).$ Combining this with $\Psi_{n2}^2=O(h_n^{2s+1})$, we obtain
\[
\operatorname{E}\Bigl[\Psi_n(q_0)-\Psi(q_0)\Bigr]^2
=
\var(\Psi_{n1})+\Psi_{n2}^2
\le
O\!\left(n^{-1}h_n^3\exp(\sigma^2h_n^{-2})\right)
+
O(h_n^{2s+1}).
\]
Now choose $h_n=c(\log n)^{-1/2}$ for some $c>\sigma$. Then
\[
n^{-1}h_n^3\exp(\sigma^2h_n^{-2})
=
n^{-1+\sigma^2/c^2}(\log n)^{-3/2}
=
o\!\left((\log n)^{-(2s+1)/2}\right),
\]
so
\[
\operatorname{E}\Bigl[\Psi_n(q_0)-\Psi(q_0)\Bigr]^2
=
O\!\left((\log n)^{-(2s+1)/2}\right).
\]
This choice of $h_n$ also guarantees the uniform consistency of $\hat g(\cdot)$, which is needed in the analysis of $W_{n2}(\delta)$ in \eqref{equ:W_n2}, see\Cref{prop:uniform_consist}.

Finally, since $\Psi(q_0)=0$, let $M_{\varepsilon}=\sqrt{C/\varepsilon}$. For every $\varepsilon>0$, Chebyshev's inequality yields
\[
\operatorname{P}\left(
(\log n)^{(2s+1)/4}\bigl|\Psi_n(q_0)\bigr|>M_{\varepsilon}
\right)
\le
\frac{
(\log n)^{(2s+1)/2}\operatorname{E}\bigl[\Psi_n(q_0)-\Psi(q_0)\bigr]^2
}{
M_{\varepsilon}^2
}
\le
\frac{C}{M_{\varepsilon}^2}
\le
\varepsilon.
\]
Therefore $\Psi_n(q_0)=O_p\!\left((\log n)^{-(2s+1)/4}\right)$. \qed

\subsubsection{Results for $\widehat{q}_G(\tau;y)$}
\label{subsubsection:result_qG}

Let $r_n=(\log n)^{(2s+1)/4}$. From the asymptotic expansion established above, we obtain that
\begin{align*}
(\log n)^{(2s+1)/4}\left(\widehat{q}_G(\tau;y)-q_G(\tau;y)\right)=\underbrace{-(\log n)^{(2s+1)/4}\left[ \phi\left(\frac{y-q_0}{\sigma}\right)g(q_0) \right]^{-1}\Psi_{n}(q_0)}_{O_p(1)}+o_p(1) \\[7pt]
\end{align*}
\vspace{-5em}

\noindent which completes the proof, establishing both consistency and the convergence rate of $\widehat{q}_G(\tau;y)$.  \qed

\subsection{Proof of \Cref{ebci_coverage_lower_bound}}

\label{subsection:proof_minimax_coverage}

The proof reduces the lower bound problem for conditional coverage errors to a multiple testing problem. We construct a collection of alternative deconvolution densities $\{g_{m}\}_{m=1}^M$ around the baseline $g_0$ such that the induced conditional coverage error has a specific ``step-function'' structure. Under this structure, if an interval had conditional coverage error uniformly smaller than the claimed rate, it could decode which alternative generated the data. The proof then shows that such decoding is impossible because the alternatives are statistically indistinguishable.

To formalize this argument, we fix a baseline density $g_0\in\mathcal{G}(s,L/2)$ introduced in Appendix \ref{subsection:proof_rate_posterior_quantile}. Then define
\begin{equation*}
I(\beta;y):=\Big{[} q_{G_0}(\beta;y) \ , \ q_{G_0}(\beta+1-\alpha;y) \Big{]}, \quad \quad \text{for} \ \beta\in[0,\alpha]
\end{equation*}
so that $\operatorname{P}_{G_0}\left( \theta\in I(\beta;y)\mid Y=y\right)=1-\alpha$. Thus every posterior interval under $G_0$ with conditional coverage $1-\alpha$ can be indexed by its left-tail probability $\beta\in[0,\alpha]$. The restriction $\beta\in[0,\alpha]$ ensures that the right endpoint is a valid posterior quantile since $\beta+1-\alpha\in[0,1]$. For a general mixing distribution $G$, write $f_{G}(y)=\int \phi(\frac{y-\theta}{\sigma})\mathrm{d}G(\theta)$ for the marginal density of $Y$ evaluated at $y$.

The proof has four steps. Step 1 constructs the alternatives, Step 2 shows that their corresponding data distributions are statistically indistinguishable, Step 3 shows that a single interval can have conditional coverage error at the claimed rate only over a local window of alternatives, and Step 4 converts this into a multiple testing problem and applies Fano's inequality.

\subsubsection{Step 1. Constructing many alternatives}
\label{subsubsection:minimax_coverage_step1}

Step 1 constructs a collection of mixing densities $\{g_{m}\}_{m=1}^M\subset\mathcal{G}(s,L)$ that are well separated in the target criterion: under alternative $m$, the associated conditional coverage error changes sign across an $m$-specific transition window of width $O(T^{-1})$, and outside that window it is uniformly separated from zero by an order of $T^{-s-1/2}$. Because of this uniform separation, if an interval has conditional coverage error of the order $T^{-s-1/2}$, then under alternative $m$ its left endpoint must lie inside the corresponding transition window.

To index these alternatives, we first choose a compact interval $\mathcal{J}\subset(0,\alpha)$ such that $\mathrm{dist}(\mathcal{J},\mathcal{J}+1-\alpha)=\inf_{\beta\in\mathcal{J},\beta'\in\mathcal{J}+1-\alpha} |\beta-\beta'|>0$. Next, choose $\Delta>0$ so that \Cref{lemma:sign_function_property,lem:oscillatory-sums,lem:localized-perturbations} apply on $\mathcal{J}$, and let $\beta_1$ be the left-endpoint of $\mathcal{J}$. We define a grid of points $\beta_m:=\beta_1+\Delta(m-1)/T$ for $m=1,...,M$, where $M$ is the largest integer such that $\beta_M\in \mathcal{J}$. Because $\mathcal{J}$ has a fixed positive length, we have $M\asymp T$.

We now define the alternatives. For each $m=1,...,M$, let $g_{m}:=g_0+\bar{H}_T+\widetilde{H}_{m,T}$ where $g_0$ is the baseline Cauchy density, $\bar{H}_T$ is a common perturbation defined in Appendix \ref{subsubsection:construction_common}, and $\tilde{H}_{m,T}$ is an $m$-specific cumulative perturbation that determines the location of the sign change inside $\mathcal{J}$. Following Appendix \ref{subsection:proof_rate_posterior_quantile}, we set $\theta_m:=q_{G_0}(\beta_m;y)$, and $H_{m,T}:=\nu T^{-s+1/2}h(T(\theta-\theta_m))$ whose Fourier transform is $H_{m,T}^{\star}(t)=\nu T^{-s-1/2}(-i\,\text{sgn}(t))\exp(it\theta_m)k^{\star}(|t|/T)$. We then define the cumulative perturbations
\begin{equation*}
\widetilde{H}_{m,T}:=\varepsilon\left( \sum_{j<m}H_{j,T}-\sum_{j\geq m} H_{j,T} \right)
\end{equation*}
where $\varepsilon>0$ is sufficiently small. Thus each alternative is obtained by adding to the baseline density a common perturbation $\bar H_T$ and an $m$-dependent signed cumulative perturbation $\widetilde H_{m,T}$.

As verified in Appendix \ref{subsubsection_verifying_alternative}, these alternatives are admissible densities within the target function class; that is, $g_{m}\in\mathcal{G}(s,L)$ for all $m=1,\dots,M$ whenever $T$ is sufficiently large.

We now define the conditional coverage errors under each alternative. For $\beta\in[0,\alpha]$, let $D_{m}(\beta;y)=\operatorname{P}_{G_{m}}\left(\theta\in I(\beta;y)\mid Y=y\right)-(1-\alpha)$. To characterize where the sign of $D_{m}(\beta;y)$ can change, partition $\mathcal{J}$ into cells
\[
B_k=
\left[\beta_k-\frac{\Delta}{2T}\ , \ \beta_k+\frac{\Delta}{2T}\right]\cap\mathcal J, \quad \quad k=1,...,M.
\]
so that $\mathcal{J}=\bigcup_{k=1}^{M} B_k$.

The key output of Step 1 is a sign-separation property for $D_{m}(\beta;y)$: away from a uniformly bounded neighborhood of the index $m$, the sign of $D_{m}(\beta;y)$ is fixed and its magnitude is of order $T^{-s-1/2}$. Specifically, Appendix \ref{subsubsection_sign_separation} shows that there exist constants $c_1>0$ and $r\geq 1$, independent of $m$ and $T$, such that for all sufficiently large $T$, the error $D_{m}(\beta;y)$ is uniformly positive on the left-hand cells:
\begin{equation}
D_{m}(\beta;y)\ge c_1T^{-s-1/2}
\qquad
\text{for }\beta\in \bigcup_{k\le m-r}B_k,
\label{equ:minimax_Dmt_left}
\end{equation}
and uniformly negative on the right-hand cells:
\begin{equation}
D_{m}(\beta;y)\le -c_1T^{-s-1/2}
\qquad
\text{for }\beta\in \bigcup_{k\geq m+r}B_k.
\label{equ:minimax_Dmt_right}
\end{equation}
Consequently, its sign changes only within the transition region:
\[
\mathcal{T}_{m}:=\bigcup_{m-r<k<m+r} B_k.
\]
Then $\beta\in \mathcal{J}\setminus \mathcal{T}_{m}$ implies that $|D_{m}(\beta;y)|\geq c_1T^{-s-1/2}.$


For each $m$, $D_{m}(\beta;y)$ changes sign across an $m$-specific transition region $\mathcal{T}_{m}$ of uniform width $O(T^{-1})$. Outside of $\mathcal{T}_{m}$, it is uniformly bounded away from zero by the order of $T^{-s-1/2}$. As a result, the family $\{D_{m}(\cdot;y)\}_{m=1}^M$ behaves like a collection of step functions with a transition region of uniformly bounded width.


\subsubsection{Step 2. The alternatives are statistically close}

\label{subsubsection:minimax_coverage_step2}


We now show that the alternatives constructed in Step 1 are statistically close. First, we bound the Kullback-Leibler divergence between the laws of $\mathbf{Y}_{-i}$ under different alternatives. Second, because the target of the lower bound is the conditional coverage probability at a fixed $y$, we show that the corresponding posterior distributions are also uniformly close.

We use two notions of divergence: the Kullback-Leibler divergence, defined in \Cref{definition:KL_divergence}, and the $\chi^2$ divergence, introduced in \Cref{prop:A2_two_point}.

\begin{definition}[Kullback-Leibler divergence, Definition 2.5, \protect\citeA{tsybakov2008}]
The Kullback-Leibler (KL) divergence between two probability measures $P$ and $Q$ is defined by
\begin{equation*}
\mathrm{KL}\left(P,Q\right)=\begin{cases}
\int \log \frac{\mathrm{d}P}{\mathrm{d}Q}\mathrm{d}P,  \quad &\text{if} \ P\ll Q \\[7pt]
+\infty                                  ,   \quad &\text{otherwise}
\end{cases}
\end{equation*}
where $P\ll Q$ means that $P$ is absolutely continuous with respect to $Q$, and $\frac{\mathrm{d}P}{\mathrm{d}Q}$ is the Radon-Nikodym derivative.
\label{definition:KL_divergence}
\end{definition}

Let $\operatorname{P}_{m}$ denote the law of $\mathbf{Y}_{-i}=\{Y_j\}_{j\neq i}$ under $g_{m}$ and let $(f_{Y})_{m}$ be the corresponding marginal density. Since $\mathbf{Y}_{-i}$ consists of $n-1$ i.i.d. observations, the joint divergence tensorizes over the marginals. Using the inequality (2.27) in \citeA{tsybakov2008}, we obtain
\begin{align*}
\mathrm{KL}(\operatorname{P}_{m},\operatorname{P}_{m'})&=(n-1)\mathrm{KL}\Big{(}(f_{Y})_{m},(f_{Y})_{m'}\Big{)} \leq(n-1)\chi^2\Big{(}(f_{Y})_{m},(f_{Y})_{m'}\Big{)} \\[7pt]
&\leq 2(n-1)\int_{\mathbb R}
\frac{\big((f_Y)_{m}(y)-(f_Y)_{m'}(y)\big)^2}{(f_Y)_0(y)}\,\mathrm{d}y
\end{align*}
The last inequality holds because Step 1 gives $g_{m}\geq c'g_0$ so that $(f_{Y})_{m}(y)\geq c'(f_{Y})_0(y)$ for some constant $c'>0$.

Define $u_{m,m'}:=(f_Y)_{m}-(f_{Y})_{m'}=(g_{m}-g_{m'})*f_\varepsilon.$ By \Cref{lemma:step2_weighted_l2_bound}, we have that
\[
\int_{\mathbb R}\frac{u_{m,m'}(y)^2}{(f_Y)_0(y)}\,\mathrm{d}y
\le
C\varepsilon^2\nu^2 T^{-2s+2}\exp(-\sigma^2T^2).
\]
Choosing $T=T_n=\sigma^{-1}\sqrt{(1+\eta)\log n}$ as in Appendix~\ref{subsection:proof_rate_posterior_quantile}, we obtain
\[
(n-1)\int_{\mathbb R}
\frac{u_{m,m'}(y)^2}{(f_Y)_0(y)}\,\mathrm{d}y \leq C(\log n)^{-s+1}n^{-\eta}
\le C
\quad\quad\text{for all } \ m,m'.
\]

Therefore
\begin{equation}
\sup_{m,m'} \ \mathrm{KL}\left(\operatorname{P}_{m},\operatorname{P}_{m'}\right)\le C.
\label{equ:minimax_A.11}
\end{equation}
Thus the full leave-one-out sample $\mathbf{Y}_{-i}$ cannot distinguish the alternatives well.

We next give an analogous argument for errors in conditional coverage probability at the fixed value $y$. Let $F_G(\cdot \mid y)$ denote the posterior distribution of $\theta_i$ given $Y_i=y$:
\begin{equation*}
F_G(x \mid y)=\cfrac{ \int \mathbf{1}\{\theta\leq x\}\phi\left(\frac{y-\theta}{\sigma}\right)\mathrm{d}G(\theta) }{ \int \phi\left(\frac{y-\theta}{\sigma}\right)\mathrm{d}G(\theta)}.
\end{equation*}
By \Cref{lemma:step2_posterior_closedness}, the posterior distributions induced by the alternatives are uniformly close to the baseline posterior:
\begin{equation}
\sup_m \ \sup_{x\in\mathbb R} \ \left| F_{G_{m}}(x \mid y)-F_{G_0}(x\mid y) \right|\leq CT^{-s-1/2}.
\label{equ:minimax_A.12}
\end{equation}
Consequently, let $I=[\ell,u]\subset\mathbb{R}$ be any interval. Then $\operatorname{P}_G(\theta\in I \mid Y=y)=F_G(u\mid y)-F_G(\ell\mid y)$, so we have that
\begin{equation}
\Big{|}\operatorname{P}_{G_{m}}(\theta\in I \mid Y=y)-\operatorname{P}_{G_0}(\theta \in I \mid Y=y)\Big{|}\leq 2CT^{-s-1/2}.
\label{equ:minimax_A.13}
\end{equation}
Hence, at the fixed value $y$, the posterior mass assigned to any interval deviates from the baseline by at most $O(T^{-s-1/2})$ across the alternatives.


\subsubsection{Step 3. Accuracy under an alternative forces localization}

This step turns the sign pattern from Step 1 into a localization statement for any candidate interval. Fix an arbitrary interval $I=[\ell,u]$. We show that if $I$ has a conditional coverage error of order at most $T^{-s-1/2}$ under alternative $m$, the corresponding baseline tail index $\beta$ must lie in the $m$-specific transition region $\mathcal{T}_{m}$. Since these transition regions have uniformly bounded widths and are eventually disjoint when the indices are sufficiently separated, a single interval can have a conditional coverage error of $O(T^{-s-1/2})$ for only finitely many alternatives.

Let $\beta_0(I):=\operatorname{P}_{G_0}(\theta\in(-\infty,\ell] \mid Y=y)$. Suppose that $|\operatorname{P}_{G_{m}}(\theta \in I\mid Y=y)-(1-\alpha)|\leq AT^{-s-1/2}$ for some $m$. Then by \eqref{equ:minimax_A.13}, setting $C_1(A)=A+2C$, we have
\begin{equation}
\big|\operatorname{P}_{G_0}(\theta\in I\mid Y=y)-(1-\alpha)\big|\le C_1(A)T^{-s-1/2}.
\label{equ:minimax_A.15}
\end{equation}

Define $\beta(I):=\beta_0(I)\wedge \alpha$ and  $\widetilde I:=I(\beta(I),y)$. We first claim that
\begin{equation}
0\le \beta_0(I)-\beta(I)\le C_1(A)T^{-s-1/2}.
\label{equ:minimax_A.16}
\end{equation}
Because $I=[\ell,u]$ is an interval, we have that $\operatorname{P}_{G_0}(\theta\in I\mid Y=y)\le \operatorname{P}_{G_0}(\theta\geq \ell \mid Y=y)=1-\beta_0(I)$. Combining this with \eqref{equ:minimax_A.15} gives $1-\beta_0(I)\ge 1-\alpha-C_1(A)T^{-s-1/2}$, which is equivalent to \eqref{equ:minimax_A.16}.

Let $F_{G_0}(\cdot\mid y)$ denote the posterior distribution under $G_0$, and write
$\widetilde I=[\widetilde\ell,\widetilde u]$. By construction, $F_{G_0}(\ell\mid y)=\beta_0(I)$ and $F_{G_0}(\widetilde\ell\mid y)=\beta(I)$. Furthermore, $ F_{G_0}(u\mid y)=\beta_0(I)+\operatorname{P}_{G_0}(\theta\in I\mid Y=y)$ and $F_{G_0}(\widetilde u\mid y)=\beta(I)+(1-\alpha)$. The continuity of $F_{G_0}(\cdot\mid y)$ yields
\begin{align*}
\operatorname{P}_{G_0}\left(\theta\in I\triangle\widetilde I \mid Y=y\right)
& \ \leq
\left|F_{G_0}(\ell\mid y)-F_{G_0}(\widetilde\ell\mid y)\right|
+
\Big{|}F_{G_0}(u\mid y)-F_{G_0}(\widetilde u\mid y)\Big{|} \\[7pt]
&\ \le
2|\beta_0(I)-\beta(I)|
+
\big|\operatorname{P}_{G_0}(\theta\in I\mid Y=y)-(1-\alpha)\big| \\[7pt]
&\ \le 3C_1(A)T^{-s-1/2}.
\end{align*}
Hence, by \eqref{equ:minimax_A.12}, $\operatorname{P}_{G_{m}}(\theta\in I\triangle\widetilde I\mid Y=y)\le (3C_1(A)+4C)T^{-s-1/2}.$ Combined with $|P_{G_{m}}(\theta \in I \mid Y=y)-(1-\alpha)|\leq AT^{-s-1/2}$, the triangle inequality implies $|\operatorname{P}_{G_{m}}(\theta\in\widetilde{I}\mid Y=y)-(1-\alpha)|\le (4A+10C)T^{-s-1/2}.$ Since $\widetilde I=I(\beta(I);y)$, setting $C_2(A):=4A+10C$, this becomes
\begin{equation}
\big{|}D_{m}(\beta(I);y)\big{|}\le C_2(A)T^{-s-1/2}.
\label{equ:minimax_A.18}
\end{equation}

We next show that $\beta(I)$ must lie in $\mathcal{J}$. If $\beta(I)\notin \mathcal{J}$, then \Cref{lemma:common_perturbation} (ii) together with the decomposition $D_{m}=\bar{D}_T+\Gamma_{\tilde{H}_{m,T}}+\tilde{R}_{m}$ implies that
\begin{align*}
|D_{m}(\beta(I);y)|&\geq |\bar{D}_T(\beta(I);y)|-|\Gamma_{\tilde{H}_{m,T}}(\beta(I);y)|-|\tilde{R}_{m}(\beta(I);y)| \\[7pt]
&\geq \left(c\eta-C\varepsilon \nu-o(1)\right) T^{-s-1/2}.
\end{align*}
Choose $\varepsilon>0$ sufficiently small so that $C\varepsilon\nu\leq c\eta/4$. Then for all sufficiently large $T$, $|D_{m}(\beta(I);y)|\geq (c\eta/2)T^{-s-1/2}$ which contradicts \eqref{equ:minimax_A.18} whenever $C_2(A)<c\eta/2$. Hence $\beta(I)\in\mathcal{J}$.


Once $\beta\in \mathcal{J}$, the sign-separation argument from Step 1 applies. In particular, if $\beta\in \mathcal{J}\setminus \mathcal{T}_{m}$, it follows that $|D_{m}(\beta;y)|\geq c_1T^{-s-1/2}$ where $c_1$ is the constant from
(\ref{equ:minimax_Dmt_left})-(\ref{equ:minimax_Dmt_right}). Now choose $A>0$ such that $C_2(A)<c_1$, which ensures $\beta(I)\in\mathcal{T}_{m}$. Thus, whenever $I$ has a conditional coverage error of order $AT^{-s-1/2}$ under alternative $m$, its corresponding baseline tail index $\beta(I)$ must lie in the $m$-specific transition region $\mathcal{T}_{m}$.

Finally, we show that any fixed interval $I$ can satisfy this coverage error bound for at most $2r+1$ alternatives. If $|m-m'|>2r+1$, the two transition regions are disjoint. To see this, assume without loss of generality that $m'>m+2r+1$. Then
\begin{equation*}
\beta_{m'+r}-\frac{\Delta}{2T}-\left(\beta_{m+r}+\frac{\Delta}{2T}\right)=\frac{\Delta}{T}\left(m'-m-2r-1\right)>0
\end{equation*}
so $\mathcal{T}_{m}\cap \mathcal{T}_{m'}=\varnothing$. Thus the right endpoint of $\mathcal{T}_{m}$ lies strictly to the left of the left endpoint of $\mathcal{T}_{m'}$. Since $\beta(I)$ is a single point, it can belong to at most $2r+1$ transition regions.

\subsubsection{Step 4. Reduction to multiple testing}

This step completes the proof by using a local Fano argument. First, we construct a finite local packing of alternatives. Choose a subset $\mathcal{M}_T\subset\{1,...,M\}$ such that $|m-m'|>2r+1$ for all distinct $m,m'\in\mathcal{M}_T$. Because $r$ is fixed and $M\asymp T$, such a subset can be chosen with $|\mathcal{M}_T|\asymp T$.

For each $m\in\mathcal{M}_T$, we define the errors in conditional coverage probability of a realized interval $I$ under $G_m$ by $\mathrm{ECP}_{m}(I;y)=\operatorname{P}_{G_{m}}\left( \theta_i\in I \mid Y_i=y \right)-(1-\alpha) $. We define the $AT^{-s-1/2}$-accuracy set under $G_m$ by
\begin{equation*}
\mathcal{A}_m=\left\{ I: |\mathrm{ECP}_{m}(I;y)|\leq  AT^{-s-1/2} \right\}.
\end{equation*}
Step 3 implies the following separation property of the finite packing:
\begin{equation*}
\mathcal{A}_m\cap \mathcal{A}_{m'}=\varnothing, \quad\quad \text{for all distinct} \ m,m'\in\mathcal{M}_T.
\end{equation*}
Thus the family $\{G_m: m\in\mathcal{M}_T\}$ is separated in the sense that no single realized interval can lie in the $AT^{-s-1/2}$-accuracy sets of two distinct alternatives in $\mathcal{M}_T$.


We now associate each interval rule with a decoder $\hat{m}=\hat{m}(\mathrm{CI}_i(\mathbf{Y}_{-i}))$: if $\mathrm{CI}_i(\mathbf{Y}_{-i})\in\mathcal{A}_m$ for some $m\in\mathcal{M}_T$, set $\hat{m}=m$; otherwise, set $\hat{m}$ to be an arbitrary fixed element of $\mathcal{M}_T$. This map is well-defined because the collection $\{\mathcal{A}_m: m\in\mathcal{M}_T\}$ is pairwise disjoint.



Now fix $m\in\mathcal{M}_T$. If the true alternative is $m$ and $\hat{m}\neq m$, then the realized interval cannot be $AT^{-s-1/2}$-accurate under $G_{m}$. Hence
\begin{equation*}
\operatorname{P}_{m}\left( \hat{m}\neq m \right)\leq \operatorname{P}_{m}\left( \big{|}\mathrm{ECP}_{m}(\mathrm{CI}_i(\mathbf{Y}_{-i});y) \big{|}>AT^{-s-1/2}  \right).
\end{equation*}
Using the inequality $x^2\geq a^2\mathbf{1}\{|x|>a\}$, we obtain
\begin{equation*}
\operatorname{E}_{m}\Big{[}\mathrm{ECP}_{m}(\mathrm{CI}_i(\mathbf{Y}_{-i});y)^2\Big{]}\geq A^2T^{-2s-1}\operatorname{P}_{m}\left(\hat{m}\neq m\right).
\end{equation*}
Taking the supremum over $m\in\mathcal{M}_T$ yields
\begin{equation}
\sup_{m\in\mathcal{M}_T} \ \operatorname{E}_{m}\Big{[}\mathrm{ECP}_{m}(\mathrm{CI}_i(\mathbf{Y}_{-i});y)^2\Big{]} \geq A^2T^{-2s-1}\sup_{m\in\mathcal{M}_T} \operatorname{P}_{m}\left(\hat{m}\neq m\right).
\label{equ:minimax_A.30}
\end{equation}
Let $\hat{\mathscr{M}}_T$ denote the class of decoders induced by interval rules in $\mathscr{C}$, and let $\mathscr{M}_T$ denote the class of all measurable decoders mapping $\mathbf{Y}_{-i}$ to $\mathcal{M}_T$. Because $\hat{\mathscr{M}}_T\subseteq\mathscr{M}_T$, we have
\begin{equation*}
\inf_{\hat{m}\in\hat{\mathscr{M}}_T} \ \sup_{m\in\mathcal{M}_T} \ \operatorname{P}_{m}\left(\hat{m}\neq m\right)\geq \inf_{\hat{m}\in\mathscr{M}_T} \ \sup_{m\in\mathcal{M}_T} \ \operatorname{P}_{m}\left(\hat{m}\neq m\right)
\end{equation*}

Taking the infimum in \eqref{equ:minimax_A.30} over all interval rules in $\mathscr{C}$ produces the lower bound
\begin{equation}
\inf_{\mathrm{CI}_i\in\mathscr{C}} \ \sup_{m\in\mathcal{M}_T} \ \operatorname{E}_{m}\Big{[}\mathrm{ECP}_{m}(\mathrm{CI}_i(\mathbf{Y}_{-i});y)^2\Big{]} \geq A^2T^{-2s-1} \inf_{\hat{m}\in\hat{\mathscr{M}}_T} \ \sup_{m\in\mathcal{M}_T} \ \operatorname{P}_{m}\left(\hat{m}\neq m\right).
\label{equ:testing_error}
\end{equation}

It remains to lower bound the testing error probability on the right-hand side of (\ref{equ:testing_error}). By Fano's inequality (Theorem 1, \citeNP{scarlett2019introductory}), we have that
\begin{align*}
\sup_{m\in\mathcal{M}_T} \ \operatorname{P}_{m}\left(\hat{m}\neq m\right)\geq \frac{1}{|\mathcal{M}_T|}\sum_{m\in\mathcal{M}_T} \operatorname{P}_m\left(\hat{m}\neq m\right)\geq 1-\cfrac{I(V;\mathbf{Y}_{-i})+\log 2}{\log |\mathcal{M}_T|}
\end{align*}
where $I(V;\mathbf{Y}_{-i})$ is the mutual information between the random index $V\sim \text{Unif}(\mathcal{M}_T)$ and the observed data $\mathbf{Y}_{-i}$, that is defined and upper bounded as follows:
\begin{align*}
I(V;\mathbf{Y}_{-i})&:=\mathrm{KL}\Big{(} \operatorname{P}_{V,\mathbf{Y}_{-i}}, \operatorname{P}_{V}\otimes \operatorname{P}_{\mathbf{Y}_{-i}}  \Big{)}\overset{(a)}{=}\frac{1}{|\mathcal{M}_T|}\sum_{m\in\mathcal{M}_T}\mathrm{KL}\left( \operatorname{P}_{\mathbf{Y}_{-i}|V=m}, \operatorname{P}_{\mathbf{Y}_{-i}} \right)  \\[7pt]
&\overset{(b)}{=}\frac{1}{|\mathcal{M}_T|}\sum_{m\in\mathcal{M}_T} \mathrm{KL}\left(\operatorname{P}_m, \operatorname{P}_{\mathbf{Y}_{-i}}\right), \quad \quad \text{here} \ \operatorname{P}_{\mathbf{Y}_{-i}}=\frac{1}{|\mathcal{M}_T|}\sum_{m'\in\mathcal{M}_T}\operatorname{P}_{m'}\\[7pt]
&\overset{(c)}{\leq} \frac{1}{|\mathcal{M}_T|^2}\sum_{m,m'\in\mathcal{M}_T}\mathrm{KL}\left(\operatorname{P}_m,\operatorname{P}_{m'}\right)\leq \sup_{m,m'\in\mathcal{M}_T} \ \mathrm{KL}\left(\operatorname{P}_{m},\operatorname{P}_{m'}\right)\overset{(d)}{\le} C
\end{align*}
where (a) follows from Lemma 4, Eq. (23) of \citeA{scarlett2019introductory}, (b) follows because $\mathbf{Y}_{-i}|V=m\sim \operatorname{P}_m$, (c) is obtained from Lemma 4, Eq. (26) of \citeA{scarlett2019introductory}, (d) follows from the uniform bound on pairwise Kullback-Leibler divergences (\ref{equ:minimax_A.11}).



Then we have that
\begin{equation}
\inf_{\hat{m}\in\mathscr{M}_T} \ \sup_{m\in\mathcal{M}_T} \ \operatorname{P}_{m}\left(\hat{m}\neq m\right) \geq 1-\cfrac{C+\log 2}{\log |\mathcal{M}_T|}.
\label{equ:minimax_A.31}
\end{equation}
Because $|\mathcal{M}_T|\asymp T\rightarrow\infty$, the right-hand side of (\ref{equ:minimax_A.31}) is bounded away from zero; in particular, it is at least $1/2$ for all sufficiently large $T$.

 Combining (\ref{equ:minimax_A.30})-(\ref{equ:minimax_A.31}), and using $\{g_{m}: m\in\mathcal{M}_T\}\subset \mathcal{G}(s,L)$, we conclude that
\[
\inf_{\mathrm{CI}_i\in\mathscr{C}}\sup_{g\in\mathcal G(s,L)}
E\Big[\big(\operatorname{P}_G(\theta_i\in \mathrm{CI}_i(\mathbf{Y}_{-i})\mid Y_i=y)-(1-\alpha)\big)^2\Big]
\ge cT^{-2s-1}
\]
for some constant $c>0$. Letting $T\asymp \sqrt{\log n}$ proves the theorem.

\subsection{Proof of \Cref{ebci_condition_coverage}}

Fix $Y_i=y$ and $\tau\in(0,1)$, write $\widehat{q}=\widehat{q}_G(\tau;y)$ and $q_0=q_{G}(\tau;y)$. First we consider the conditional coverage probability $\operatorname{P}(\theta\leq \widehat{q} \ |\ Y_i=y)$ for a fixed $\tau\in(0,1)$.

Because the posterior distribution $\operatorname{P}(\theta\leq u|Y=y)$ is differentiable in a neighborhood around $q_0$ with the posterior density (\ref{equ:posterior_density}), the mean-value theorem gives that
\begin{equation}
\operatorname{P}\left(\theta\leq \widehat{q} \ | \ Y_i=y\right)=\operatorname{P}\left(\theta\leq q_0 \ | \ Y_i=y\right)+\pi(q^{*}|Y_i=y)\left(\widehat{q}-q_0\right)
\label{equ:appendix_cond_coverage}
\end{equation}
where $q^{*}$ lies between $\widehat{q}$ and $q_0$.  Notice that since $\widehat{q}\p q_0$ by \Cref{quantile_estimator} and $|q^{*}-q_0|\leq |\widehat{q}-q_0|$, it must be that $q^{*}\p q_0$. By \Cref{assumption1} and $s>1/2$, the mixing density $g(\cdot)$ and thus $\pi(\cdot|Y=y)$ is continuous, the continuous mapping theorem yields that $\pi(q^{*}|Y_i=y)\p \pi(q_0|Y_i=y)$. By $\operatorname{P}\left(\theta\leq q_0 \ | \ Y_i=y\right)=\tau$ and \Cref{quantile_estimator}, we obtain that
\begin{align*}
\operatorname{P}\left(\theta_i\leq \hat{q}_G(\tau;Y_i) \ | \ Y_i=y \right)&=\tau+\Big{(}\pi(q_0|Y_i=y)+o_p(1)\Big{)}\left(\widehat{q}-q_0\right) \\[7pt]
&=\tau+O_p\left((\log n)^{-(2s+1)/4}\right).
\end{align*}
Applying this result with $\tau=\alpha/2$ and $\tau=1-\alpha/2$ obtains the desired results.

\subsection{Proof of \Cref{ebci_coverage}}

Only in this subsection, for $\tau\in (0,1)$, we write $q(y):=q_G(\tau;y)$, $\hat{q}_{-i}(y):=\hat{q}_G(\tau;y)$, where $\hat{q}_{-i}(y)$ is computed from the leave-one-out sample $\mathbf{Y}_{-i}=\{Y_j\}_{j\neq i}$. Also let $r_n:=(\log n)^{(2s+1)/4}$. It suffices to prove that
\begin{equation*}
\operatorname{P}\Big{(} \theta_i\leq \hat{q}_{-i}(Y_i) \Big{)}=\tau
+O\!\left((\log n)^{-(2s+1)/4}\right).
\end{equation*}

Applying this with $\tau=\alpha/2$ and $\tau=1-\alpha/2$ gives the desired result for
\begin{equation*}
\mathrm{CI}_i^{\mathrm{NP}}=\Big{[}\hat{q}_{G}(\alpha/2;Y_i)\ , \ \hat{q}_{G}(1-\alpha/2;Y_i)\Big{]}.
\end{equation*}

By iterated expectations and the independence of $\mathbf{Y}_{-i}$ and $(\theta_i,Y_i)$,
\begin{equation}
\begin{aligned}
\operatorname{P}\left( \theta_i\leq \hat{q}_{-i}(Y_i) \right)-\tau &=\operatorname{E}\left[\operatorname{P}\left( \theta_i\leq \hat{q}_{-i}(Y_i) \mid Y_i,\mathbf{Y}_{-i} \right)-\operatorname{P}\left( \theta_i\leq q(Y_i) \mid Y_i \right)\right] \\[7pt]
&=\operatorname{E}\left[\operatorname{P}\left( \theta_i\leq \hat{q}_{-i}(Y_i) \mid Y_i \right)-\operatorname{P}\left( \theta_i\leq q(Y_i) \mid Y_i \right)\right].
\end{aligned}
\label{eq:iterated-marginal}
\end{equation}
Therefore, by the mean-value theorem, there exists $q^{*}(y)$ between $q(y)$ and $\hat{q}_{-i}(y)$ such that
\begin{equation}
\operatorname{P}\left( \theta_i\leq \hat{q}_{-i}(Y_i) \right)-\tau=\operatorname{E}\Big{[} \pi(q^{*}(Y_i)|Y_i)\left(\hat{q}_{-i}(Y_i)-q(Y_i)\right) \Big{]}
\label{eq:mean_value_theorem}
\end{equation}
Next we use the same argument as in Appendix \ref{appendix:proof_theorem3_2}, but uniformly over $y\in\mathcal{Y}$. Specifically \Cref{assumption_uniform_bound_posterior_density} gives that
\begin{equation*}
\phi\!\left(\frac{y-q(y)}{\sigma}\right)g(q(y))\ge \underline{\pi}
\qquad\text{for all }y\in\mathcal{Y}
\end{equation*}
and \Cref{prop:uniform_consist} gives the uniform consistency of $\hat{g}$. Hence the proof of the asymptotic expansion in Appendix \ref{appendix:proof_theorem3_2} yields
\begin{equation}
\sup_{y\in\mathcal{Y}} \ \left| r_n\left(\hat{q}_{-i}(y)-q(y)\right)+r_n\left[ \phi\left(\frac{y-q(y)}{\sigma}\right)g(q(y)) \right]^{-1}\Big{(} \Psi_{n1,-i}(y)+\Psi_{n2}(y) \Big{)} \right|=o_p(1)
\label{equ:uniform_equ_1}
\end{equation}
where $\Psi_{n1,-i}(y)$ and $\Psi_{n2}(y)$ are the stochastic and bias terms from the decomposition in Appendix \ref{appendix:proof_theorem3_2} with $\widehat f_Y^\star$ formed from $Y_{-i}$.

Moreover, the same argument also gives that $\sup_{y\in\mathcal{Y}} |\hat{q}_{-i}(y)-q(y)|=o_p(1)$. Since $|q^{*}(y)-q(y)|\leq |\hat{q}_{-i}(y)-q(y)|$, it follows that $\sup_{y\in\mathcal{Y}} |q^{*}(y)-q(y)|=o_p(1)$. Because $g$ is continuous by \Cref{assumption2}, the posterior density $\pi(\cdot|Y=y)$ is continuous in its first argument. Hence we have that
\begin{equation}
\sup_{y\in\mathcal{Y}} \ \Big{|} \pi\left(q^{*}(y) \mid Y=y\right)-\pi\left( q(y) \mid Y=y \right) \Big{|}=o_p(1).
\label{equ:uniform_equ_2}
\end{equation}

Substituting \eqref{equ:uniform_equ_1} and \eqref{equ:uniform_equ_2} into \eqref{eq:mean_value_theorem}, we obtain that
\begin{equation}
\operatorname{P}\Big{(} \theta_i\leq \hat{q}_{-i}(\tau ;Y_i) \Big{)}-\tau
=
A_{n1}+A_{n2}+o\!\left(r_n^{-1}\right),
\label{eq:An-decomposition}
\end{equation}
where
\begin{equation*}
A_{n1}
:=
\operatorname{E}\!\left[
\frac{\pi(q(Y_i)\mid Y_i)}
{\phi\!\left(\frac{Y_i-q(Y_i)}{\sigma}\right)g(q(Y_i))}
\Psi_{n1,-i}(Y_i)
\right], \quad
A_{n2}
:=
\operatorname{E}\!\left[
\frac{\pi(q(Y_i)\mid Y_i)}
{\phi\!\left(\frac{Y_i-q(Y_i)}{\sigma}\right)g(q(Y_i))}
\Psi_{n2}(Y_i)
\right].
\end{equation*}

We first consider $A_{n1}$. For every fixed $y$, $\operatorname{E}\!\left[\Psi_{n1,-i}(y)\mid Y_i=y\right]=0,$ because $\Psi_{n1,-i}(y)$ is the centered empirical-process term built from $\mathbf{Y}_{-i}$, and $\mathbf{Y}_{-i}$ is independent of $Y_i$. Hence $A_{n1}=0$.

We next consider $A_{n2}$. \Cref{assumption_uniform_bound_posterior_density} implies that
\begin{equation}
|A_{n2}|\leq \underline{\pi}^{-1}\operatorname{E}\left[ \pi(q(Y_i)\mid Y_i) |\Psi_{n2}(Y_i)| \right].
\label{equ:lower_bound_AN2}
\end{equation}

Now, the same argument used to analyze $\Psi_{n2}(q_0)$ in Appendix \ref{appendix:proof_theorem3_2} applies uniformly over $y\in\mathcal{Y}$ because \Cref{propo:fourier_decay} gives the $t^{-1}$ bound on the Fourier coefficient of the weighting function uniformly over $y\in\mathcal{Y}$. Therefore
\begin{equation}
\sup_{y\in\mathcal{Y}} \ |\Psi_{n2}(y)|
=
O\!\left(h_n^{\,s+1/2}\right).
\label{eq:uniform-bias-bound}
\end{equation}
Combining \eqref{equ:lower_bound_AN2} and \eqref{eq:uniform-bias-bound}, we have
\[
|A_{n2}|
\le
\underline{\pi}^{-1}
\sup_{y\in\mathcal{Y}}|\Psi_{n2}(y)|
\int \pi(q(\tau;y)\mid Y_i=y)\,dF_Y(y).
\]
By \Cref{assumption_expected_density}, the integral on the right-hand side is finite. Therefore $A_{n2}=O\!(h_n^{\,s+1/2}).$

Since $h_n\asymp (\log n)^{-1/2}$, we have that $h_n^{\,s+1/2}=O\!\left((\log n)^{-(2s+1)/4}\right).$ Substituting the results for $A_{n1}$ and $A_{n2}$ into (\ref{eq:An-decomposition}) gives the desired result. \qed

\section{Computational details of NP-EBCI}
\label{section:computation_npebci}

The practical implementation in \Cref{section:practical_implementation}, especially the bandwidth selection procedure in \Cref{subsection:bandwidth_choice}, can be computationally demanding. The procedure requires the computation of two posterior-quantile endpoints for each candidate bandwidth $h\in\mathcal{H}_n$, each fold $v$, and each held-out unit $i\in\mathcal{I}_v$. In large samples, this repeated evaluation can be costly.

The purpose of this section is therefore computational. We outline two strategies to speed up the computation. First, for each fold and each candidate bandwidth, the deconvolution objects are common to all held-out units and can therefore be computed once and reused throughout interval construction. Second, since the leave-fold-out NPMLE has a discrete structure, the coverage criterion in \Cref{subsection:bandwidth_choice} reduces to finite weighted sums over its support points. Consequently, the bandwidth search avoids repeated simulation or additional high-dimensional numerical integration.

Fix a fold $v\in\{1,...,V\}$ and a candidate bandwidth $h\in\mathcal{H}_n$. Define
\begin{equation*}
\hat{f}_Y^{\star(-v)}(t)=\frac{1}{|\mathcal{I}_v^{c}|}\sum_{j\in \mathcal{I}_v^c} \exp(itY_j), \quad \quad
\bar{f}_{\varepsilon}^{\star(-v)}(t)=\frac{1}{|\mathcal{I}_v^{c}|}\sum_{j\in \mathcal{I}_v^c} \exp\left(-\frac{1}{2}\sigma_j^2t^2\right),
\end{equation*}
and define the fold-specific deconvolution weight
\begin{equation*}
\hat{g}_{v,h}^{\star}(t):=\cfrac{\hat{f}_Y^{\star(-v)}(t) }{ \bar{f}_{\varepsilon}^{\star(-v)}(t) } K^{\star}(ht).
\end{equation*}
Then, for each held-out unit $i\in\mathcal{I}_v$, the posterior quantile estimator in Step 2 of \Cref{subsection:bandwidth_choice} can be written as
\begin{equation}
\hat{q}_h^{(-v)}\left(\tau;\mathcal{D}_i\right)=\underset{q}{\arg\min} \ \int \overline{ M^{\star}(q,t;\mathcal{D}_i) }\ \hat{g}_{v,h}^{\star}(t) \mathrm{d}t.
\label{equ:q_compute_integrand}
\end{equation}

A key computational feature is that $\hat{g}_{v,h}^{\star}(t)$ depends on the training sample $\mathcal{I}_v^c$ and on $h$, but not on the held-out unit $i$. Hence, within each fold $v$, the same deconvolution weight can be used for all units $i\in\mathcal{I}_v$. In implementation, we \textit{pre-compute} $\hat{g}_{v,h}^{\star}(t)$ once for each $(v,h)$ and then reuse it to evaluate the one-dimensional objective for all held-out units in that fold.

The remaining computational bottleneck is the repeated evaluation of $\overline{ M^{\star}(q,t;\mathcal{D}_i) }$ over many candidate values of $q$, many held-out units $i$, and many bandwidths $h$. We compute these transforms in batches. Specifically, we stack the integrands in (\ref{equ:q_compute_integrand}) for many held-out units into a single large matrix, rewrite the complex exponential terms in real arithmetic, and evaluate the resulting objective by matrix multiplication. This replaces repeated scalar quadrature problems with a small number of batched linear algebra operations. In practice, we evaluate the objective on a moderate grid of candidates $q$-values centered near $Y_i$, locate the minimizing grid point, and then apply a local quadratic refinement to obtain the posterior quantile estimates.

We next turn to Step 3 of \Cref{subsection:bandwidth_choice}. For each fold $v$, let the leave-fold-out NPMLE be
\begin{equation*}
\hat{G}^{(-v)}(A)=\sum_{\ell=1}^{L_v}\hat{p}_{\ell v} \mathbf{1}\{ \hat{\theta}_{\ell v} \in A \},\quad \quad \quad \hat{p}_{\ell v}\geq 0, \quad \sum_{\ell=1}^{L_v}\hat{p}_{\ell v}=1
\end{equation*}
for any Borel set $A\subset\mathbb{R}$. The discreteness of $\hat{G}^{(-v)}$ implies that, for each held-out unit $i\in \mathcal{I}_v$, the posterior distribution under $\hat{G}^{(-v)}$ is also discrete. Its posterior atom probabilities, given $\mathcal{D}_i=(Y_i,\sigma_i)$, are
\begin{equation*}
\hat{\pi}_{\ell v}(\mathcal{D}_i)=\cfrac{\hat{p}_{\ell v}\phi( (Y_i-\hat{\theta}_{\ell v})/\sigma_i )}{ \sum_{m=1}^{L_v} \hat{p}_{mv}\phi( (Y_i-\hat{\theta}_{mv})/\sigma_i ) }, \quad \quad \ell=1,...,L_v.
\end{equation*}
Therefore the average conditional coverage under the out-of-fold prior estimate $\hat{G}^{(-v)}(\cdot)$ is
\begin{equation*}
\hat{C}(h)=\frac{1}{n}\sum_{v=1}^{V}\sum_{i\in\mathcal{I}_v} \sum_{\ell=1}^{L_v} \hat{\pi}_{\ell v}(\mathcal{D}_i) \mathbf{1}\left\{ \hat{q}_h^{(-v)}(\alpha/2;\mathcal{D}_i) \leq \hat{\theta}_{\ell v}\leq \hat{q}_h^{(-v)}(1-\alpha/2;\mathcal{D}_i)\right\}
\end{equation*}
which is a finite weighted sum over the discrete support of the NPMLE. Hence, once the endpoints have been computed, evaluating $\hat{C}(h)$ requires no additional numerical integration.

Finally, for Step 4 of the bandwidth selection procedure in \Cref{subsection:bandwidth_choice}, one can replace the hard feasibility rule with a penalized criterion:
\begin{equation*}
\hat{h}=\underset{h\in\mathcal{H}_n}{\arg\min} \ \left\{ \cfrac{\hat{L}(h)}{ L_{\mathrm{naive}} }+\lambda\left( \max\{ (1-\alpha)-\hat{C}(h), 0 \} \right)^2 \right\},
\end{equation*}
where $L_{\mathrm{naive}}=\frac{1}{n}\sum_{i=1}^n 2z_{1-\alpha/2}\sigma_i$ is the average length of the naive $z$-interval. The soft penalty continues to favor short intervals among candidates whose estimated coverage is close to the nominal level. But unlike the hard-feasibility rule, it remains well defined even when no candidate bandwidth satisfies the nominal coverage requirement, and it is often more numerically stable in finite samples.

\section{Extended discussions}

\subsection{Monotonicity of posterior quantiles }
\label{appendix:posterior_quantile}

\begin{proposition}
Consider the model $Y_i=\theta_i+\varepsilon_i$ where $\theta_i\sim G(\cdot), \varepsilon_i\sim F_{\varepsilon}(\cdot)$ and $\theta_i\independent \varepsilon_i$. The posterior quantile $q_{G}(\tau;y)$ is non-decreasing in $y$ if $f_{\varepsilon}(\cdot)$ is log-concave on $\mathbb{R}$.
\label{prop:quantile_mlr}
\end{proposition}
\begin{proof}
Given $Y=y$, the posterior density of $\theta$ is $\pi(\theta|y)\propto g(\theta)f_{\varepsilon}(y-\theta)$. Since $f_{\varepsilon}$ is log-concave, the location family $\{f_{\varepsilon}(y-\theta)\}$ satisfies the monotone likelihood ratio property, i.e. for any $y_1<y_2$, the ratio $\theta \mapsto f_{\varepsilon}(y_2-\theta)/f_{\varepsilon}(y_1-\theta)$ is non-decreasing in $\theta$ (Proposition 2.3, \citeA{saumard2014log}). Therefore
\begin{equation*}
\frac{\pi(\theta|y_2)}{\pi(\theta|y_1)}=\cfrac{\int g(u)f_{\varepsilon}(y_1-u)\mathrm{d}u}{\int g(u)f_{\varepsilon}(y_2-u)\mathrm{d}u}\cdot \cfrac{f_{\varepsilon}(y_2-\theta)}{f_{\varepsilon}(y_1-\theta)}
\end{equation*}
is also non-decreasing in $\theta$.

Since both $\pi(\theta|y_1)$ and $\pi(\theta|y_2)$  integrate to one, the function $\pi(\theta|y_2)-\pi(\theta|y_1)$ can change sign at most once, and if it changes sign, it does so from negative to positive. Hence there exists $c$ such that $\pi(\theta|y_2)\leq \pi(\theta|y_1)$ for $\theta\leq c$ and $\pi(\theta|y_2)\geq \pi(\theta|y_1)$ for $\theta\geq c$. It follows that
\begin{equation*}
\operatorname{P}_G\left(\theta\leq u | Y=y_2\right)-\operatorname{P}_G\left(\theta\leq u | Y=y_1\right)=\int_{-\infty}^u \left( \pi(\theta|y_2)-\pi(\theta|y_1) \right)\mathrm{d}\theta\leq 0.
\end{equation*}
Thus, for each $u$, the posterior distribution $\operatorname{P}_{G}(\theta\leq u|Y=y)$ is non-increasing in $y$, which in turn implies that the posterior quantile $q_G(\tau;y)$ is non-decreasing in $y$.
\end{proof}

\subsection{Fourier-based kernel estimator for posterior means}
\label{appendix:posterior_mean}

In this subsection, we consider the Fourier-based kernel estimator for the posterior mean (\citeNP{zhang1997empirical}) analogous to the procedure (\ref{equ:estimator_quantile}) for the posterior quantile. Our purpose is not to propose an estimator for the posterior mean per se, but to facilitate a comparison of its rate of convergence with that of the posterior quantile.

Recall that the posterior mean $m_G(y)=\operatorname{E}_G[\theta|Y=y]:=m_0$. Let the estimator $\hat{m}_G(y)$ be \Cref{equ:estimator_quantile} but replace the check function with $\ell_2$ loss. Following the strategy in \Cref{subsubsection:expansion}, one can establish the following asymptotic expansion $r_n(\hat{m}-m_0)=-r_nf_Y^{-1}(y)\Psi_n(m_0)+o_p(1)$ with a sequence $r_n\rightarrow\infty$ as $n\rightarrow\infty$, that will be specified later. The leading term $\Psi_n(m_0)$ is
\begin{equation*}
\Psi_n(m_0)=(2\pi)^{-1} \int  \cfrac{\overline{\psi^{\star}(m_0,t;y)}}{f_{\varepsilon}^{\star}(t)} \left(\widehat{f}^{\star}_Y(t)K^{\star}(h_n t) \right) \mathrm{d}t
\end{equation*}
where the weighting function for the posterior mean is $\overline{\psi^{\star}(m_0,t;y)}:=\int \exp(-it\theta)\left( \theta-m_0\right)\phi\left(\frac{y-\theta}{\sigma}\right)\mathrm{d}\theta$.

We focus on the analysis of $\Psi_n(m_0)$ for the posterior mean and compare it with the corresponding term for the posterior quantile. First, we demonstrate that the weighting function $\overline{\psi^{\star}(m_0,t;y)}$ is able to cancel out the exponential decay of $f_{\varepsilon}^{\star}(t)$ in \Cref{prop:decay_rate_pm}.
\begin{proposition}
It calculates that $\overline{\psi^{\star}(m_0,t;y)}=\sigma\left[(y-m_0)-i\sigma^2 t\right]\exp\left(-ity-\frac{1}{2}\sigma^2 t\right)$.
\label{prop:decay_rate_pm}
\end{proposition}
Using the explicit expression in (\ref{prop:decay_rate_pm}), we observe that the Fourier-based kernel estimator for the posterior mean is directly linked to the classical kernel density estimator for $f_Y(y)$ and its derivative. In particular, the leading term can be rewritten as follows:
\begin{align*}
\Psi_n(m_0)&=(2\pi)^{-1} \int \sigma\left(y-m_0-i\sigma^2 t\right)\exp\left(-ity\right)  \left(\widehat{f}^{\star}_Y(t)K^{\star}(h_n t) \right) \mathrm{d}t \\[7pt]
&=\frac{\sigma}{2\pi n}\sum_{j=1}^n \int \left( y-m_0-i\sigma^2t \right)\exp\left(-it(y-Y_j)\right)K^{\star}(h_nt)t\\[7pt]
&\overset{(a)}{=}\frac{\sigma}{2\pi nh_n}\sum_{j=1}^n \int \left(y-m_0-i\sigma^2\frac{u}{h_n}\right)\exp\left( -iuh_n^{-1}(y-Y_j) \right)K^{\star}(u)\mathrm{d}u \\[7pt]
&=\frac{\sigma}{2\pi n} \sum_{j=1}^n \left[ (y - m_0) \int \exp\left(-iu\frac{y - Y_j}{h_n}\right) K^{\star}(u) \, \mathrm{d}u - \frac{\sigma^2 i}{h_n^2} \int u \exp\left(-iu\frac{y - Y_j}{h_n}\right) K^{\star}(u) \, \mathrm{d}u \right] \\[7pt]
&\overset{(b)}{=}\frac{\sigma}{n}\sum_{j=1}^n \Biggl[  (y-m_0)K\left(\frac{y-Y_i}{h_n}\right)+\frac{\sigma^2}{h_n^2}K'\left(\frac{y-Y_i}{h_n}\right) \Biggr]
\end{align*}
where (a) uses the change-of-variable $t=h_n^{-1}u$; (b) obtains by inverse Fourier transform formulas $K(z)=(2\pi)^{-1}\int \exp(-itz)K^{\star}(t)\mathrm{d}t$ and $K'(z)=(2\pi)^{-1}\int \exp(-it z)(-it) K^{\star}(t)\mathrm{d}t$.

Similarly, we can obtain the expression of its popular counterpart:
\begin{equation*}
\Psi(m_0)=(2\pi)^{-1}\sigma \int \left(y-m_0-i\sigma^2 t\right)\exp\left(-ity\right)  f_{Y}^{\star}(t) \mathrm{d}t=\sigma (y-m_0)f_{Y}(y)+\sigma^3f_{Y}'(y)\equiv 0,
\end{equation*}
and $\Psi_n(m_0)-\Psi(m_0)$ can be represented as the following
\begin{equation*}
\sigma (y-m_0)\left( \frac{1}{n}\sum_{i=1}^n K\left(\frac{y-Y_i}{h_n}\right)-f_{Y}(y) \right)+\sigma^3\left(\frac{1}{nh_n^2}K'\left(\frac{y-Y_i}{h_n}\right)-f_{Y}'(y)\right).
\end{equation*}
  The convergence rate analysis follows directly from standard results on kernel density estimation (see e.g. \citeA{ullah1999nonparametric,li2007nonparametric}; \citeA[Appendix]{horowitz2009semiparametric} for textbook treatments). We note that the optimal rates of convergence of $\Psi_n(m_0)-\Psi(m_0)$ coincides with that for estimating the derivative of a density, namely $n^{-(r-1)/(2r+1)}$, where $r$ is the smoothness of $f_{Y}(y)$ (\citeNP{stone1980optimal}). This smoothness $r$ in our setting is characterized in \Cref{prop:pm_smoothness}.


\begin{proposition}
Under the normal location model (\ref{normal_mean}) and the smoothness assumption for $g\in\mathcal{G}(s,L)$ (\ref{assumption1}), we have that $f_{Y}(y)$ is super-smooth in the sense that
\begin{equation*}
 \int \ |f_{Y}^{\star}(t)|^2(t^2+1)^{r}\mathrm{d}t < \infty, \ \ \ \text{for all} \ r\geq 0.
\end{equation*}
\label{prop:pm_smoothness}
\end{proposition}
\vspace{-3.5em}
\begin{proof}
Given that $f_{Y}^{\star}(t)=g^{\star}(t)\exp\left(-\frac{1}{2}\sigma^2t^2\right)$. For any $r\geq 0$, we have that
\begin{equation}
 \int \ |f_{Y}^{\star}(t)|^2(t^2+1)^{r}\mathrm{d}t=\int \ |g^{\star}(t)|^2(t^2+1)^{s}\underbrace{\left[\left(t^2+1\right)^{r-s}\exp\left(-\sigma^2 t^2\right)\right]}_{w_r(t)}\mathrm{d}t.
 \label{equ:appendix_integral}
\end{equation}
The function $w_r(t)$ is continuous and vanishes as $|t|\rightarrow\infty$, thus it must achieve its supremum at some finite point. That is, $W_r=\sup_{t\in\mathbb{R}} w_r(t)<\infty$. Then the right-hand side of \Cref{equ:appendix_integral} is not greater than $W_r\int |g^{\star}(t)|^2(1+t^2)^s\mathrm{d}t<W_rL=:L_r<\infty$. Hence $f_{Y}\in \mathcal{G}(r,L_r)$ for all $r\geq 0$ and thus supersmooth.
\end{proof}
\vspace{-1em}
By \Cref{prop:pm_smoothness}, it's clear that the optimal rates of convergence of $\Psi_n(m_0)-\Psi(m_0)$ -- and thus the posterior mean -- is $O_p(n^{-1/2})$. To let the kernel estimator match this parametric rate, one should use an infinite-order kernel. For instance, the sinc kernel $K(z)=\sin(z)/\pi z$ has infinite order,\footnote{The fourier transformation of the sinc kernel is $K^{\star}(t)=\mathbf{1}_{[-1,1]}(t)$. It is flat at $t=0$, that enables to ensure that all moments of the kernel vanish, i.e. $\int z^qK(z)\mathrm{d}z=i^{q}\frac{\mathrm{d}^qK^{\star}(t)}{\mathrm{d}t^q}\big{|}_{t=0}=0$ for any $q\geq 1$. } see e.g. \citeA[Page 42]{meister2009deconvolution} for more details.

\begin{remark}
The point of parametric rate of the posterior mean in the normal location model has been previously made in the empirical Bayes literature. \citeA[Theorem 2]{zhang1997empirical} showed that the kernel estimator for the $k$-th derivative of marginal density $\hat{f}^{(k)}(\cdot)$ converges at parametric rate (up to a logarithm factor) in $L_2$ and $L_{\infty}$ norm. By Tweedie's formula, one can obtain that the posterior mean shall also convergence at parametric rate. \citeA[Theorem 1]{zhang2009generalized} showed that the generalized maximum likelihood estimator (GMLE) for the marginal density $f_{Y}(y)$ converges at parametric rate (up to a logarithm factor) in Hellinger distance given that the mixing distribution has compact support or has a heavier-than-normal exponential tail.
\end{remark}

\subsection{Extension to precision dependence}
\label{subsection:precision_dependene}

The baseline model in \Cref{subsection:baseline_implementation} treats the sampling variance $\sigma_i^2$ as observed and
assumes precision independence, in the sense that $\theta_i \perp\!\!\!\perp \sigma_i$. Following \citeA{chen2026empirical}, a tractable relaxation is to model the conditional distribution of $\theta_i$ given $\sigma_i$ as a location-scale family.

Specifically, suppose
\[
Y_i \mid \theta_i,\sigma_i \sim N(\theta_i,\sigma_i^2), \qquad
\theta_i = m(\sigma_i) + s(\sigma_i)\xi_i,
\]
where $m(\cdot)$ and $s(\cdot)>0$ are unknown functions, $\xi_i \sim G(\cdot)$, and
$\xi_i \perp\!\!\!\perp \sigma_i$. Under this specification, precision dependence operates through the conditional location $m(\sigma_i)$ and conditional scale $s(\sigma_i)$, while the shape of condition distribution of $\theta_i\mid \sigma_i$, govern by $G$, is common across values of $\sigma_i$.

Define the normalized variables
\begin{equation}
Z_i = \frac{Y_i - m(\sigma_i)}{s(\sigma_i)}, \qquad
v_i = \frac{\sigma_i}{s(\sigma_i)},
\qquad \widetilde{\mathcal{D}}_i = (Z_i,v_i).
\label{equ:transformation}
\end{equation}
Then
\begin{equation}
Z_i = \xi_i + \eta_i, \qquad \eta_i \mid v_i \sim N(0,v_i^2),
\label{equ:normalized_model}
\end{equation}
and, by construction, $\xi_i \perp\!\!\!\perp v_i$. Hence the transformed data satisfy the
same heteroskedastic normal means model as in \Cref{subsection:baseline_implementation}, but with latent effect $\xi_i$
and known heteroskedastic variances $v_i^2$.

Let $\mathcal{D}_i=(Y_i,\sigma_i)$. In the precision-dependent setting, the posterior quantile for $\theta_i$
is defined conditional on $\mathcal{D}_i$:
\[
q_{\theta_i}(\tau;\mathcal{D}_i)
=
\inf\left\{
u : P(\theta_i \le u \mid \mathcal{D}_i)\ge \tau
\right\},
\qquad \tau \in (0,1).
\]
Because $\theta_i = m(\sigma_i)+s(\sigma_i)\xi_i$ is strictly increasing in $\xi_i$, posterior
quantiles transform equivariantly:
\[
q_{\theta_i}(\tau;\mathcal{D}_i)
=
m(\sigma_i)+s(\sigma_i)\,q_{\xi_i}(\tau;\widetilde{\mathcal{D}}_i),
\]
where $q_{\xi_i}(\tau;\widetilde{\mathcal{D}}_i)$ denotes the posterior quantile of $\xi_i$ under the normalized
model (\ref{equ:normalized_model}).

Accordingly, the oracle NP-EBCI under precision dependence is
\[
\mathrm{CI}_i^{\mathrm{NP}*,\mathrm{CLOSE}}
=
\Big[
q_{\theta_i}(\alpha/2;\mathcal{D}_i),\ q_{\theta_i}(1-\alpha/2;\mathcal{D}_i)
\Big]
\]
or, equivalently,
\[
\mathrm{CI}_i^{\mathrm{NP}*,\mathrm{CLOSE}}
=
\Big[
m(\sigma_i)+s(\sigma_i)\,q_{\xi_i}(\alpha/2;\widetilde{\mathcal{D}}_i),\
m(\sigma_i)+s(\sigma_i)\,q_{\xi_i}(1-\alpha/2;\widetilde{\mathcal{D}}_i)
\Big].
\]
By construction, the oracle NP-EBCI achieves exact conditional coverage:
\[
\operatorname{P}\!\left(
\theta_i \in \mathrm{CI}_i^{\mathrm{NP}*,\mathrm{CLOSE}} \ \Big{|} \  \mathcal{D}_i
\right)
=
1-\alpha.
\]

In practice, we replace the unknown functions $m(\cdot)$ and $s(\cdot)$ with their respective estimators, $\hat m(\cdot)$ and $\hat s(\cdot)$. We then construct the empirical normalized variables:
\[
\hat Z_i = \frac{Y_i-\hat m(\sigma_i)}{\hat s(\sigma_i)}, \qquad
\hat v_i = \frac{\sigma_i}{\hat s(\sigma_i)}, \qquad
\widehat{\widetilde{\mathcal{D}}}_i = (\hat Z_i,\hat v_i).
\]
Applying the baseline heteroskedastic posterior-quantile estimator in \Cref{subsection:baseline_implementation} to the normalized sample $\{(\hat Z_j,\hat v_j)\}_{j=1}^n$, we obtain $\hat q_{\xi_i}(\tau;\hat{\tilde{\mathcal{D}}}_i)$. The feasible NP-EBCI is then
\[
\mathrm{CI}_i^{\mathrm{NP},\mathrm{CLOSE}}
=
\Big[
\hat m(\sigma_i)+\hat s(\sigma_i)\,\hat q_{\xi_i}(\alpha/2;\widehat{\widetilde{\mathcal{D}}}_i),\
\hat m(\sigma_i)+\hat s(\sigma_i)\,\hat q_{\xi_i}(1-\alpha/2;\widehat{\widetilde{\mathcal{D}}}_i)
\Big].
\]
We leave the comprehensive theoretical analysis of this NP-EBCI under precision dependence for future work.

\section{Auxiliary results}

\begin{proposition}
Let $\rho_{\tau}(u)=u(\tau-\mathbf{1}\{ u\leq 0 \})$ be the check function. Then, for all $x, y \in \mathbb{R}$,
\begin{equation*}
\rho_{\tau}(x-y)-\rho_{\tau}(x)=-y\left( \tau-\mathbf{1}\{x\leq 0\} \right)+\int_0^{y} \left( \mathbf{1}\{x\leq s\}-\mathbf{1}\{x\leq 0\} \right)\mathrm{d}s.
\end{equation*}
\label{proposition_c1}
\end{proposition}
\vspace{-3.5em}
\begin{proof}
It follows from the identity from \citeA{knight1998limiting} and the reformulation of the check function $\rho_{\tau}(u) = (\tau - \tfrac{1}{2}) u + \tfrac{1}{2}|u|$.
\end{proof}

\subsection{Uniform consistency of deconvolution kernel density estimator}

Uniform consistency of the deconvolution kernel density estimator has been considered in the literature. In particular, \citeA{taylor1990strongly} establish a stronger uniform strong-consistency result for a broader class of deconvolution estimators. \Cref{prop:uniform_consist} is provided only as a self-contained specialization tailored to the deconvolution kernel density estimator used in our proof.

\begin{proposition}
Let $\hat{g}(\theta)$ be the deconvolution kernel density estimator defined in \Cref{equ:deconvolution_estimator}. Let \Cref{assumption_g_regularity,assumption1,assumption2,assumption3} hold.  Setting $h_n=c(\log n)^{-1/2}$ with any constant $c\geq\sigma$, we have that
\begin{equation*}
\sup_{\theta\in\mathbb{R}} \ \big{|} \hat{g}(\theta)-g(\theta)  \big{|} \p 0.
\end{equation*}
\label{prop:uniform_consist}
\end{proposition}
\vspace{-2.5em}
\begin{proof}[Proof of \Cref{prop:uniform_consist}]
We decompose $\hat{g}(\theta)-g(\theta)=S_{n1}(\theta)+S_{n2}(\theta)$ where
\begin{align*}
S_{n1}(\theta)&=(2\pi)^{-1}\int \frac{\exp(-it\theta)}{f_{\varepsilon}^{\star}(t)}\left(\hat{f}_{Y}^{\star}(t)-f_{Y}^{\star}(t)\right)K^{\star}(h_nt)\mathrm{d}t \\[7pt]
S_{n2}(\theta)&=(2\pi)^{-1}\int \frac{\exp(-it\theta)}{f_{\varepsilon}^{\star}(t)}f_{Y}^{\star}(t)\big{(} K^{\star}(h_nt)-1 \big{)}\mathrm{d}t
\end{align*}
where $S_{n1}(\theta)$ is the stochastic deviation of the estimator from its expectation, $S_{n2}(\theta)$ is the bias term. Note that $\sup_{\theta\in\mathbb{R}} |\hat{g}(\theta)-g(\theta)|\leq \sup_{\theta\in\mathbb{R}} |S_{n1}(\theta)|+\sup_{\theta\in\mathbb{R}} |S_{n2}(\theta)|$.

To analyze the bias term $S_{n2}(\theta)$, note first that
\begin{align*}
\sup_{\theta\in\mathbb{R}} \ \left| S_{n2}(\theta) \right| &\leq (2\pi)^{-1}\int \big{|}g^{\star}(t)\big{|} \big{|}K^{\star}(h_nt)-1   \big{|}\mathrm{d}t 
\end{align*}
Since $K\in L^1(\mathbb{R})$, the function $K^{\star}$ is bounded and continuous, and $K^{\star}(0)=1$. Thus $K^{\star}(h_nt)\rightarrow 1$ for every fixed $t$, while $|K^{\star}(h_nt)-1|\leq \|K^{\star}\|_{\infty}+1$. By \Cref{assumption_g_regularity}, $g^{\star}\in L^1(\mathbb{R})$. Therefore the dominated convergence theorem yields that $\sup_{\theta\in\mathbb{R}}|S_{n2}(\theta)|=o(1)$.

For the first term $S_{n1}(\theta)$, since $K^{\star}$ is supported on $[-1,1]$, we have that
\begin{align}
\sup_{\theta\in \mathbb{R} } \ \left| S_{n1}(\theta) \right|&=\sup_{\theta\in \mathbb{R} } \ \left| (2\pi)^{-1}\int  \cfrac{\exp(-it\theta)}{f_{\varepsilon}^{\star}(t)}\left(\widehat{f}^{\star}_Y(t)-f^{\star}_Y(t)\right)K^{\star}(h_n t)  \mathrm{d}t \right| \notag \\[7pt]
&\leq \sup_{|t|\leq h_n^{-1}} \  \left|\widehat{f}^{\star}_Y(t)-f^{\star}_Y(t)  \right| \times (2\pi)^{-1} \int_{-h_n^{-1}}^{h_n^{-1}} \frac{|K^{\star}(h_nt)|}{|f_{\varepsilon}^{\star}(t)|}\mathrm{d}t \label{equ:C1_S_n1}
\end{align}
where the last line obtains from H\"{o}lder's inequality. By Theorem 1 in \citeA{csorgHo1983long}, if $\log (h_n^{-1})/n\rightarrow 0$, then
\begin{equation*}
\sup_{|t|\leq h_n^{-1}} \  |\widehat{f}^{\star}_Y(t)-f^{\star}_Y(t)  |=O_p\left( \sqrt{\frac{ \log(h_n^{-1}) }{n}} \right).
\end{equation*}
Meanwhile, for the second product in \Cref{equ:C1_S_n1}, we have that
\begin{align*}
(2\pi)^{-1} \int_{-h_n^{-1}}^{h_n^{-1}} \frac{|K^{\star}(h_nt)|}{|f_{\varepsilon}^{\star}(t)|}\mathrm{d}t &\leq \frac{\|K\|_{\infty}}{2\pi}\int_{-h_n^{-1}}^{h_n^{-1}}\exp\left(\frac{\sigma^2t^2}{2}\right)\mathrm{d}t=\frac{\|K\|_{\infty}}{2\pi h_n}\int_{-1}^1 \exp\left( \frac{\sigma^2u^2}{2h_n^2} \right)\mathrm{d}u \\[7pt]
&\leq \frac{\|K\|_{\infty}}{\pi h_n}\int_{0}^1 \exp\left( \frac{\sigma^2u}{2h_n^2} \right)\mathrm{d}u=O\left( h_n\exp\left(\frac{\sigma^2}{2h_n^2}\right) \right).
\end{align*}
Then if $h_n=c(\log n)^{-1/2}$ with $c\geq \sigma$, we obtain that
\begin{equation*}
\sup_{\theta\in\mathbb{R}}\ |S_{n1}(\theta)|=O_p\left( h_n\exp\left(\frac{\sigma^2}{2h_n^2}\right)\sqrt{ \frac{\log(h_n^{-1})}{n} } \right)=O_p\left( n^{\sigma^2/(2c^2)-1/2}\sqrt{\frac{\log\log n}{\log n}} \right)=o_p(1).
\end{equation*}
Combining the results for $S_{n1}(\theta)$ and $S_{n2}(\theta)$, we obtain that $\sup_{\theta\in\mathbb{R}} |\hat{g}(\theta)-g(\theta)|=o_p(1)$, as desired.
\end{proof}

\subsection{Technical Lemmas for the proof of \Cref{ebci_coverage_lower_bound}}

\label{appendix:technical_lemma_minimax}

This subsection collects the auxiliary results used in Appendix \ref{subsection:proof_minimax_coverage}. The lemmas establish several ingredients of the lower-bound argument: the construction
of the common perturbation $\bar{H}_T$, the verification that the alternatives $g_{m}$ belong to $\mathcal G(s,L)$, the sign separation of the induced conditional coverage errors, and the bounds showing that the alternatives remain statistically close. Throughout this subsection, we maintain the notation of Appendix \ref{subsection:proof_minimax_coverage}.

\subsubsection{Construction of common perturbation $\bar{H}_T$}
\label{subsubsection:construction_common}

Consider $\mathcal{J}=[\beta_{L},\beta_{U}]\subset(0,\alpha)$ defined in Appendix \ref{subsection:proof_minimax_coverage}. Choose a constant $c_0>0$ such that for all sufficiently large $T$,
\begin{equation*}
\mathcal{J}_{T}:=\left[\beta_L-\frac{c_0}{T},\beta_U+\frac{c_0}{T}\right]\subset (0,\alpha) \quad \text{and} \quad \mathrm{dist}\left(\mathcal{J}_T,\mathcal{J}_T+1-\alpha\right)>0.
\end{equation*}

Let $\theta_L=q_{G_0}(\beta_L;y), \theta_U=q_{G_0}(\beta_U+1-\alpha;y)$. We construct the common perturbation $\bar{H}_T$ as follows:
\begin{equation*}
\bar{H}_T(\theta):=\eta T^{-s+1/2}\Big{[} \kappa_{-}\left( T(\theta-\theta_L) \right)-\kappa_{+}\left(T(\theta-\theta_U)\right)\Big{]}
\end{equation*}
where $\eta>0$ is a sufficiently small constant, and $\kappa_-,\kappa_+\in C_c^\infty(\mathbb R)$ are nonnegative functions satisfying
\[
\int \kappa_-(u)\,du=\int \kappa_+(u)\,du=1,
\qquad
\mathrm{supp}(\kappa_-)\subset[-2,-1],
\qquad
\mathrm{supp}(\kappa_+)\subset[1,2].
\]
\begin{lemma}
\label{lemma:common_perturbation}

\noindent (i) Define $\bar{g}_T:=g_0+\bar{H}_T$ where $g_0$ is the baseline Cauchy density, then $\int \bar{H}_T(\theta)\mathrm{d}\theta=0, \bar{g}_T(\theta)\geq c\,g_0(\theta)\geq 0$ for all $\theta\in\mathbb R$ for some constant $c>0$ independent of $T$, and $\bar{g}_T\in G(s,3L/4).$

\noindent (ii) Define the conditional coverage error of the interval $I(\beta;y)$ under the common perturbation:
\begin{equation*}
\bar{D}_T(\beta;y):=\operatorname{P}_{\bar{G}_T}\left( \theta\in I(\beta;y) \mid Y=y \right)-(1-\alpha)
\end{equation*}
where $\bar{G}_T$ has density $\bar{g}_T$. Then there exists a constant $c>0$ such that
\begin{align*}
\bar D_T(\beta;y)\geq c\eta T^{-s-1/2}
\qquad &\text{for all }\beta\le \beta_L-\frac{c_0}{T}, \\[10pt]
\bar D_T(\beta;y)\leq -\,c\eta T^{-s-1/2}
\qquad & \text{for all }\beta\ge \beta_U+\frac{c_0}{T}.
\end{align*}
\end{lemma}
\begin{proof}[Proof of \Cref{lemma:common_perturbation}]

First we check that \[
\int_{\mathbb R}\bar H_T(\theta)\,d\theta
=
\eta T^{-s+1/2}
\Big{[} \int \kappa_{-}\left( T(\theta-\theta_L) \right)\mathrm{d}\theta-\int \kappa_{+}\left(T(\theta-\theta_U)\right)\mathrm{d}\theta\Big{]}
=0.
\]
Its Fourier transform is
\[
\bar{H}_T^{\star}(t)
=
\eta T^{-s-1/2}
\Big{[}
\exp\left(it\theta_L\right)\kappa_-^{\star}(t/T)-\exp\left(it\theta_U\right)\kappa_+^{\star}(t/T)\Big{]},
\]
so $\int_{\mathbb R}|\bar{H}_T^{\star}(t)|^2(1+t^2)^s\mathrm{d}t
\le C\eta^2$ where $C$ depends only on $\kappa_-$ and $\kappa_+$. Hence, for $\eta$ sufficiently small,
\[
\int_{\mathbb R}|\bar{H}_T^{\star}(t)|^2(1+t^2)^s\mathrm{d}t\le \frac{L}{4}.
\]
Since $\bar{H}_T$ is supported on two $O(T^{-1})$-neighborhoods of the fixed points $\theta_L$ and $\theta_U$,
and has size $O(\eta T^{-s+1/2})$, the same argument in Appendix \ref{subsection:proof_minimax_coverage} implies $|\bar H_T(\theta)|\le C\eta(1+\theta^2)^{-1}$ uniformly in $\theta$ and $T$. Because the baseline density $g_0$ satisfies $g_0(\theta)\ge c_0'(1+\theta^2)^{-1}$, choosing $\eta$ sufficiently small yields
\[
\bar g_T(\theta)=g_0(\theta)+\bar H_T(\theta)\ge c\,g_0(\theta)\ge 0
\]
for all $\theta$, and therefore $\bar g_T\in G(s,3L/4)$, which proves part (i).

We now prove part (ii). We first consider the left-hand case $\beta\leq \beta_{L}-c_0/T$. We write $q(\beta):=q_{G_0}(\beta;y)$ and let $c_{-}=\inf_{\beta\in\mathcal{J}}q'(\beta)>0$. By the mean-value theorem,
\begin{equation*}
q(\beta_L)-q(\beta)=q'(\tilde{\beta})\left(\beta_L-\beta\right)\geq c_{-}\frac{c_0}{T}>\frac{2}{T}
\end{equation*}
for some $\tilde{\beta}\in[\beta,\beta_{L}]\subset\mathcal{J}$. Hence $q(\beta)\leq \theta_{L}-\frac{2}{T}$. Since $\mathrm{supp}(\kappa_{-}(T(\cdot-\theta_L)))\subset[\theta_L-\frac{2}{T},\theta_L-\frac{1}{T}]$, the positive bump is contained in the interval $I(\beta;y)$. On the other hand, $\beta+1-\alpha\leq \beta_{L}+1-\alpha-\frac{c_0}{T}<\beta_U+1-\alpha$. So by monotonicity of $q$, $q(\beta+1-\alpha)<\theta_U$. Therefore $I(\beta;y)$ excludes the support of the negative bump, which lies in $[\theta_U+\frac{1}{T},\theta_U+\frac{2}{T}]$.

The right-hand case $\beta\geq \beta_U+c_0/T$ is symmetric. Since $q(\beta+1-\alpha)-q(\beta_U+1-\alpha)\geq c_{-}\frac{c_0}{T}>\frac{2}{T}$, we have that $q(\beta+1-\alpha)\geq \theta_U+\frac{2}{T}$, which implies that $I(\beta;y)$ contains the support of the negative bump. Also $\beta>\beta_L$, hence $q(\beta)>\theta_L$, so $I(\beta;y)$ excludes the support of the positive bump.

A standard numerator-denominator expansion around $G_0$ therefore gives
\[
\bar D_T(\beta;y)
=
f_{G_0}(y)^{-1}
\int
\phi\!\left(\frac{y-\theta}{\sigma}\right)
\bigl(1\{\theta\in I(\beta;y)\}-(1-\alpha)\bigr)\bar H_T(\theta)\,d\theta
+
o\!\left(T^{-s-1/2}\right),
\]
uniformly over $\beta$. In the left-hand case, the coefficient
\(
1\{\theta\in I(\beta;y)\}-(1-\alpha)
\)
equals $\alpha$ on the support of the positive bump and equals $-(1-\alpha)$ on the support of the
negative bump; because the perturbation itself has signs $+$ and $-$ on these two supports, both
contributions are positive and of order $\eta T^{-s-1/2}$. The right-hand case is symmetric and
gives a negative quantity of the same order. This proves part (ii).
\end{proof}

\subsubsection{Verifying that alternatives $g_{m}\in\mathcal G(s,L)$}
\label{subsubsection_verifying_alternative}

Recall from \Cref{subsubsection:minimax_coverage_step1} that the alternatives are defined by $g_{m}=g_0+\bar{H}_T+\tilde{H}_{m,T}$. We now verify that each alternative $g_{m}$ belongs to $\mathcal{G}(s,L)$. The argument has two parts. First, we control the Sobolev norm of $\tilde{H}_{m,T}$. We then show that this perturbation is uniformly small relative to the baseline density, so that $g_{m}$ remains nonnegative.

Because each $H_{j,T}$ has an integral of zero, it follows that $\int \widetilde{H}_{m,T}(\theta)\mathrm{d}\theta=0$. Moreover, the Fourier transform of the cumulative perturbation is
\begin{equation*}
\widetilde{H}_{m,T}^{\star}(t)=\varepsilon\nu T^{-s-1/2}(-i\,\text{sgn}(t))k^{\star}(|t|/T)S_{m,T}(t)
\end{equation*}
where
\begin{equation*}
S_{m,T}(t):=\sum_{j<m}\exp(it\theta_j)-\sum_{j\geq m}\exp(it\theta_j).
\end{equation*}
By \Cref{lem:oscillatory-sums}, $|S_{m,T}(t)|\le C$ uniformly in $m$ and in
$t$ on the support of $k^\star(|t|/T)$. Hence
\begin{equation}
\sup_{m,T} \ \int_{\mathbb R}
\big|\widetilde H_{m,T}^\star(t)\big|^2(1+t^2)^s\mathrm{d}t
\le C\varepsilon^2\nu^2.
\label{equ:minimax_A.7}
\end{equation}

It remains to verify the positivity of $g_{m}$. Since $h$ belongs to the Schwartz space $S(\mathbb R)$, taking $N=4$ gives that
\[
|h(u)|\le C_4(1+|u|)^{-4}
\qquad\text{for all }u\in\mathbb R.
\]
Therefore we can bound the cumulative perturbation by
\begin{equation}
|\widetilde H_{m,T}(\theta)|
\le C\varepsilon\nu T^{-s+1/2}
\sum_{j=1}^M (1+T|\theta-\theta_j|)^{-4}.
\label{equ:minimax_A.8}
\end{equation}
The centers $\{\theta_j\}_{j=1}^M$ lie within the compact set $q_{G_0}(\mathcal{J};y)$ and
have a spacing of order $T^{-1}$. Thus, for any $\theta$ inside a fixed compact set, the sum on the right-hand side of \eqref{equ:minimax_A.8} is uniformly bounded. On the other hand, outside a fixed
compact neighborhood of $q_{G_0}(\mathcal{J};y)$, we have
\[
\sum_{j=1}^M (1+T|\theta-\theta_j|)^{-4}
\le
CT^{-3}(1+|\theta|)^{-4},
\]
because $M\asymp T$. Combining these two bounds, and using $s\geq 1/2$, we obtain
\begin{equation*}
|\tilde H_{m,T}(\theta)|
\le C\varepsilon\nu (1+\theta^2)^{-1} \quad \quad \text{for all} \ \theta, m,
\end{equation*}
for all sufficiently large $T$.

Now recall that $g_0$ is the baseline Cauchy density; thus, there exists $c>0$ such that $g_0(\theta)\ge c(1+\theta^2)^{-1}$ for all $\theta\in\mathbb R$. In addition, \Cref{lemma:common_perturbation} (i) gives that $\bar{g}_T(\theta)=g_0+\bar{H}_T\geq c_1g_0(\theta)$ for all $\theta$ and for all sufficiently large $T$. Then we may choose $\varepsilon>0$ sufficiently small enough such that $g_{m}(\theta)=\bar{g}_T(\theta)+\tilde{H}_{m,T}(\theta)
\geq c' g_0(\theta)\ge 0$ for all $\theta,m$ and all sufficiently large $T$.


Finally, $\bar{g}_T$ already satisfies the Sobolev bound by \Cref{lemma:common_perturbation} (i), and the display above gives the corresponding bound for $\tilde{H}_{m,T}$. We conclude that  $g_{m}\in\mathcal G(s,L)$ for all $m$.

\subsubsection{Sign separation of errors in conditional coverage probabilities}
\label{subsubsection_sign_separation}

We consider the sign of the conditional coverage error $D_{m}(\beta;y)$. The argument has two steps. First, we decompose $D_{m}(\beta;y)$ into a leading linear term and a negligible common and remaining term. Second, we show that the leading term has a step-function sign pattern across the cells $\{B_k\}_{k=1}^M$. Since the other two terms are smaller than the separation scale $T^{-s-1/2}$, the actual conditional coverage error has the same sign pattern.

For each $m=1,...,M$, recall that
\begin{align*}
D_{m}&=\operatorname{P}_{G_{m}}\left(\theta\in I(\beta;y)\mid Y=y\right)-(1-\alpha) \\[7pt]
&=f_{G_{m}}^{-1}(y)\left[ \int \phi\left(\frac{y-\theta}{\sigma}\right)\Big{(} \mathbf{1}\left\{ \theta\in I(\beta;y) \right\}-(1-\alpha) \Big{)}g_{m}(\theta)\mathrm{d}\theta \right].
\end{align*}
For any perturbation $H$, we also define
\begin{equation*}
\Gamma_H(\beta;y)=f_{G_0}^{-1}(y)\left[ \int \phi\left(\frac{y-\theta}{\sigma}\right)\Big{(} \mathbf{1}\left\{ \theta\in I(\beta;y) \right\}-(1-\alpha) \Big{)}H(\theta)\mathrm{d}\theta \right].
\end{equation*}
This is the linearized effect of the perturbation $H$ on the conditional coverage probability, evaluated around the baseline $G_0$.

\begin{lemma}
For each $m=1,...M$, the error of conditional coverage probability has the decomposition
\begin{equation*}
D_{m}(\beta;y)=\bar{D}_T(\beta;y)+\Gamma_{\tilde{H}_{m,T}}(\beta;y)+\tilde{R}_{m}(\beta;y)
\end{equation*}
where $\bar{D}_T(\beta;y)$ is defined in \Cref{lemma:common_perturbation}, and the remainder term is
\begin{equation*}
\tilde{R}_{m}(\beta;y)=-\bar{D}_T(\beta;y)\cfrac{(\tilde{H}_{m,T}*f_{\varepsilon})(y)}{f_{G_{m}}(y)}
-\Gamma_{\tilde{H}_{m,T}}(\beta;y)\cfrac{ ((\bar{H}_T+\widetilde{H}_{m,T})*f_{\varepsilon})(y) }{f_{G_{m}}(y)}.
\end{equation*}
Moreover,
\begin{equation}
\sup_{m}\ \sup_{\beta\in\mathcal{J}}\ |\tilde{R}_{m}(\beta;y)|\
=
o\!\left(T^{-s-1/2}\right).
\label{equ:sup_control_R_mT}
\end{equation}
\label{lemma_decomposition_D_mt}
\end{lemma}
\vspace{-2.5em}
\begin{proof}[Proof of \Cref{lemma_decomposition_D_mt}]

By \Cref{lemma:common_perturbation} (ii), $\sup_{\beta\in\mathcal{J}} |\bar{D}_T(\beta;y)|\leq C\eta T^{-s-1/2}$.

Next, by linearity,
\begin{align*}
\Gamma_{\tilde H_{m,T}}(\beta;y)
=
\varepsilon\left(
\sum_{j<m}\Gamma_{H_{j,T}}(\beta;y)
-
\sum_{j\ge m}\Gamma_{H_{j,T}}(\beta;y)
\right).
\end{align*}
Hence by \Cref{lem:localized-perturbations} (ii), $\bigl|\Gamma_{\widetilde H_{m,T}}(\beta;y)\bigr|
\le
\varepsilon C_0 T^{-s-1/2}
\sum_{j=1}^M
\bigl(1+T|\beta-\beta_j|\bigr)^{-2}$. Because the grid points satisfy $\beta_{j+1}-\beta_j=\Delta/T$, the sum on the right-hand side is uniformly bounded over $\beta\in\mathcal{J}$. Therefore
\[
\sup_m \ \sup_{\beta\in\mathcal J} \
\bigl|\Gamma_{\widetilde H_{m,T}}(\beta;y)\bigr|
\le
C\varepsilon T^{-s-1/2}.
\]

We now bound the convolution terms. Since the support of $k^\star(|t|/T)$ is contained in the interval $[T,2T]$,
\Cref{lem:oscillatory-sums} yields $|\widetilde H_{m,T}^\star(t)| \le C\varepsilon\nu T^{-s-1/2}|k^\star(|t|/T)|$. Hence
\[
\left|(\widetilde H_{m,T}*f_\varepsilon)(y)\right|
\le
\frac{1}{2\pi}
\int_{\mathbb R}
|\widetilde H_{m,T}^\star(t)|\,\exp(-\sigma^2 t^2/2)\,dt
\le
C\varepsilon\nu T^{-s+1/2}\exp(-cT^2)
\]
for some $c>0$.
For the common perturbation, we have that
\[
\bigl|(\bar H_T*f_\varepsilon)(y)\bigr|
\le
\left\|\phi\!\left(\frac{y-\cdot}{\sigma}\right)\right\|_\infty
\|\bar H_T\|_1
\le
C\eta T^{-s-1/2}.
\]
Therefore
\[
\bigl|((\bar H_T+\widetilde H_{m,T})*f_\varepsilon)(y)\bigr|
\le
C\eta T^{-s-1/2}
+
C\varepsilon\nu T^{-s+1/2}\exp(-cT^2).
\]

Finally, by \Cref{lemma:common_perturbation}(i), $f_{\bar G_T}(y)\ge c\,f_{G_0}(y)>0$. Since \((\widetilde H_{m,T}*f_\varepsilon)(y)=o(1)\) uniformly in \(m\),
it follows that $\inf_m f_{G_{m}}(y)\ge c_y>0$ for all sufficiently large \(T\).

Combining the previous bounds, uniformly in \(m\) and \(\beta\in\mathcal J\),
\begin{align*}
|\widetilde R_{m}(\beta;y)|
&\le
C|\bar D_T(\beta;y)|\,|(\widetilde H_{m,T}*f_\varepsilon)(y)|
+
C|\Gamma_{\widetilde H_{m,T}}(\beta;y)|\,
\bigl|((\bar H_T+\widetilde H_{m,T})*f_\varepsilon)(y)\bigr| \\[7pt]
&\le
C\eta\varepsilon\nu\,T^{-2s}\exp(-cT^2)
+
C\eta\varepsilon\,T^{-2s-1}
+
C\varepsilon^2\nu\,T^{-2s}\exp(-cT^2)
=o\!\left(T^{-s-1/2}\right)
\end{align*}
as desired.
\end{proof}

Recall from \Cref{subsubsection:minimax_coverage_step1} that
\[
B_k=
\left[\beta_k-\frac{\Delta}{2T},\,\beta_k+\frac{\Delta}{2T}\right]\cap\mathcal J, \quad \quad k=1,...,M.
\]
Suppose that $\Delta$ is chosen so that
\begin{equation}
2C_0\sum_{\ell\ge 1}
\left(1+\Delta\left(\ell-\frac12\right)\right)^{-2} \leq \frac{c_0}{4},
\label{equ:appendix_tail_bound_2}
\end{equation}
and that $r$ is chosen so that
\begin{equation*}
C_0\sum_{\ell\geq r}\left(1+\Delta\left(\ell-\frac{1}{2}\right)\right)^{-2}\leq\frac{c_0}{4}.
\end{equation*}

\begin{lemma}
For each $m=1,...,M$,
\begin{align}
\Gamma_{\widetilde H_{m,T}}(\beta;y)
\geq
\frac{\varepsilon c_0}{2}T^{-s-1/2}, &\qquad
\text{for }\beta\in \bigcup_{k\le m-r}B_k  \label{equ:sign_function_property_1} \\[10pt]
\Gamma_{\widetilde H_{m,T}}(\beta;y)
\leq
-\frac{\varepsilon c_0}{2}T^{-s-1/2}, &\qquad
\text{for }\beta\in \bigcup_{k\geq m+r}B_k. \label{equ:sign_function_property_2}
\end{align}
\label{lemma:sign_function_property}
\end{lemma}
\vspace{-2.5em}
\begin{proof}[Proof of \Cref{lemma:sign_function_property}]
We prove the first claim; the second is symmetric. Fix $k\le m-r$ and let $\beta\in B_k$. Isolating the $j=k$ term gives
\begin{align*}
\Gamma_{\tilde H_{m,T}}(\beta;y)
&=
\varepsilon\Biggl(
\Gamma_{H_{k,T}}(\beta;y)
+
\sum_{\substack{j<m\\ j\neq k}}\Gamma_{H_{j,T}}(\beta;y)
-
\sum_{j\ge m}\Gamma_{H_{j,T}}(\beta;y)
\Biggr) \\
&\ge
\varepsilon\Biggl(
\Gamma_{H_{k,T}}(\beta;y)
-
\sum_{\substack{j<m\\ j\neq k}}
\bigl|\Gamma_{H_{j,T}}(\beta;y)\bigr|
-
\sum_{j\ge m}
\bigl|\Gamma_{H_{j,T}}(\beta;y)\bigr|
\Biggr).
\end{align*}
Since $\beta\in B_k$, \Cref{lem:localized-perturbations} (iii) yields $\Gamma_{H_{k,T}}(\beta;y)\ge c_0T^{-s-1/2}$. For $j\neq k$, the same spacing argument gives that
\begin{equation}
T|\beta-\beta_j|
\ge
T|\beta_j-\beta_k|-T|\beta-\beta_k| \ge
\Delta|j-k|-\frac{\Delta}{2}  \ge
\Delta\left(|j-k|-\frac12\right).
\label{equ:appendix_tail_bound_1}
\end{equation}

Then by \Cref{lem:localized-perturbations} (ii), we obtain that
\begin{align*}
\sum_{\substack{j<m\\ j\neq k}}
\bigl|\Gamma_{H_{j,T}}(\beta;y)\bigr|
\leq \frac{c_0}{4}T^{-s-1/2},  \quad \quad
\sum_{j\ge m}
\bigl|\Gamma_{H_{j,T}}(\beta;y)\bigr|\leq
\frac{c_0}{4}T^{-s-1/2}.
\end{align*}
Combining these three bounds gives \eqref{equ:sign_function_property_1}, as claimed.
\end{proof}

We now turn to the actual errors in conditional coverage probability. By \Cref{lemma:common_perturbation} (ii), $\sup_{\beta\in\mathcal{J}} |\bar{D}_T(\beta;y)|\leq C\eta T^{-s-1/2}$. Choose $\eta$ so that $C\eta\leq \varepsilon c_0/8$. Then combining with \Cref{lemma_decomposition_D_mt} and \Cref{lemma:sign_function_property}, we have \eqref{equ:minimax_Dmt_left}-\eqref{equ:minimax_Dmt_right} for all sufficiently large $T$. Thus, the actual conditional coverage error has the same step-function sign pattern as its linearization, up to the transition band.

\subsubsection{Upper bounds on divergences}

Recall from Appendix \ref{subsubsection:minimax_coverage_step2} that $\operatorname{P}_{m}$ denotes the law of $\mathbf{Y}_{-i}=\{Y_j\}_{j\neq i}$ under $g_{m}$, and $(f_{Y})_{m}$ is the corresponding marginal density, and $F_G(\cdot \mid y)$ denotes the posterior distribution of $\theta_i$ given $Y_i=y$ under the mixing distribution $G$.

\begin{lemma}
Recall that $u_{m,m'}:=(f_Y)_{m}-(f_{Y})_{m'}=(g_{m}-g_{m'})*f_\varepsilon$, we have
\[
\int_{\mathbb R}\frac{u_{m,m'}(y)^2}{(f_Y)_0(y)}\,\mathrm{d}y
\le
C\varepsilon^2\nu^2 T^{-2s+2}\exp(-\sigma^2T^2).
\]
\label{lemma:step2_weighted_l2_bound}
\end{lemma}
\vspace{-2.5em}
\begin{proof}[Proof of \Cref{lemma:step2_weighted_l2_bound}]
First, as shown in Appendix~\ref{subsection:proof_rate_posterior_quantile}, we have $(f_Y)_0(y)\ge c(1+y^2)^{-1}.$ Therefore, by Parseval's identity,
\begin{align*}
\int_{\mathbb R}\frac{u_{m,m'}(y)^2}{(f_Y)_0(y)}\,\mathrm{d}y
&\le
C\int_{\mathbb R}(1+y^2)u_{m,m'}(y)^2\,dy \\[7pt]
&=
\frac{C}{2\pi}\int_{\mathbb R}\Big(|u_{m,m'}^\star(t)|^2+|\partial_t u_{m,m'}^\star(t)|^2\Big)\,dt.
\end{align*}

Fix $m<m'$. Then $g_{m}-g_{m'}=-2\varepsilon\sum_{j=m}^{m'-1}H_{j,T}$. Taking the Fourier transformation yields that
\begin{equation*}
\left(g_{m}-g_{m'}\right)^\star(t)
=
-2\varepsilon \nu T^{-s-1/2}(-i\,\mathrm{sgn}(t))k^\star(|t|/T)
\sum_{j=m}^{m'-1}\exp(it\theta_j).
\end{equation*}
By \Cref{lem:oscillatory-sums}, the oscillatory sum and its derivative are bounded uniformly on $t \in [T, 2T]$:
\begin{equation*}
\sup_{m,m'} \ \sup_{t\in[T,2T]} \
\left|
\sum_{j=m}^{m'-1}\exp(it\theta_j)
\right|
\le C, \quad \text{and} \quad
\sup_{m,m'} \ \sup_{t\in[T,2T]} \
\left|\partial_t\left(
\sum_{j=m}^{m'-1}\exp(it\theta_j)\right)
\right|
\le CT.
\end{equation*}
Recall that $u_{m,m'}^\star(t)=(g_{m}-g_{m'})^\star(t)\exp(-\sigma^2 t^2/2)$.

On the support of $k^\star(|t|/T)$, we have $|t| \asymp T$. Applying the oscillatory bounds gives:
\begin{align*}
\left|u_{m,m'}^\star(t)\right|
\le
C\varepsilon\nu T^{-s-1/2}|k^\star(|t|/T)|\exp(-\sigma^2 t^2/2), \\[12pt]
\left|\partial_t u_{m,m'}^\star(t)\right|
\le
C\varepsilon\nu T^{-s+1/2}|k^\star(|t|/T)|\exp(-\sigma^2 t^2/2).
\end{align*}
Because the support of $k^\star(|t|/T)$ has length $O(T)$, integrating these point-wise bounds yields:
\begin{align*}
\int_{\mathbb R}|u_{m,m'}^\star(t)|^2\,\mathrm{d}t &\le C\varepsilon^2\nu^2 T^{-2s}\exp(-\sigma^2T^2), \\[12pt]
\int_{\mathbb R}|\partial_t u_{m,m'}^\star(t)|^2\,\mathrm{d}t &\le C\varepsilon^2\nu^2 T^{-2s+2}\exp(-\sigma^2T^2).
\end{align*}
The derivative term clearly dominates. Substituting this upper bound back into our Plancherel inequality gives the desired claim.
\end{proof}

\begin{lemma}
For all sufficiently large $T$,
\begin{equation*}
\sup_m \ \sup_{x\in\mathbb R} \ \left| F_{G_{m}}(x \mid y)-F_{G_0}(x\mid y) \right|\leq CT^{-s-1/2}.
\end{equation*}
\label{lemma:step2_posterior_closedness}
\end{lemma}
\vspace{-2.5em}
\begin{proof}[Proof of \Cref{lemma:step2_posterior_closedness}]
We define the effect of a perturbation $H$ on this posterior distribution:
\begin{equation*}
\Lambda_{H}(x;y):=f_{G_0}^{-1}(y)\left[ \int \phi\left(\frac{y-\theta}{\sigma}\right)\Big{(} \mathbf{1}\left\{ \theta\leq x \right\}-F_{G_0}(x\mid y) \Big{)}H(\theta)\mathrm{d}\theta \right].
\end{equation*}
By part (i) of Lemma \ref{lem:localized-perturbations},
\begin{equation*}
\left|\Lambda_{H_{j,T}}(x;y)\right|
\le
C_0T^{-s-1/2}
\bigl(1+T|F_{G_0}(x\mid y)-\beta_j|\bigr)^{-2}.
\end{equation*}

Using similar arguments as in the proof of \Cref{lemma_decomposition_D_mt}, we obtain that
\begin{align*}
&\sup_m \ \sup_{x\in\mathbb R} \
\left|\Lambda_{\tilde{H}_{m,T}}(x;y)\right|
\leq \sup_m \ \sup_{x\in\mathbb R} \ \varepsilon\sum_{j=1}^M \left|  \Lambda_{H_{j,T}}(x;y) \right| 
\leq C T^{-s-1/2}
\end{align*}
where we used \Cref{lem:localized-perturbations} (i), \eqref{equ:appendix_tail_bound_2}, and \eqref{equ:appendix_tail_bound_1}.

The same Fourier arguments used in \Cref{lem:localized-perturbations} (i), applied to the two compactly supported bumps in $\bar{H}_T$, yields that $\sup_{x\in\mathbb{R}} |\Lambda_{\bar{H}_T}(x;y)|\leq C\eta T^{-s-1/2}$, and then
\begin{equation*}
\sup_{m} \ \sup_{x\in\mathbb{R}} \ \left| \Lambda_{\bar{H}_T+\tilde{H}_{m,T}}(x;y) \right|\leq CT^{-s-1/2}.
\end{equation*}

Now we apply the numerator-denominator expansion
\begin{equation*}
F_{G_0+H}\left(x \mid y\right)-F_{G_0}(x\mid y)=\Lambda_{H}(x;y)-\Lambda_H(x;y)\cfrac{(H*f_{\varepsilon})(y)}{f_{G_0+H}(y)}.
\end{equation*}
Since the denominator remains bounded away from zero for the perturbations, the previous expansion gives the desired result.
\end{proof}

\subsubsection{Oscillatory-sum and localization lemmas}

\begin{lemma}[Oscillatory exponential sums]
\label{lem:oscillatory-sums}
For the grid points $\{\theta_m\}_{m=1}^M$ defined in  Appendix \ref{subsection:proof_minimax_coverage}, there exists a constant $C$, independent of $m',m,t$, and $T$, such that for all sufficiently large $T$,
\begin{equation*}
\sup_{m,m'} \ \sup_{t\in[T,2T]} \
\left|
\sum_{j=m}^{m'-1}\exp(it\theta_j)
\right|
\le C, \quad \text{and} \quad
\sup_{m,m'} \ \sup_{t\in[T,2T]} \
\left|\partial_t\left(
\sum_{j=m}^{m'-1}\exp(it\theta_j)\right)
\right|
\le CT.
\end{equation*}
Consequently,
\begin{equation*}
|S_{m,T}(t)|=\left| \sum_{j<m}\exp(it\theta_j)-\sum_{j\geq m}\exp(it\theta_j) \right|\le C
\end{equation*}
uniformly in $m$ and in $t\in[T,2T]$.
\end{lemma}

\begin{proof}[Proof of \Cref{lem:oscillatory-sums}]
We write $q(\beta):=q_{G_0}(\beta;y)$ so that $\theta_j=q(\beta_j)$. Because the baseline density $g_0$ is Cauchy, the posterior distribution $F_{G_0}(\cdot\mid y)$ has a
continuous density that is strictly positive on the compact set $q_{G_0}(\mathcal{J};y)$. Consequently, the inverse
map $\beta\mapsto q(\beta)$ is twice continuously differentiable on $\mathcal{J}$. Thus, there exist constants $0<c_-\le c_+<\infty$ and $C_2<\infty$ such that for all $\beta\in \mathcal{J}$
\begin{equation*}
c_-\le q'(\beta)\le c_+, \quad |q''(\beta)|\le C_2.
\end{equation*}

Let $\Delta_j(t):=t(\theta_{j+1}-\theta_j)$ for $j=1,\dots,M-1.$ The grid points are evenly spaced such that $\beta_{j+1}-\beta_j=\Delta T^{-1}$. By the mean-value theorem, $\theta_{j+1}-\theta_j
=
q(\beta_{j+1})-q(\beta_j)
=
q'(\tilde{\beta}_j)\Delta T^{-1}$ for some $\tilde\beta_j\in(\beta_j,\beta_{j+1})\subset \mathcal{J}$. Therefore, for $t\in[T,2T]$, the phase difference is bounded by $\Delta c_-\le \Delta_j(t)\le 2\Delta c_+$. Then choose $\Delta>0$ small enough that $2\Delta c_{+}<\pi$, and set $\delta:=\min\{\Delta c_{-},\pi-2\Delta c_{+}\}>0$. Then we obtain that $0<\delta\le \Delta_j(t)\le \pi-\delta$ for all $j\le M-1, t\in[T,2T]$.

Also, by the mean-value theorem for second differences, $\theta_{j+1}-2\theta_j+\theta_{j-1}
=
q(\beta_{j+1})-2q(\beta_j)+q(\beta_{j-1})
=
q''(\hat\beta_j)\Delta^2T^{-2}$ for some $\hat\beta_j\in(\beta_{j-1},\beta_{j+1})\subset \mathcal{J}$. Hence $|\Delta_j(t)-\Delta_{j-1}(t)|
\le CtT^{-2}
\le CT^{-1}$ for $j=2,\dots,M-1,\ t\in[T,2T].$

We now prove the partial-sum bound. Fix $1\le m\le m' \le M+1$. If $m=m'$, the sum is zero, so assume $m<m'$. If $m'=M+1$, we can split the sum as $\sum_{j=m}^{m'-1} \exp(it\theta_j)
=
\sum_{j=m}^{M-1} \exp\left(it\theta_j\right)+\exp\left(it\theta_M\right)$. Because the single tail term $\exp(it\theta_M)$ is trivially bounded by 1, it suffices to treat the case $m<m'\le M$. For such $m'$, define the multiplier $b_j(t):=1/\left[1-\exp(i\Delta_j(t))\right]$ for $j=m,\dots,m'-1.$ Because $\Delta_j(t)\in[\delta,\pi-\delta]$, the denominator is bounded away from zero, implying $|b_j(t)|\le C$. Using $\exp(it\theta_j)=b_j(t)\left[\exp(it\theta_j)-\exp(it\theta_{j+1})\right]$, summation by parts yields
\[
\sum_{j=m}^{m'-1} \exp(it\theta_j)
=
b_{m}(t)\exp\left(it\theta_{m}\right)
-
b_{m'-1}(t)\exp\left(it\theta_{m'}\right)
+
\sum_{j=m+1}^{m'-1} \bigl(b_j(t)-b_{j-1}(t)\bigr)\exp\left(it\theta_j\right).
\]
Since $u\mapsto (1-e^{iu})^{-1}$ is continuously differentiable on $[\delta,\pi-\delta]$, hence Lipschitz continuous on that interval. Therefore
\begin{equation*}
|b_j(t)-b_{j-1}(t)|
\le C|\Delta_j(t)-\Delta_{j-1}(t)|
\le CT^{-1}.
\end{equation*}
Since $M\lesssim T$, the sum of the absolute differences is bounded: $\sum_{j=m+1}^{m'-1} |b_j(t)-b_{j-1}(t)|\le C.$ Applying the triangle inequality to the summation-by-parts expansion gives that $\left|\sum_{j=m}^{m'-1} \exp\left(it\theta_j\right)\right|\le C$ uniformly in $m,m',t$, and $T$.

For the derivative bound, note that $|\theta_j|=|q(\beta_j)|\leq \sup_{\beta\in\mathcal{J}} |q(\beta)|\leq C$ uniformly in $j$. Hence
\begin{equation*}
\left|\partial_t\sum_{j=m}^{m'-1} \exp\left(it\theta_j\right)\right|=\left| i\sum_{j=m}^{m'-1} \theta_j\exp(it\theta_j) \right|\leq C\sum_{j=m}^{m'-1}1\leq CM\leq CT.
\end{equation*}
The final claim for $S_{m,T}$ follows from the triangle inequality.
\end{proof}

\begin{lemma}[Localized perturbations]
\label{lem:localized-perturbations}
Assume the notation of Appendix \ref{subsection:proof_minimax_coverage}. Also recall that
\begin{equation*}
\Lambda_{H}(x;y):=f_{G_0}^{-1}(y)\left[ \int \phi\left(\frac{y-\theta}{\sigma}\right)\Big{(} \mathbf{1}\left\{ \theta\leq x \right\}-F_{G_0}(x\mid y) \Big{)}H(\theta)\mathrm{d}\theta \right].
\end{equation*}
Then there exist constants $c_0,C_0>0$, independent of $m$ and $T$, s.t. for all sufficiently large $T$:
\begin{itemize}
\item[(i)] for every $x\in\mathbb R$, $|\Lambda_{H_{m,T}}(x;y)|
\le
C_0T^{-s-1/2}
(1+T|F_{G_0}(x\mid y)-\beta_m|)^{-2};$
\item[(ii)] for every $\beta\in\mathcal{J}$, $|\Gamma_{H_{m,T}}(\beta;y)| \le C_0T^{-s-1/2} (1+T|\beta-\beta_m|)^{-2}$;
\item[(iii)] Let $B_m:=\left[\beta_m-\frac{\Delta}{2T},\,\beta_m+\frac{\Delta}{2T}\right]\cap \mathcal{J}$, then $\inf_{\beta\in B_m}\Gamma_{H_{m,T}}(\beta;y)\ge c_0T^{-s-1/2}$.
\end{itemize}
\end{lemma}
\begin{proof}[Proof of \Cref{lem:localized-perturbations}]
We choose a compact interval $\mathcal{J}^{+}\subset(0,1)$ such that $\mathcal{J}\cup(\mathcal{J}+1-\alpha)\subset\mathrm{int}(\mathcal{J}^{+})$, and the posterior quantile $q_{G_0}(\beta;y)$ is continuously differentiable on this $\mathcal{J}^{+}$. Therefore there exist constants
$0<c_-\le c_+<\infty$ such that, for all $\beta\in \mathcal{J}^{+}$ and all $m\le M$,
\begin{equation}
c_-|\beta-\beta_m|
\le
\bigl|q_{G_0}(\beta;y)-\theta_m\bigr|
\le
c_+|\beta-\beta_m|.
\label{equ:derivative_bound}
\end{equation}

\noindent\underline{Proof of part (i).}

We write that $\Lambda_{H_{m,T}}(x;y):=f_{G_0}^{-1}(y) \int \psi_x(\theta;y)H_{m,T}(\theta)\mathrm{d}\theta$ where
\[
\psi_x(\theta;y)
:=
\phi\!\left(\frac{y-\theta}{\sigma}\right)
\Bigl\{\mathbf{1}\{\theta\le x\}-F_{G_0}(x\mid y)\Bigr\}.
\]
As in \Cref{propo:fourier_decay}, integration by parts gives
\begin{equation}
\overline{\psi_x^\star(t;y)}
=
\phi\!\left(\frac{y-x}{\sigma}\right)\frac{\exp(-itx)}{it}
+
r_x(t;y),
\qquad
|r_x(t;y)|\le Ct^{-2},
\label{equ:integration_by_part}
\end{equation}
uniformly in $x$ over compact sets. Then by Parseval's identity and the specific form of the perturbation $H_{m,T}^{\star}(t)
=
\nu T^{-s-1/2}(-i\,\mathrm{sgn}(t))\exp(it\theta_m)k^\star(|t|/T)$, we decompose that $\Lambda_{H_{m,T}}(x;y)=\Lambda_{m,T,1}+\Lambda_{m,T,2}$ where
\begin{equation}
\Lambda_{m,T,1}=\frac{\nu T^{-s-1/2}}{2\pi f_{G_0}(y)}\phi\left(\frac{y-x}{\sigma}\right)\int \frac{\exp(-it x)}{it}\left(-i\mathrm{sgn}(t)\right)\exp(it\theta_m)k^{\star}(|t|/T)\mathrm{d}t
\label{equ:term_lambda_mt1}
\end{equation}
and
\begin{equation}
\Lambda_{m,T,2}=\frac{1}{2\pi f_{G_0}(y)}\int r_x(t;y)H_{m,T}^{\star}(t)\mathrm{d}t.
\label{equ:term_lambda_mt2}
\end{equation}

For the first term (\ref{equ:term_lambda_mt1}), since $(-i\mathrm{sgn}(t))/(it)=-1/|t|$, we have that
\begin{align*}
\Lambda_{m,T,1}&=-\frac{\nu T^{-s-1/2}}{2\pi f_{G_0}(y)}\phi\left(\frac{y-x}{\sigma}\right)\int \exp(it(\theta_m-x))\frac{k^{\star}(|t|/T)}{|t|}\mathrm{d}t \\[10pt]
&\overset{(a)}{=}-\frac{\nu T^{-s-1/2}}{2\pi f_{G_0}(y)}\phi\left(\frac{y-x}{\sigma}\right)\int_{|u|\in[1,2]} \exp(iTu(\theta_m-x))\frac{k^{\star}(|u|)}{|u|}\mathrm{d}u \\[10pt]
&\overset{(b)}{=}-\frac{\nu T^{-s-1/2}}{2\pi f_{G_0}(y)}\phi\left(\frac{y-x}{\sigma}\right)\int_1^2 \frac{k^{\star}(u)}{u}\Big{[}\exp(iuT(\theta_m-x))+\exp(-iuT(\theta_m-x)) \Big{]}\mathrm{d}u \\[10pt]
&=-\frac{\nu T^{-s-1/2}}{\pi f_{G_0}(y)}\phi\left(\frac{y-x}{\sigma}\right)\int_1^2 \frac{k^{\star}(u)}{u}\cos\left( uT(\theta_m-x) \right)\mathrm{d}u \\[10pt]
&=-\nu T^{-s-1/2}
\frac{\phi((y-x)/\sigma)}{f_{G_0}(y)}
\mathcal{K}\!\bigl(T(\theta_m-x)\bigr)
\end{align*}
where (a) is obtained by the change of variable $t=Tu$, (b) is obtain by splitting positive and negative $u$, and the kernel $\mathcal{K}$ in the last step is defined as
\[
\mathcal{K}(v)
:=
\frac{1}{\pi}\int_1^2 \frac{k^\star(u)}{u}\cos(uv)\mathrm{d}u.
\]

For the second term (\ref{equ:term_lambda_mt2}), we use that $|r_x(t;y)|\leq C|t|^{-2}$ and $|H_{m,T}^{\star}(t)|\leq CT^{-s-1/2}|k^{\star}(|t|/T)|$. And also $k^{\star}(|t|/T)$ is supported where $|t|\asymp T$, we have that $|\Lambda_{m,T,2}|\leq CT^{-s-1/2}\int_{|t|\asymp T} |t|^{-2}\mathrm{d}t=CT^{-s-1/2}T^{-1}=CT^{-s-3/2}$. Combining with the first term, we obtain that
\[
\Lambda_{H_{m,T}}(x;y)
=
-\nu T^{-s-1/2}
\frac{\phi((y-x)/\sigma)}{f_{G_0}(y)}
\mathcal{K}\!\bigl(T(\theta_m-x)\bigr)
+
O(T^{-s-3/2}).
\]
Since $u\mapsto k^\star(u)/u$ belongs to $C_c^\infty((1,2))$, we have $\mathcal{K}\in\mathcal S(\mathbb R)$, and hence $| \mathcal{K}(v)|\le C(1+|v|)^{-2}$ and
\begin{equation*}
\left|\Lambda_{H_{m,T}}(x;y)\right|
\lesssim
C T^{-s-1/2}
\left(1+T|x-\theta_m|\right)^{-2}
\end{equation*}
uniformly over $x\in q_{G_0}(\mathcal{J}^{+};y)$.

To bound $\Lambda_{H_{m,T}}(x;y)$ uniformly over $x\in\mathbb{R}$, we consider two cases. If $F_{G_0}(x\mid y)\in \mathcal{J}^+$, then $x=q_{G_0}(\beta;y)$ for some $\beta\in \mathcal{J}^+$. By \Cref{equ:derivative_bound}, we have that
\[
\bigl|\Lambda_{H_{m,T}}(x;y)\bigr|
\lesssim CT^{-s-1/2}\Bigl(1+T\bigl|F_{G_0}(x\mid y)-\beta_m\bigr|\Bigr)^{-2}.
\]

If $F_{G_0}(x\mid y)\notin\mathcal J^{+}$, the asymptotic expansion (\ref{equ:integration_by_part}) does not apply, so we need a different argument. First observe that there exists $d_0$ such that $|F_{G_0}(x\mid y)-\beta_m|\geq d_0$ for all $m$. Because $q_{G_0}(\cdot;y)$ is continuous and strictly
increasing, we have that there exists $d_1>0$ such that $|x-\theta_m|=|q_{G_0}(F_{G_0}(x\mid y);y)-q_{G_0}(\beta_m;y)|\ge d_1>0$ uniformly in $m$.

Now define $U_m=\{ |\theta-\theta_m|\leq d_1/2 \}$, the function $\theta\mapsto \psi_x(\theta;y)$ is $C^2$ on a fixed neighborhood of $\theta_m$, with derivatives bounded uniformly in $x$ and $m$. We let $P_{x,m}(\theta):=\psi_x(\theta_m;y)+\psi_x'(\theta_m;y)(\theta-\theta_m)$ be the first-order Taylor polynomial of $\psi_{x}(\cdot;y)$ at $\theta_m$.

Also, given that $H_{m,T}^{\star}(t)
=
\nu T^{-s-1/2}(-i\,\mathrm{sgn}(t))\exp(it\theta_m)k^\star(|t|/T)$, we have
\begin{align*}
\int H_{m,T}(\theta)\mathrm{d}\theta&=H_{m,T}^{\star}(0)=0, \\[7pt]
\int(\theta-\theta_m)H_{m,T}(\theta)\mathrm{d}\theta&=-i\partial_t\left( \exp(-it\theta_m)H_{m,T}^{\star}(t) \right)\Big{|}_{t=0}=0.
\end{align*}
Hence we obtain that
\begin{align*}
\Lambda_{H_{m,T}}(x;y)&=f_{G_0}^{-1}(y) \int \left(\psi_x(\theta;y)-P_{x,m}(\theta)\right)H_{m,T}(\theta)\mathrm{d}\theta \\[7pt]
&=f_{G_0}^{-1}(y) \left( \int_{U_m}+\int_{U_m^c}  \right)\Big{(}\psi_x(\theta;y)-P_{x,m}(\theta)\Big{)}H_{m,T}(\theta)\mathrm{d}\theta
\end{align*}
On $U_m$, Taylor's theorem gives that $|\psi_x(\theta;y)-P_{x,m}(\theta)|\leq C|\theta-\theta_m|^2$ uniformly in $x,m$. Therefore
\begin{align*}
\left| \int_{U_m}\Big{(}\psi_x(\theta;y)-P_{x,m}(\theta)\Big{)}H_{m,T}(\theta)\mathrm{d}\theta \right|&\leq C\int_{U_m} |\theta-\theta_m|^2|H_{m,T}(\theta)|\mathrm{d}\theta \\[7pt]
&\overset{(a)}{\leq}CT^{-s-5/2}\int |u|^2|h(u)|\mathrm{d}u\overset{(b)}{\leq} CT^{-s-5/2}
\end{align*}
where (a) uses the changes of variable $u=T(\theta-\theta_m)$ and (b) is due to $h\in\mathcal{S}(\mathbb{R})$.

On $U_m^c$, since $|\psi_x(\theta;y)-P_{x,m}(\theta)|\leq C(1+|\theta-\theta_m|)$, we have that
\begin{align*}
\left| \int_{U_m^c}\Big{(}\psi_x(\theta;y)-P_{x,m}(\theta)\Big{)}H_{m,T}(\theta)\mathrm{d}\theta \right| &\leq C\int_{U_m^c} (1+|\theta-\theta_m|)|H_{m,T}(\theta)|\mathrm{d}\theta \\[7pt]
&\overset{(a)}{=}\nu T^{-s-1/2}\int_{|u|>Td_1/2}\left(1+\frac{|u|}{T}\right)|h(u)|\mathrm{d}u \\[7pt]
&\overset{(b)}{\leq} \nu T^{-s-1/2}CT^{-3}=CT^{-s-7/2}
\end{align*}
where (a) uses again $u=T(\theta-\theta_m)$ and (b) uses $|h(u)|\leq C(1+|u|)^{-5}$ as $h\in\mathcal{S}(\mathbb{R})$.

 Combining these two bounds gives that $|\Lambda_{H_{m,T}}(x;y)|\leq CT^{-s-5/2}$ whenever $F_{G_0}(x\mid y)\notin\mathcal J^{+}$. Since $|F_{G_0}(x\mid y)-\beta_m|\geq d_0$, it follows that $T^{-s-5/2}\leq CT^{-s-1/2}(1+T|F_{G_0}(x\mid y)-\beta_m|)^{-2}$ and thus again
\[
\bigl|\Lambda_{H_{m,T}}(x;y)\bigr|
\le CT^{-s-5/2}
\le CT^{-s-1/2}\Bigl(1+T\bigl|F_{G_0}(x\mid y)-\beta_m\bigr|\Bigr)^{-2}.
\]
which completes the proof of part (i).

\noindent\underline{Proof of part (ii).}

By the definition of $I(\beta;y)$, we can write
\begin{equation}
\Gamma_H(\beta;y)
=
\Lambda_H\big{(}q_{G_0}(\beta+1-\alpha;y);y\big{)}
-
\Lambda_H\big{(}q_{G_0}(\beta;y);y\big{)}.
\label{equation:gamma_h}
\end{equation}
Applying part (i) to the second term of \eqref{equation:gamma_h} gives
\[
\Big{|}\Lambda_{H_{m,T}}(q_{G_0}(\beta;y);y)\Big{|}
\le
CT^{-s-1/2}(1+T|\beta-\beta_m|)^{-2}.
\]
For the first term of \eqref{equation:gamma_h}, we also apply part (i) and obtain that
\begin{align*}
\Big{|}\Lambda_{H_{m,T}}(q_{G_0}(\beta+1-\alpha;y);y)\Big{|}\leq CT^{-s-1/2}\left( 1+T|(\beta+1-\alpha)-\beta_m| \right)^{-2}.
\end{align*}
Now with the assumption $\operatorname{dist}(\mathcal{J},\mathcal{J}+1-\alpha)>0$ and $\beta_m\in\mathcal{J}$, we have that there exists $d_0>0$ such that $|(\beta+1-\alpha)-\beta_m|\geq d_0$ for all $\beta\in\mathcal{J}$ and all $m\le M$. Therefore
\begin{align*}
\Big{|}\Lambda_{H_{m,T}}(q_{G_0}(\beta+1-\alpha;y);y)\Big{|}&\leq CT^{-s-1/2}(1+Td_0)^{-2}\\[7pt]
&\leq CT^{-s-5/2}
\le
CT^{-s-1/2}(1+T|\beta-\beta_m|)^{-2}
\end{align*}
where the final step follows because $\beta,\beta_m\in\mathcal{J}$ implies the uniform bound $|\beta-\beta_m|\leq C'$, which yields that $T^{-2}\leq C'(1+T|\beta-\beta_m|)^{-2}$.

Combining these bounds via the triangle inequality directly yields
\[
\Big{|}\Gamma_{H_{m,T}}(\beta;y)\Big{|}
\le
C_0T^{-s-1/2}(1+T|\beta-\beta_m|)^{-2},
\]
which proves part (ii).

\noindent\underline{Proof of part (iii).}

First note that
\begin{equation*}
\mathcal{K}(0)=\frac{1}{\pi}\int_1^2 \frac{k^\star(u)}{u}\,du>0.
\end{equation*}
By continuity of $ \mathcal{K}$, there exists $\delta_0>0$ such that $ \mathcal{K}(v)\ge \frac12 \mathcal{K}(0)>0$ whenever $|v|\le \delta_0$.

Choose $\Delta>0$ so small that $c_+\Delta/2\le \delta_0$. If $\beta\in B_m$, then \eqref{equ:derivative_bound} gives
\[
\bigl|q_{G_0}(\beta;y)-\theta_m\bigr|
\le
c_+|\beta-\beta_m|
\le
\frac{c_+\Delta}{2T}
\le
\frac{\delta_0}{T},
\]
this ensures that the argument of $\mathcal{K}$ satisfies $T|\theta_m-q_{G_0}(\beta;y)|\leq \delta_0$, and thus $ \mathcal{K}\!(T(\theta_m-q_{G_0}(\beta;y)))\ge \frac12 \mathcal{K}(0)$. Because $\beta\mapsto f_{G_0}^{-1}(y)\phi((y-q_{G_0}(\beta;y))/\sigma)$ is continuous and strictly positive on $\mathcal{J}$, we obtain the uniform lower bound $-\Lambda_{H_{m,T}}(q_{G_0}(\beta;y);y)\ge c\,T^{-s-1/2}$ for all $\beta\in B_m$. On the other hand, as shown in part (ii), $q_{G_0}(\beta+1-\alpha;y)$ remains uniformly separated from $\theta_m$, so
\[
\bigl|\Lambda_{H_{m,T}}(q_{G_0}(\beta+1-\alpha;y);y)\bigr|
\le
CT^{-s-1/2}T^{-2}.
\]
Therefore, for all sufficiently large $T$, the dominant term determines the lower bound:
\[
\Gamma_{H_{m,T}}(\beta;y)
=
\Lambda_{H_{m,T}}(q_{G_0}(\beta+1-\alpha;y);y)
-
\Lambda_{H_{m,T}}(q_{G_0}(\beta;y);y)
\ge
c_0T^{-s-1/2}
\]
for all $\beta\in B_m$. This proves part (iii).
\end{proof}

\end{appendices}

\end{document}